\RequirePackage{fix-cm}

\documentclass[table]{article}

\usepackage[utf8]{inputenc}
\usepackage{authblk}
\usepackage{hyperref}
\usepackage{arxiv}

\usepackage[ruled, vlined]{algorithm2e}

\usepackage[symbol]{footmisc}

\usepackage{tikz}
\usetikzlibrary{automata, arrows}
\usetikzlibrary{decorations.pathreplacing,angles,quotes} 
\usetikzlibrary{positioning}
\usetikzlibrary{shapes,arrows}
\usetikzlibrary{decorations.shapes}
\usetikzlibrary{arrows.meta}
\usetikzlibrary{backgrounds, fit}
\usetikzlibrary{calc}
\tikzstyle{decision} = [diamond, draw, fill=blue!20, 
text width=4.5em, text badly centered, node distance=3cm, inner sep=0pt]
\tikzstyle{block} = [rectangle, draw, fill=blue!20, 
text width=5em, text centered, rounded corners, minimum height=4em]
\tikzstyle{line} = [draw, -latex']
\tikzstyle{cloud} = [draw, ellipse,fill=red!20, node distance=3cm,
minimum height=2em]

\DeclareRobustCommand\mysolidline{\begin{tikzpicture}[baseline=-0.5ex]\draw[line width=2] (0,0) -- (0.7,0);\node[] at (0,0) {}; \node[] at (0.75,0) {};\end{tikzpicture}}
\DeclareRobustCommand\mydashedline{\begin{tikzpicture}[baseline=-0.5ex]\draw[densely dashed, line width=2] (0,0) -- (0.7,0); \node[] at (0,0) {}; \node[] at (0.75,0) {};\end{tikzpicture}}
\DeclareRobustCommand\mydottedline{\begin{tikzpicture}[baseline=-0.5ex]\draw[densely dotted, line width=2] (0,0) -- (0.7,0); \node[] at (0,0) {}; \node[] at (0.75,0) {}; \end{tikzpicture}}

\usepackage{xcolor}
\usepackage{soul}

\usepackage{graphicx}
\usepackage{subcaption}
\usepackage{scrextend}

\graphicspath{{./img/}} 

\usepackage{mathtools}	
\usepackage{xcolor}
\usepackage[normalem]{ulem}

\usepackage{amsthm}
\usepackage{amsmath}
\usepackage{float}

\theoremstyle{definition}
\newtheorem{definition}{Definition}
\newtheorem{theorem}{Theorem}
\newtheorem{corollary}{Corollary}[theorem]
\newtheorem{lemma}{Lemma}[theorem]

\begin{document}
	
\title{Mapping NP-Hard Problems to Restricted Adiabatic Quantum Architectures
}

\author[1, $\S$]{Gary J. Mooney}
\author[1, $\S$]{Sam U.Y. Tonetto}
\author[1]{Charles D. Hill}
\author[1, $\dag$]{Lloyd C.L. Hollenberg}

\affil[1]{School of Physics \protect\\
          University of Melbourne \protect\\
          Parkville 3020, AUSTRALIA}

\date{Received: date / Accepted: date}

\maketitle
\begin{abstract}
We introduce a framework for mapping NP-Hard problems to adiabatic quantum computing~(AQC) architectures that are heavily restricted in both connectivity and dynamic range of couplings, for which minor-embedding -- the standard problem mapping method -- cannot be directly applied. Separating the mapping into two distinct stages, we introduce problem-specific reductions for both quadratic unconstrained binary optimisation (QUBO) and satisfiability (SAT) and develop the subdivision-embedding method that is suitable for directly embedding onto these heavily restricted architectures. The theory underpinning this framework provides tools to aid in the manipulation of Ising Hamiltonians for the purposes of Ising energy minimisation, and could be used to assist in developing and optimising further problem mapping techniques. For each of the problem mapping methods presented, we examine how the physical qubit count scales with problem size on architectures of varying connectivity.

\end{abstract}

\footnotetext[2]{lloydch@unimelb.edu.au}

\footnotetext[4]{Contributed equally to this work as first authors.}

\section{Introduction}
\label{Intro}
  Quantum computer hardware is rapidly developing, and while much focus is on the standard circuit model of quantum computing~\cite{IBM, Google, Rigetti, Intel, IonQ, Xanadu}, adiabatic quantum computing (AQC), first proposed by Farhi \textit{et al.}~\cite{farhi2000quantum}, is another promising near-term alternative~\cite{DWaveSys}. It has been shown to be polynomially equivalent~\cite{aharonov2008adiabatic, mizel2007simple} to the circuit model and is believed to be more robust against control errors and certain kinds of decoherence~\cite{tiersch2007non, ashhab2006decoherence, roland2005noise, sarandy2005adiabatic, childs2001robustness}. In contrast to circuit model quantum computing, AQC utilises continuous-time Hamiltonian evolution and the adiabatic theorem of quantum mechanics. It can leverage quantum tunnelling to help explore complex energy landscapes, making it well-suited to solving a wide variety of optimisation and sampling problems in a diverse range of fields~\cite{lucas2014ising}. Examples of such fields include machine learning~\cite{pudenz2013quantum,adachi2015application}, logistics~\cite{neukart2017traffic}, biochemistry~\cite{perdomo2008construction}, and quantum chemistry~\cite{babbush2014adiabatic,streif2019solving}. Applications are commonly modelled in the form of NP-Hard quadratic unconstrained binary optimisation (QUBO) problems~\cite{kochenberger2004unified} since they are at the centre of a framework unifying a broad variety of combinatorial optimisation problems including satisfiability (SAT)~\cite{farhi2000quantum}, number partitioning~\cite{grass2016quantum}, max clique~\cite{chapuis2019finding} and max cut~\cite{hamerly2019experimental}. Most experimental investigations of~AQC have thus far been conducted on devices from \textit{D-Wave Systems}~\cite{DWaveSys}, a Canadian company which recently announced a~5000-qubit programmable quantum annealer~\cite{boothby2019next}. Although quantum speedup on their devices has not been conclusively demonstrated, they have successfully exhibited quantum behaviour such as multi-qubit tunnelling~\cite{boixo2016computational}. As the development of AQC technology progresses, it has the potential to address large-scale real-world problems in both academia and industry.
  
  The typical method for solving an optimisation problem using AQC is to first map it onto a physical implementation of an Ising spin glass with programmable coupling strengths in the form of an Ising Hamiltonian, such that the ground state of the Hamiltonian corresponds to the solution to the original problem. Problem mapping onto the architecture is normally performed using minor-embedding~\cite{choi2008minor,choi2011minor}. Qualitatively, in minor-embedding logical spins are `stretched out' to form clusters of physical spins (where the spin representation is used to describe qubits) which are coupled tightly enough to act as single logical spins for the purposes of AQC. Thus, the range of allowable coupling strengths and degree of connectivity in a physical architecture limits the variety of problems that can be directly mapped using minor-embedding, and in more heavily constrained architectures can entirely prevent minor-embedding from being directly applied.
  
  A range of approaches targeted at alleviating architectural limitations of dynamic range have been developed within the AQC literature. Quantum annealing correction has been shown to effectively increase dynamic range by allowing a quantum device to more robustly operate at lower energy scales in the presence of thermal noise~\cite{pudenz2015quantum, pudenz2014error}. Several techniques based on engineering the quantum adiabatic algorithm have been shown to improve control errors~\cite{king2014algorithm}. A method that finds representations of Ising energies with large classical energy gaps between ground and first excited states has also been shown to reduce control errors~\cite{bian2014discrete}. Alternatively, for more restricted architectures, an AQC architecture based on a triangular lattice with settable local magnetic fields and $\{0, \pm J\}$ coupling strengths has been shown to be capable of embedding max-independent-set problems~\cite{kaminsky2004scalable}.
  
  Here we develop an alternative AQC problem mapping method to minor-embedding called subdivision-embedding, which is tailored towards heavily restricted architectures. In contrast to minor-embedding where logical \emph{spins} are `stretched out' so that they are represented by clusters of physical qubits, subdivision-embedding stretches out logical couplings to chains of physical qubits. In the context of graph theory, subdivision-embedding is a topological concept related to the graph homeomorphism problem~\cite{garey2002computers}, which was shown to be solvable in polynomial time for fixed-size undirected input graphs~\cite{robertson1995graph}. 
  Subdivision-embedding on a physical architecture for AQC requires particular parameter settings, for which ferromagnetically-coupled ancilla spins have been proposed to propagate couplings~\cite{kaminsky2004scalable}. We formally combine and extend these ideas to a general subdivision-embedding method for AQC. We then introduce reduction techniques for QUBO and SAT problems and show how they can be used in conjunction with subdivision-embedding to map these problems to architectures with constraints as restrictive as follows: each spin couples to a maximum of three other spins, coupling strengths consist of only $\{0, \pm J\}$ (unit coupling strengths) and spins have no local magnetic fields. The reduction technique used for QUBO reduces the number of couplings needed for any particular spin by replacing the spin with a cluster of spins, which increases the total number of required spins. This technique is general and can be applied to problems non-specific to QUBO.
  
  Our results provide explicit methods for mapping arbitrary~QUBO or~SAT instances to AQC architectures with heavily restricted dynamic range, and is thus in principle applicable to a wider range of architectures than current methods. This increases the plausibility of such restricted architectures being useful for~AQC and could potentially develop into more efficient methods of embedding particular problems onto less constrained architectures as well. We provide a theoretical basis that could aid in the development of further architecture embedding techniques. For example, formalised concepts such as effective coupling strength and edge influence may help with analysing clusters of spins within Ising spin glasses relating to the ground state, and the complete degree reduction transformation can be used to reduce local connectivity of an Ising spin glass encoding a problem. Since qubits, in general, are difficult to implement and the qubit count is typically inversely correlated with the energy gap in the quantum adiabatic algorithm, we perform qubit scaling analyses for both~QUBO and~SAT problem embeddings. Additionally, to help inform architecture design, we compare the embedding efficiency on architectures of varying connectivity. For the special case of grid-like architectures with crossings, we demonstrate an efficient and deterministic subdivision-embedding for SAT problems which greatly reduces both time and space complexity overheads when compared to heuristic embedding algorithms.
  
  This paper is organised as follows. In Section~\ref{sec:preliminaries} we introduce a pipeline that is used to help break down the process of mapping a problem onto a physical architecture into two primary stages: the reduction stage and the embedding stage.
  The embedding stage, which utilises subdivision-embedding, is the same for both problem classes and is described in more detail in~Section~\ref{sec:subdivisionembedding}. The reduction stage is described for~QUBO problems in~Section~\ref{QUBO} and for~SAT problems in~Section~\ref{sec:SAT}. An analysis of how the mapping procedure for~QUBO and~SAT scales with respect to problem size is provided in Section~\ref{sec:scaling}. A proof for subdivision-embedding is provided in Appendix~\ref{sec:spin-chain-analysis}. A proof for the complete degree reduction used for~QUBO problems is provided in~Appendix~\ref{sec:proof-of-degree-reduction} and its complexity analysis is provided in~Appendix~\ref{sec:degree-reduction-complexity}.

\newpage
\section{Overview of mapping stages} \label{sec:preliminaries}

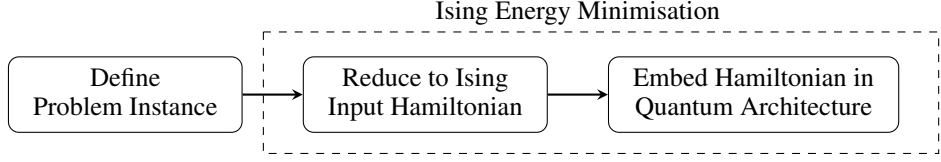
\begin{figure}[h!]
  \centering
  \begin{tikzpicture}

  \definecolor{tempcolor1}{RGB}{237,248,177}
  \definecolor{tempcolor2}{RGB}{127,205,187}
  \definecolor{tempcolor3}{RGB}{44,127,184}
  \definecolor{midcolors}{RGB}{49,130,189}

  \tikzstyle{arrow} = [thick,->,>=stealth]
  \tikzstyle{arrow2} = [dotted, thick, ->, >=stealth]
  \tikzstyle{process} = [rectangle, rounded corners, minimum width=2cm, minimum height=.5cm, text centered, draw=black, fill=midcolors!0]
  \tikzstyle{startstop} = [rectangle, rounded corners, minimum width=2cm, minimum height=.5cm,text centered, draw=black, fill=blue!0]
  \tikzstyle{sideprocess} = [rectangle, rounded corners, minimum width=2cm, minimum height=.5cm, text centered, draw = black, fill = orange!0]

  \node[sideprocess] (Problem) {\begin{tabular}{c} Define \\ Problem Instance \end{tabular}};
  \node[process, right=0.8cm of Problem] (Input) {\begin{tabular}{c} Reduce to Ising \\ Input Hamiltonian \end{tabular}};
  \node[startstop, right= 0.8cm of Input] (Embedded) {\begin{tabular}{c} Embed Hamiltonian in \\ Quantum Architecture \end{tabular}};

  \node[above left=0.2cm and 0.4cm of Input] (TopLeft) {};
  \node[below right=0.15cm and 0.4cm of Embedded] (BotRight) {};
  \draw[dashed] (TopLeft) rectangle (BotRight);

  \draw[arrow] (Problem) -- (Input);
  \draw[arrow] (Input) -- (Embedded) node[above=0.8cm, midway] {Ising Energy Minimisation};

\end{tikzpicture}
  \caption{A simplified pipeline for preparing a problem to be solved by AQC. A problem is reduced to a standardised Ising input Hamiltonian and then embedded in a quantum architecture by transforming the input Hamiltonian to an embedded Hamiltonian. In the last two stages, the problem has been reduced to the problem of Ising energy minimisation.\label{fig:SimplePipeline}}
\end{figure}

In order to solve a problem using AQC on a particular architecture, it can be pre-processed according to Figure~\ref{fig:SimplePipeline}. Broadly, this pre-processing can be separated into two stages: the reduction stage and the embedding stage. 

The {\em reduction stage} transforms a given problem (which can take various forms) into an Ising energy minimisation problem over logical Ising `spins', where the Ising Hamiltonian satisfies any constraints that may be required for an embedding algorithm to then map directly onto the architecture. We call an Ising Hamiltonian in this form the \emph{input Hamiltonian}. The Ising energy minimisation problem can be described as finding a spin configuration~$\boldsymbol{Z}:=(Z_1, \ldots Z_N)$ that minimises the cost function defined by an Ising Hamiltonian
\begin{align}
  \label{eq:IsingHamiltonian}
  H(\boldsymbol{Z}) = -\sum_{ij} J_{ij} Z_iZ_j - \sum_i h_i Z_i,
\end{align}
where $Z_i \in \{-1, +1\}$ is an eigenvalue of the Pauli $Z$ operator on spin $i$, $h_i$ are corresponding {\em local fields},~$J_{ij}$ are {\em coupling strengths}, and the first sum is over all pairs of spins~$ij$. With this sign convention,~$J_{ij} > 0$ corresponds to a ferromagnetic interaction. An Ising Hamiltonian~$H$, may be canonically represented by an Ising graph $G=(V, E)$, where each vertex represents a spin~$i\in V$ and is weighted by the local field~$h_i$, and each edge~$ij\in E$ represents a coupling and is weighted by the coupling strength~$J_{ij}$ (as shown in Figure~\ref{fig:HamilGraph}). We call the Ising graph representation for the input Hamiltonian the \emph{input graph} and for the programmable Ising Hamiltonian of the architecture the \emph{hardware graph}. For a particular vertex, the number of adjacent vertices is called its~\emph{degree} and a graph is called a \emph{cubic graph} when all vertices have degree three.

The {\em embedding stage} uses an embedding algorithm to map the input Hamiltonian directly onto a physical architecture, resulting in a modified input Hamiltonian called the \emph{embedded Hamiltonian}. This typically involves taking the logical spins in the input Hamiltonian and representing them with physical qubits in the embedded Hamiltonian, with coupling and local field strengths chosen such that ground states of the input Hamiltonian can be recovered from the ground states of the embedded Hamiltonian.

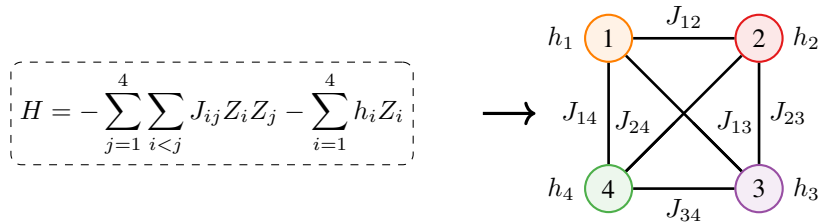
\begin{figure}[h!]
  \centering
  \begin{tikzpicture}

\definecolor{mygreen}{RGB}{77, 175, 74}
\definecolor{myyellow}{RGB}{255, 127, 0}
\definecolor{mypurple}{RGB}{152, 78, 163}
\definecolor{myred}{RGB}{228, 26, 28}
\definecolor{myblue}{RGB}{55, 126, 184}

  \node[rectangle, rounded corners, minimum width=2cm, minimum height=.5cm, text centered, draw=black, dashed] at (-1,0) {$\displaystyle H = -\sum_{j=1}^4\sum_{i < j}J_{ij}Z_iZ_j - \sum_{i=1}^4 h_iZ_i $};

  \draw[->, very thick] (2.6, 0) -- (3.3,0);

  \begin{scope}
    \pgftransformcm{1}{0}{0}{1}{\pgfpoint{150}{0}}
    \node[draw, circle, thick, label=left:$h_4$, draw=mygreen, fill=mygreen!10] (A) at (-1,-1) {4};
    \node[draw, circle, thick, label=left:$h_1$, draw=myyellow, fill=myyellow!10] (B) at (-1,1) {1};
    \node[draw, circle, thick, label=right:$h_3$, draw=mypurple, fill=mypurple!10] (C) at (1,-1) {3};
    \node[draw, circle, thick, label=right:$h_2$, draw=myred, fill=myred!10] (D) at (1,1) {2};
    \node[] (mid) at ($(A)!0.5!(D)$) {};
    \draw[line width=1] (A) -- (B) node[left, midway] {$J_{14}$};
    \draw[line width=1] (B) -- (C);
    \draw[line width=1] (mid) -- (C) node[midway, above right = 0.05cm and -0.15cm] {$J_{13}$};
    \draw[line width=1] (C) -- (D) node[right, midway] {$J_{23}$};
    \draw[line width=1] (D) -- (A);
    \draw[line width=1] (mid) -- (A) node[midway, above left = 0.05cm and -0.15cm] {$J_{24}$};
    \draw[line width=1] (A) -- (C) node[below, midway] {$J_{34}$};
    \draw[line width=1] (B) -- (D) node[above, midway] {$J_{12}$};
  \end{scope}

\end{tikzpicture}
  \caption{An Ising Hamiltonian (left) may be represented as an Ising graph (right), with weights $J_{ij}$ and $h_i$, and exhibiting the same connectivity. \label{fig:HamilGraph}}
\end{figure}
\section{Subdivision-embedding}\label{sec:subdivisionembedding}
The standard minor-embedding method of problem mapping cannot be directly applied to architectures with a small dynamic range of coupling strengths, since it assumes significantly larger strengths for pairs of qubits within logical spin clusters than pairs of qubits between the clusters. We develop subdivision-embedding as an alternative approach that requires minimal assumptions on the variability of coupling strengths. A flowchart summarising the difference between the mapping process for restricted and unrestricted architectures is shown in Figure~\ref{fig:pipeline}.

\begin{figure}[htb!]
  \begin{center}
    \pgfdeclarelayer{background}
\pgfdeclarelayer{foreground}
\pgfsetlayers{background,main,foreground}

\begin{tikzpicture}

  \definecolor{tempcolor1}{RGB}{237,248,177}
  \definecolor{tempcolor2}{RGB}{127,205,187}
  \definecolor{tempcolor3}{RGB}{44,127,184}
  \definecolor{midcolors}{RGB}{49,130,189}

  \tikzstyle{arrow} = [very thick,->,>=stealth]
  \tikzstyle{arrow2} = [dotted, thick, ->, >=stealth]
  \tikzstyle{process} = [rectangle, rounded corners, minimum width=2cm, minimum height=.5cm, text centered, draw=black, fill=midcolors!0]
  \tikzstyle{startstop} = [rectangle, rounded corners, minimum width=2cm, minimum height=.5cm,text centered, draw=black, fill=blue!0]
  \tikzstyle{sideprocess} = [rectangle, rounded corners, minimum width=2cm, minimum height=.5cm, text centered, draw=black, fill = orange!0]

  \node[sideprocess] (Problem) {\begin{tabular}{c} Define Problem \\ (e.g. SAT, QUBO) \end{tabular}};
  \node[process, below left=1cm and 0cm of Problem] (SubInput) {\begin{tabular}{c}Reduce to \\ subdivision-embeddable \\ Input $H_{\text{in}}$\end{tabular}};
  \node[process, below right=1cm and 0.3cm of Problem] (MinorInput) {\begin{tabular}{c}Reduce to \\ minor-embeddable \\ Input $H_{\text{in}}$\end{tabular}};
  \node[process, below=0.6cm of SubInput] (SubEmbed) {\begin{tabular}{c} Subdivision-embed to \\ get $H_{\text{emb}}$\end{tabular}};
  \node[process, below=0.6cm of MinorInput] (MinorEmbed) {\begin{tabular}{c} Minor-embed to \\ get $H_{\text{emb}}$\end{tabular}};

  \draw [arrow] (Problem) -| (SubInput) node[below=0.7cm, midway, fill=white, anchor=center, text=black] {Restricted Architecture};
  \draw [arrow] (Problem) -| (MinorInput) node[below=0.7cm, midway, fill=white, anchor=center, text=black] {Unrestricted Architecture};

  \draw [arrow] (SubInput) -- (SubEmbed);
  \draw [arrow] (MinorInput) -- (MinorEmbed);                                                          

\begin{pgfonlayer}{background}
  \node[above right=0.1cm and 0.3cm of MinorInput] (TopRight) {};
  \node[below left=0.1cm and 0.3cm of SubEmbed] (BotLeft) {};  
  \draw[dashed, fill=black!5] (BotLeft) rectangle (TopRight);
\end{pgfonlayer}

  \node (Stage1) at ($(SubInput)!0.52!(MinorInput)$) {\begin{tabular}{c} \textbf{Step 1:} \\ Reduction\end{tabular}};
  \node (Stage2) at ($(SubEmbed)!0.52!(MinorEmbed)$) {\begin{tabular}{c} \textbf{Step 2:} \\ Embedding\end{tabular}};
  \node[below=0.25cm of Stage2] (IEM) {Ising Energy Minimisation};

\end{tikzpicture}
    \caption{Detailed pipeline to prepare a problem for~AQC implementation on restricted and unrestricted architectures. A problem instance must be reduced to an Input Hamiltonian $H_{\mathrm{in}}$ before it is embedded as $H_{\mathrm{emb}}$ on a quantum processor. These Hamiltonians may be equivalently thought of as Ising Graphs $G_{\mathrm{in}}$ and $G_{\mathrm{emb}}$ respectively. The constraints of the quantum architecture dictate the choice of embedding method (subdivision-embedding on the left, minor-embedding on the right), which in turn determines the reduction procedure. \label{fig:pipeline}}
  \end{center}
\end{figure}
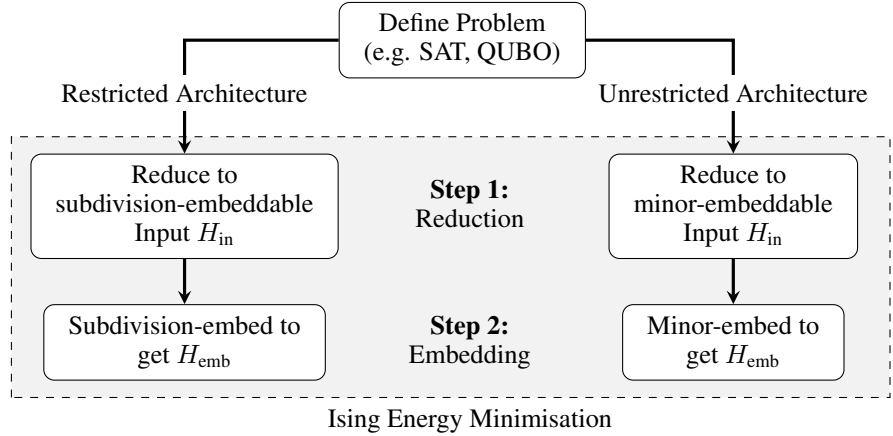

In order to precisely describe problem embedding, it can be useful to separate the topological aspects of embedding (the topological mapping) from the problem of assigning the appropriate coupling and local field strengths to an embedded Hamiltonian (the parameter-setting problem). For subdivision-embedding, the topological definition is as follows.
\begin{definition}[Subdivision-Embedding]
  A subdivision-embedding of a graph~$G$ in a graph~$F$ is a sequence of edge operations that each replace an edge $ik$ in $G$ with a vertex $j$ and two incident edges $ij$ and $jk$ such that the resulting graph, denoted by~$G_\mathrm{emb}$, is isomorphic to a subgraph of~$F$. We call~$G_\mathrm{emb}$ the embedded graph.
\end{definition}

An example of the difference between subdivision-embedding and minor-embedding at a topological level is shown in~Figure~\ref{fig:subminoremb}. In subdivision-embedding, the logical spins of the input Hamiltonian are mapped one-to-one onto physical qubits of the hardware graph, while in minor-embedding, they are mapped one-to-one onto highly-ferromagnetic clusters of qubits. The topology of a subdivision-embedding can always be used as a minor-embedding, however the converse is only true for input graphs with vertices of at most degree three. An important constraint for subdivision-embedding is that the degree of the input graph vertices must be smaller than or equal to the degree of the corresponding image vertices of the hardware graph. This degree constraint is addressed in the QUBO and SAT reduction stages introduced in Sections~\ref{QUBO} and~\ref{sec:SAT} respectively, which are designed to produce input graphs satisfying it.

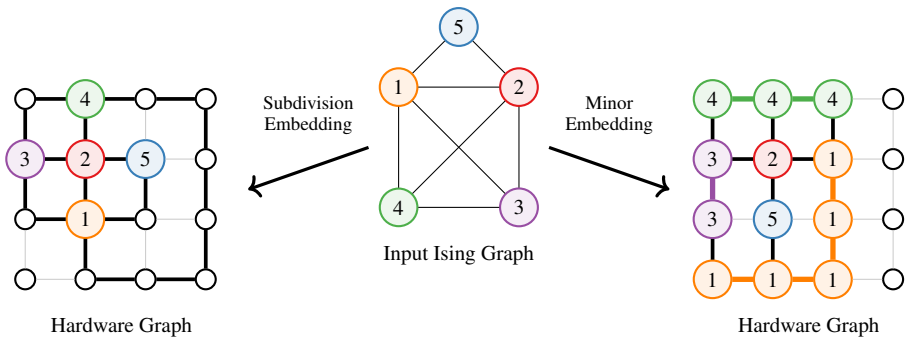
\begin{figure}[htb!]
  \centering
  \begin{tikzpicture}[scale=0.8, transform shape]

\definecolor{mygreen}{RGB}{77, 175, 74}
\definecolor{myyellow}{RGB}{255, 127, 0}
\definecolor{mypurple}{RGB}{152, 78, 163}
\definecolor{myred}{RGB}{228, 26, 28}
\definecolor{myblue}{RGB}{55, 126, 184}

\begin{scope}
  \pgftransformcm{1}{0}{0}{1}{\pgfpoint{0}{20}}
  \node[draw, circle, thick, fill=mygreen!10, draw=mygreen] (A) at (-1,-1) {4};
  \node[draw, circle, thick, fill=myyellow!10, draw=myyellow] (B) at (-1,1) {1};
  \node[draw, circle, thick, fill=mypurple!10, draw=mypurple] (C) at (1,-1) {3};
  \node[draw, circle, thick, fill=myred!10, draw=myred] (D) at (1,1) {2};
  \node[draw, circle, thick, fill=myblue!10, draw=myblue] (E) at (0,2) {5};
  \draw (A) -- (B);
  \draw (B) -- (C);
  \draw (C) -- (D);
  \draw (D) -- (A);
  \draw (A) -- (C);
  \draw (B) -- (D);
  \draw (E) -- (B);
  \draw (E) -- (D);
  \node[] at (0, -1.8) {Input Ising Graph};
\end{scope}

  \begin{scope} 
    \pgftransformcm{1}{0}{0}{1}{\pgfpoint{120}{-42.5}}

    \node[draw, fill=myyellow!10, circle, thick, draw=myyellow] (00) at (0,0) {1};
    \node[draw, fill=myyellow!10, circle, thick, draw=myyellow] (10) at (1,0) {1};
    \node[draw, fill=myyellow!10, circle, thick, draw=myyellow] (20) at (2,0) {1};
    \node[draw, circle, thick] (30) at (3,0) {};
    \node[draw, fill=mypurple!10, circle, thick, draw=mypurple] (01) at (0,1) {3};
    \node[draw, fill=myblue!10, circle, thick, draw=myblue] (11) at (1,1) {5};
    \node[draw, fill=myyellow!10, circle, thick, draw=myyellow] (21) at (2,1) {1};
    \node[draw, circle, thick] (31) at (3,1) {};
    \node[draw, fill=mypurple!10, circle, thick, draw=mypurple] (02) at (0,2) {3};
    \node[draw, fill=myred!10, circle, thick, draw=myred] (12) at (1,2) {2};
    \node[draw, fill=myyellow!10, circle, thick, draw=myyellow] (22) at (2,2) {1};
    \node[draw, circle, thick] (32) at (3,2) {};
    \node[draw, fill=mygreen!10, circle, thick, draw=mygreen] (03) at (0,3) {4};
    \node[draw, fill=mygreen!10, circle, thick, draw=mygreen] (13) at (1,3) {4};
    \node[draw, fill=mygreen!10, circle, thick, draw=mygreen] (23) at (2,3) {4};
    \node[draw, circle, thick] (33) at (3,3) {};

    \node[] at (1.6, -0.8) {Hardware Graph};

    \draw[line width=2pt, draw=mygreen] (03) -- (13) -- (23);
    \draw[line width=2pt, draw=mypurple] (02) -- (01);
    \draw[line width=2pt, draw=myyellow] (00) -- (10) -- (20) -- (21) -- (22);

    \draw[black, line width=1.4pt] (03) -- (02);
    \draw[black, line width=1.4pt] (23) -- (22);
    \draw[black, line width=1.4pt] (13) -- (12);
    \draw[black, line width=1.4pt] (11) -- (12);
    \draw[black, line width=1.4pt] (12) -- (22);
    \draw[black, line width=1.4pt] (02) -- (12);
    \draw[black, line width=1.4pt] (00) -- (01);
    \draw[black, line width=1.4pt] (10) -- (11);
    \draw[black!20] (11) -- (21);

    \draw[black!20] (01) -- (11);
    \draw[black!20] (23) -- (33);
    \draw[black!20] (22) -- (32);
    \draw[black!20] (21) -- (31);
    \draw[black!20] (20) -- (30);
    \draw[black!20] (30) -- (31);
    \draw[black!20] (31) -- (32);
    \draw[black!20] (32) -- (33);
  \end{scope}

  \begin{scope} 
    \pgftransformcm{1}{0}{0}{1}{\pgfpoint{-205}{-42.5}}

    \node[draw, circle, thick] (00) at (0,0) {};
    \node[draw, circle, thick] (10) at (1,0) {};
    \node[draw, circle, thick] (20) at (2,0) {};
    \node[draw, circle, thick] (30) at (3,0) {};
    \node[draw, circle, thick] (01) at (0,1) {};
    \node[draw, circle, thick, fill=myyellow!10, draw=myyellow] (11) at (1,1) {1};
    \node[draw, circle, thick] (21) at (2,1) {};
    \node[draw, circle, thick] (31) at (3,1) {};
    \node[draw, circle, fill=mypurple!10, thick, draw=mypurple] (02) at (0,2) {3};
    \node[draw, circle, fill=myred!10, thick, draw=myred] (12) at (1,2) {2};
    \node[draw, circle, fill=myblue!10, thick, draw=myblue] (22) at (2,2) {5};
    \node[draw, circle, thick] (32) at (3,2) {};
    \node[draw, circle, thick] (03) at (0,3) {};
    \node[draw, circle, fill=mygreen!10, thick, draw=mygreen] (13) at (1,3) {4};
    \node[draw, circle, thick] (23) at (2,3) {};
    \node[draw, circle, thick] (33) at (3,3) {};
    
    \node[] at (1.6, -0.8) {Hardware Graph};

    \draw[draw=black, line width=1.4pt] (02) -- (03) -- (13) -- (23) -- (33) -- (32) -- (31) -- (30) -- (20) -- (10) -- (11);
    \draw[draw=black, line width=1.4pt] (11) -- (21) -- (22) -- (12);
    \draw[draw=black, line width=1.4pt] (11) -- (12);
    \draw[draw=black, line width=1.4pt] (12) -- (02) -- (01) -- (11);
    \draw[draw=black, line width=1.4pt] (12) -- (13);

    \draw[black!20] (00) -- (01);
    \draw[black!20] (00) -- (10);
    \draw[black!20] (22) -- (23);
    \draw[black!20] (22) -- (32);
    \draw[black!20] (21) -- (31);
    \draw[black!20] (21) -- (20);

  \end{scope}
  
  \draw[->, very thick] (1.5,0.7) -- (3.5,0) node[above=0.4cm, midway] {\small \begin{tabular}{c}Minor \\ Embedding \end{tabular}};
  \draw[->, very thick] (-1.5,0.7) -- (-3.5,0) node[above=0.4cm, midway] {\small \begin{tabular}{c}Subdivision \\ Embedding\end{tabular}};
  
\end{tikzpicture}
  \caption{Subdivision-embedding (left) vs minor-embedding (right), of an Ising input graph (centre). In the subdivision-embedding, the logical spins are mapped one-to-one onto physical spins, while in minor-embedding, they are mapped to highly-ferromagnetic clusters. It is always possible for the topology of a subdivision-embedding to be used as a minor-embedding, but the converse is not always true unless the input graph vertices have at most degree three. Colour online.}
  \label{fig:subminoremb}
  \end{figure} 

Coupling and local field strength parameters need to be set in the embedded graph such that its ground state recovers the ground state of the input graph. In Definition~\ref{def:effective-coupling-strength} below we define the \textit{effective coupling strength}. We show in Appendix~A that transformations on spin chains (see Def.~\ref{def:spin-chain}) with no local fields that preserve the effective coupling strength ensure that the spin chain ground states of leaf (end) spins and the first-excited state energy penalty are preserved.

\begin{definition}[Effective Coupling Strength]\label{def:effective-coupling-strength}
The \emph{effective coupling strength} between leaf spins of a spin chain $Q$ with no local fields is defined as
\begin{equation}\label{eq:effective-coupling}
J_\mathrm{Q}^\text{eff} := (-1)^{p} \min_{ij}|J_{ij}|,
\end{equation} 
where $J_{ij}$ are coupling strengths of the edges along the spin chain, and $p$ is the number of these edges with negative coupling strengths.
\end{definition}

We further show the following four edge replacement operations (Corollaries \ref{theorem:chain-extension}-\ref{theorem:commuting-parity-edges}) preserve effective coupling strength and can potentially be used as steps in subdivision-embedding algorithms that include parameter-setting.

\begin{itemize}
	\item \textbf{Identity (Corollary~\ref{theorem:chain-extension}):} A single edge with coupling strength $J$ is replaced by a chain of two edges with coupling strengths $J$ and $J^\prime$ where $J^\prime \geq |J|$ and vice versa.
	\item \textbf{Inverse (Corollary~\ref{theorem:inverse-edges}):} A spin chain of two edges with negative coupling strengths~$J_1$ and~$J_2$ is replaced with a single edge with coupling strength $\min \{|J_1|, |J_2|\}$ and vice versa.
	\item \textbf{Commutativity (Corollary~\ref{theorem:commuting-edges}):} A spin chain of two edges is replaced with another spin chain of two edges where the coupling strengths for each edge have been swapped.
	\item \textbf{Parity Commutativity (Corollary~\ref{theorem:commuting-parity-edges}):} A spin chain of two edges is replaced with another spin chain of two edges where the sign of each edge has been swapped.
\end{itemize}
	
In particular, the identity operation provides a straightforward way for edges of the input graph to be subdivided for subdivision-embedding onto a hardware graph with uniform dynamic range of coupling strengths. Although in cases of more complicated coupling strength restrictions (e.g. some edges within the embedded graph can only have negative coupling strengths), a search algorithm that iteratively tries different available coupling strengths in the hardware graph with the above operations can help determine compatible coupling parameters.

\section{Reduction of the QUBO problem}\label{QUBO}

\begin{figure}[h!]
     \centering
     \begin{subfigure}[a]{0.45\textwidth}
         \centering
         \subcaption[short for lof]{Degree 4 to 3}\label{fig:device-4to3}
         \includegraphics[scale=0.25]{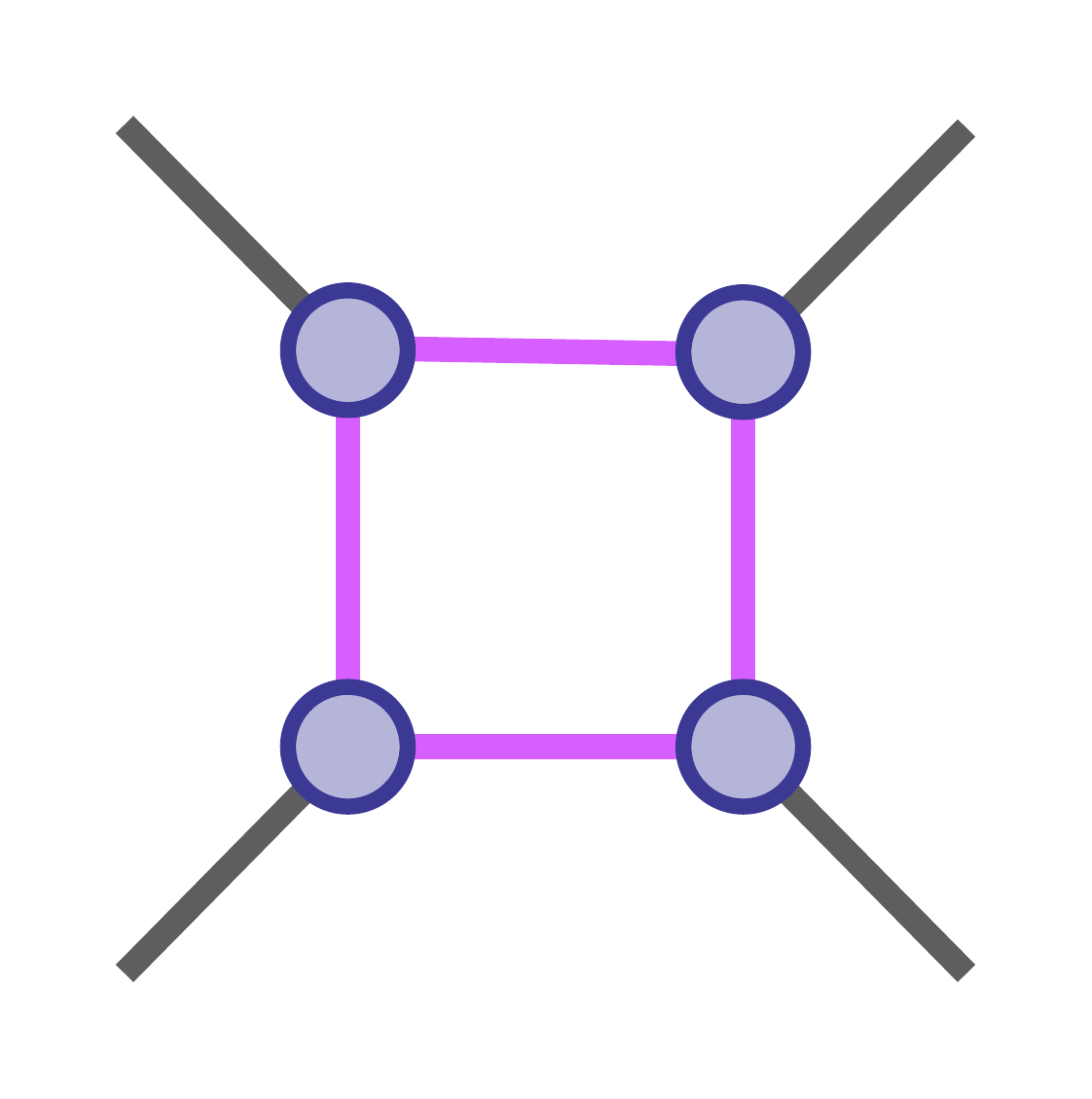}
     \end{subfigure}
     \hfill
     \begin{subfigure}[a]{0.45\textwidth}
         \centering
         \subcaption[short for lof]{Degree 5 to 4}\label{fig:device-5to4}
         \includegraphics[scale=0.25]{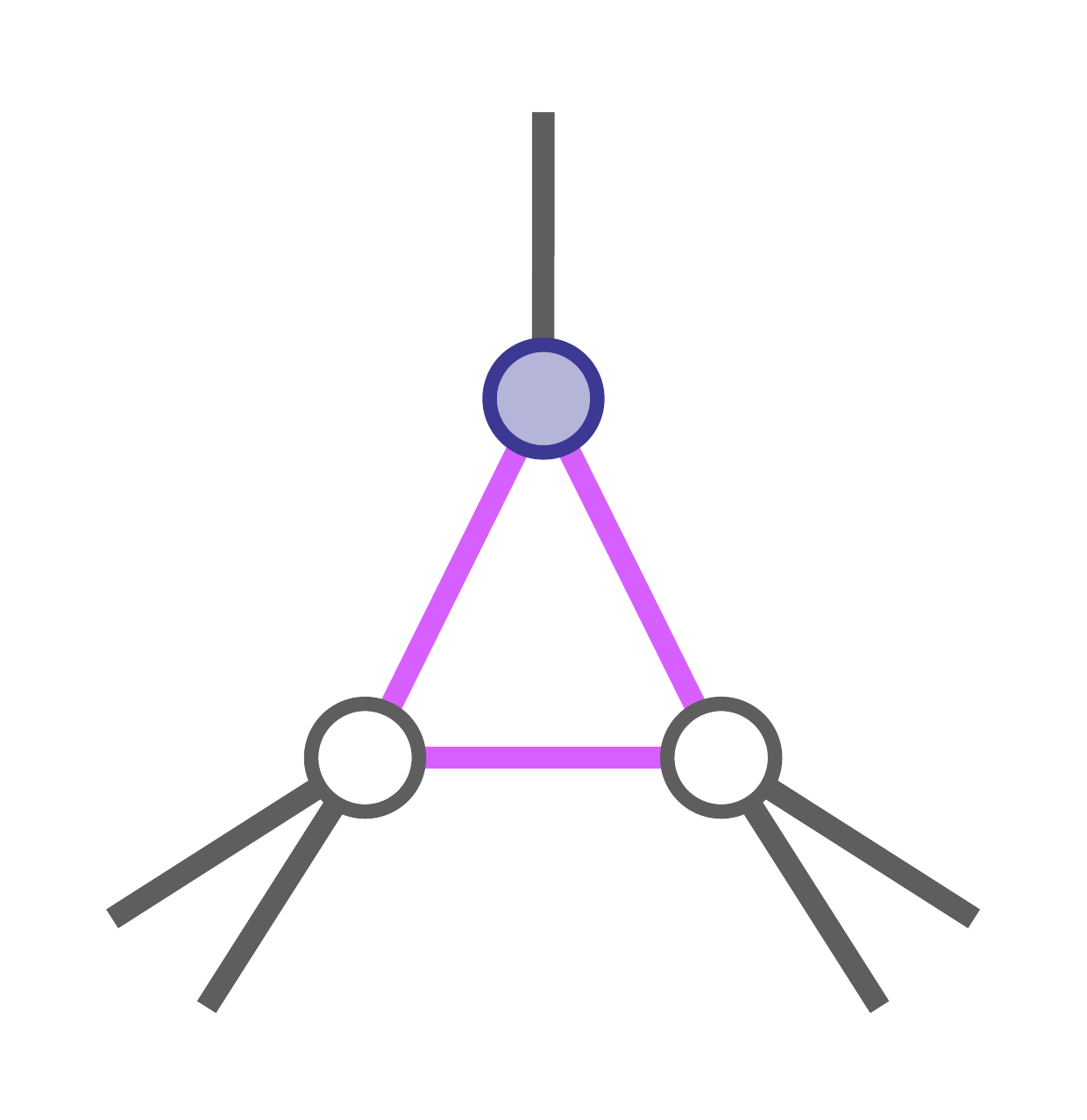}
     \end{subfigure}
        \caption{Two graphs produced after degree reducing a single vertex. These reductions were verified through simulation of graphs with unit coupling strengths. Shaded vertices represent spins that are logically equivalent to the initial spin and unshaded vertices represent ancilla spins. \textbf{(a)} Initial vertex of degree four transformed to a graph with vertices of degree three. 
  \textbf{(b)} Initial vertex of degree five transformed to a graph with vertices of degree three and four.}\label{fig:degree-reduction}
\end{figure}

\begin{figure}
  \centering
  \begin{subfigure}[a]{0.35\textwidth}
     \centering
     \subcaption[short for lof]{Before reduction}
     \includegraphics[scale=0.45]{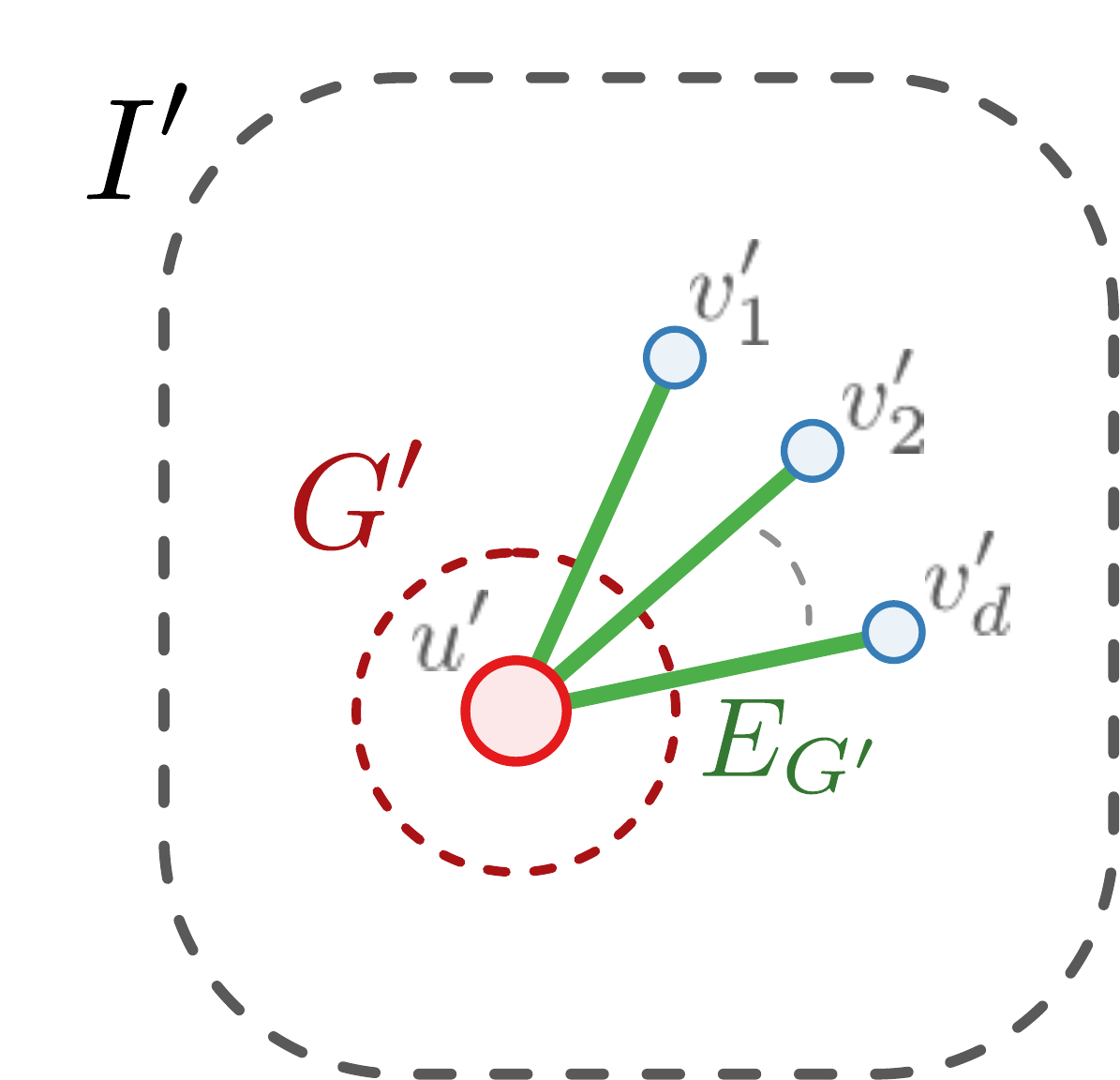}
  \end{subfigure}
  \hfill
  \begin{subfigure}[a]{0.60\textwidth}
     \centering
     \subcaption[short for lof]{After reduction}
     \includegraphics[scale=0.45]{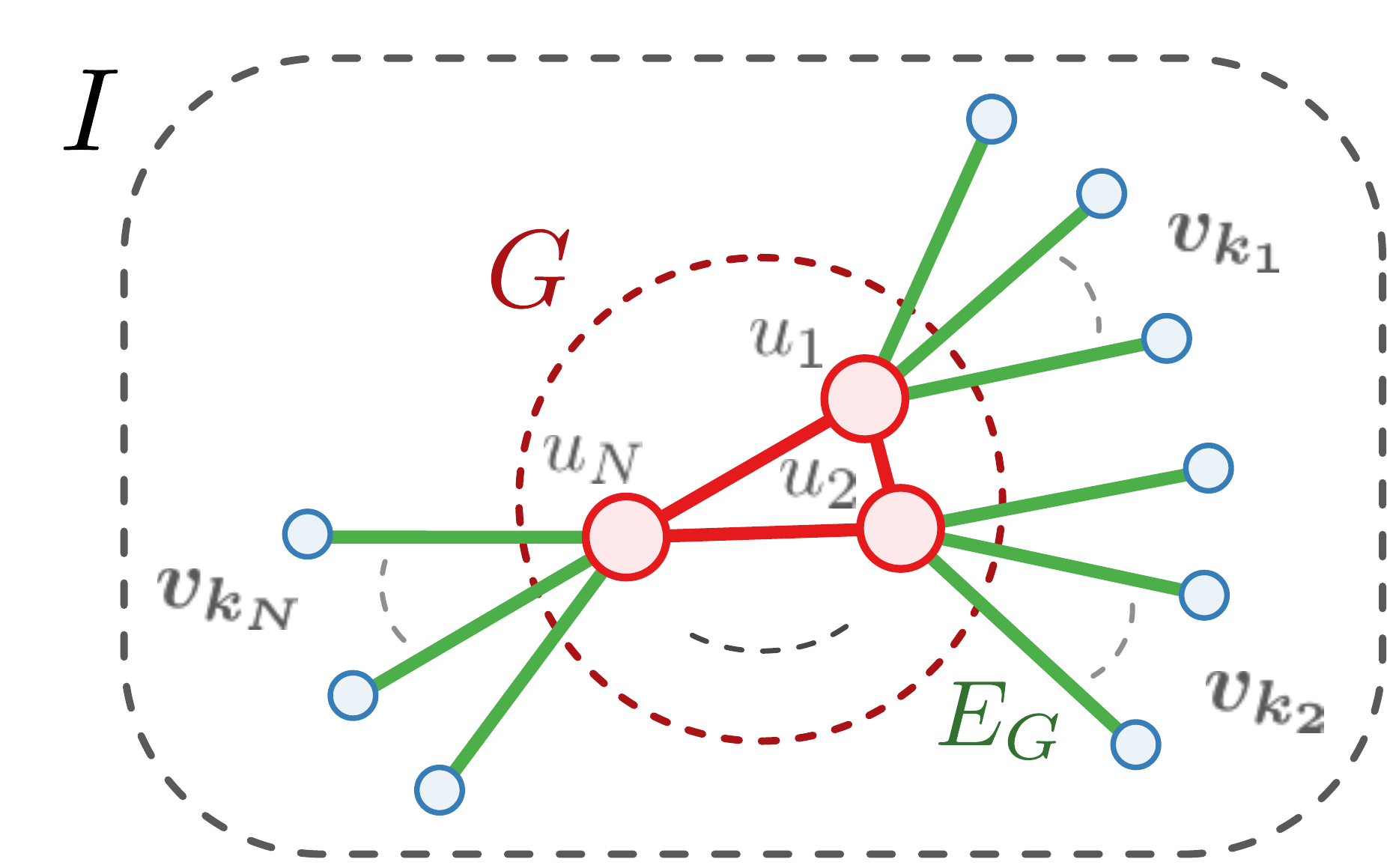}
  \end{subfigure}
  \hfill
		\caption{Graphics representing the complete degree reduction described in Definition~\ref{def:complete-degree-reduction} and Theorem~\ref{theorem:complete-degree-reduction}. 
        The two logical devices $(G^\prime, E_{G^\prime})$ and $(G, E_{G})$ (see Definition~\ref{def:logical-device}) shown in (a) and (b) respectively are equivalent according to Definition~\ref{def:logically-equivalent}.
        The single vertex graph $G^\prime$ of $I^\prime$ in (a) can effectively be replaced by the induced subgraph $G$ of $I$ in (b) (or vice versa) for the purposes of Ising energy minimisation. The Ising Hamiltonians for $I$ and $I^\prime$ are shown in Eq.~\ref{eq:degree-reduction-hamiltonians} and are assumed to have only 2-local interaction terms. \textbf{(a)}~A graph~$G^\prime$ consisting of a single 
        vertex~$u^\prime$ with $d:=|E_{G^\prime}|$ external edges. \textbf{(b)}~A complete graph $G$ of $N$ vertices where the number of external edges, $d = |E_{G}| =\sum_{i=1}^N \text{dim}(\boldsymbol{v_{k_i}})$, satisfies the inequality $d \leq N/2$.}
		\label{fig:complete-degree-reduction-hamiltonian}\par\medskip
\end{figure}\noindent

\begin{figure}
  \centering
  \includegraphics[width=0.75\textwidth]{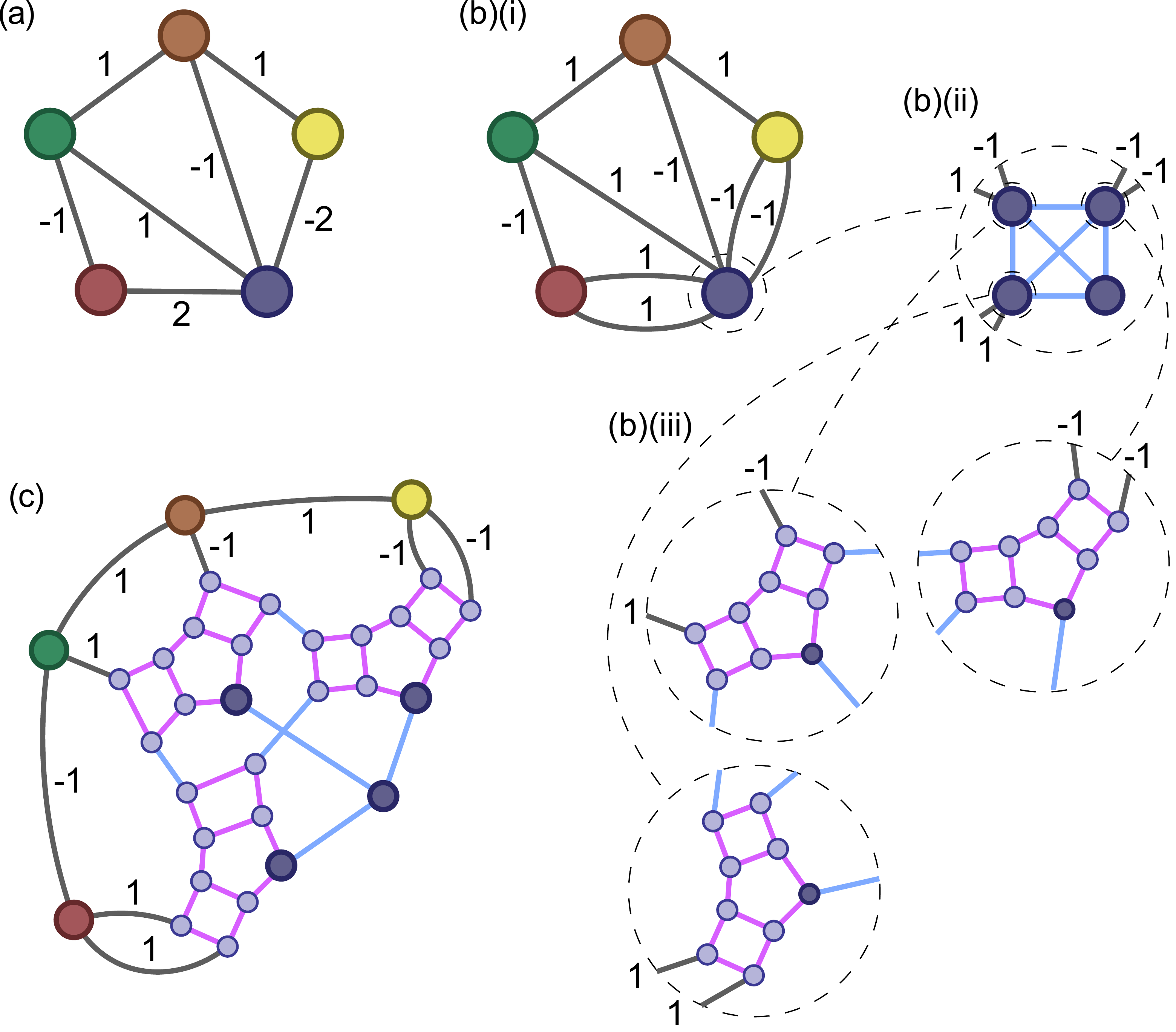}
		\caption{An example of the reduction stage for mapping a QUBO problem onto a cubic hardware graph with unit couplings. The numbers associated with edges represent coupling strengths and unnumbered edges represent~+1 coupling strengths. \textbf{(a)} A QUBO problem instance in the form 
        of an Ising graph. \textbf{(b)(i)} Edges with coupling strength other than $\pm 1$ are replaced by multiple edges with $\pm 1$ coupling strength 
        that sum to the original. \textbf{(b)(ii)} Complete degree reduction (described in Definition~\ref{def:complete-degree-reduction} and shown in Fig.~\ref{fig:complete-degree-reduction-hamiltonian}) is applied to the vertex with degree six which 
        lowers the highest degree of the graph to five. \textbf{(b)(iii)} The two degree reductions shown in Fig.~\ref{fig:degree-reduction} 
        are applied to further lower the highest degree of the graph to three. \textbf{(c)}~The reduced graph satisfies the
        constraint that all edges have unit coupling strength and the maximum degree of any vertex is three.}
  \label{fig:unit-degree-reduction-example}\par\medskip
\end{figure}\noindent

In order for a problem to be subdivision-embedded, it is required to first be in the form of an input graph such that the degree of its vertices are less than or equal to the degree of hardware graph vertices that they will be mapped to in the embedding. Additionally, the coupling strengths of the input graph are required to be sufficiently bounded to be expressible with the couplings in the hardware graph. Here we introduce a method for transforming an arbitrary QUBO instance to the form of a valid input graph for direct subdivision-embedding onto a cubic hardware graph with limited dynamic range. 

QUBO forms the basis of a framework that unifies a wide variety of combinatorial optimisation problems. It consists of problems in which a cost function
\begin{align}
  f(x_1, \ldots, x_n) = \sum_{i,j=1}^n c_{ij} x_i x_j + \sum_{i=1}^n c_i x_i,
\end{align}
is to be minimised over binary variables $(x_1, \ldots, x_n)\in \lbrace 0,1\rbrace^n$. Using the substitution $Z_i = 1-2x_i$, this cost function can be transformed into the form of an Ising Hamiltonian shown in Eq.~\ref{eq:IsingHamiltonian} called the problem Hamiltonian~$H_\text{p}$ with associated problem graph $G_\text{p}$, such that the solution to the QUBO problem is encoded in its ground state. To reduce the problem graph into the form of a valid input graph, we initially reduce the range of coupling strengths in the problem graph by distributing large couplings over multiple edges, which increases the degree of the vertices. The degree of the vertices are then reduced using a degree reduction procedure, which increases the number of vertices to compensate.

The dynamic range of edges in the problem graph can be reduced by replacing an edge possessing coupling strength~$J$ with multiple edges that are assigned smaller coupling strengths which sum to~$J$. If the coupling strengths do not divide evenly into the desired hardware graph coupling strengths, then they can be scaled before dividing. As an example, take a hardware graph with unit coupling strengths. If the problem graph has coupling strengths~$\{0.1, 0.23, -0.5\}$, then they can be scaled to~$\{10, 23, -50\}$. This allows the edges of the problem graph to be readily divided into multiple edges with unit coupling strengths. The number of edges can be traded off with accuracy by rounding to an appropriate precision. For example, approximating~$0.23 \approx 0.2$ allows the coupling strengths to be scaled to~$\{1, 2, -5\}$ which divides into fewer edges of unit coupling strength.

When the problem graph has the desired dynamic range of the hardware graph, the degree of its vertices can be reduced by applying a series of graph transformations. Initially, a \emph{complete degree reduction} transformation (defined below) can be recursively applied to reduce any vertex of finite degree to a graph of vertices with at most degree~five. Then two degree reduction transformations, shown in Fig.~\ref{fig:degree-reduction} (which were verified through simulation using unit-couplings) can be applied to further reduce the graph to vertices of at most degree~three. Complete degree reduction directly follows from Theorem~\ref{theorem:complete-degree-reduction} and is defined as follows.

\begin{definition}[Complete Degree Reduction -- Figure~\ref{fig:complete-degree-reduction-hamiltonian}]\label{def:complete-degree-reduction}
  A spin $u$ in an Ising lattice with couplings $J_{uv_i}$ to neighbouring spins is replaced with a complete graph $K_N$ of $N$ vertices with internal couplings satisfying $J_\text{int} \geq \max{|J_{uv_i}|}$, where each vertex of $K_N$ is adjacent to at most $N/2$ neighbours of $u$.
\end{definition}

The ground states of the Ising lattice will remain unchanged when the spin states of $K_N$ are identified with the spin state of $u$. For an Ising graph, this transformation replaces a vertex of high 
degree with a graph of vertices of lower degree such that the ground state of the original graph 
can be determined by the ground state of the transformed graph. The intuition is that the vertices that replace the original vertex are coupled strongly enough so that the spin states act as a single logical spin state. An example of the full reduction process for unit coupling strengths is shown in Figure~\ref{fig:unit-degree-reduction-example}. 
 
 The proof for complete degree reduction uses the assumption that the Ising Hamiltonian is in the form of a $ZZ$ Ising Hamiltonian which has only 2-local terms and is defined as 
\begin{equation}
H = -\sum \limits_{ij} J_{ij} Z_i Z_j.
\end{equation}
However, as shown in Fig.~\ref{fig:effective-local-fields}, an ancilla spin $z$ can be used to transform any Ising Hamiltonian with non-zero local fields to a~$ZZ$~Ising Hamiltonian. This can be achieved by coupling spins to the ancilla and setting the coupling strengths to the corresponding desired local field strengths. If the value of the ancilla is~$+1$ then the local fields are already effectively applied, if it is instead~$-1$ then the rest of the spins can be flipped to effectively apply the fields. This is due to the spin flip symmetry of the~$ZZ$~Ising Hamiltonian.

\begin{figure}
	\centering
		\includegraphics[width=0.50\textwidth]{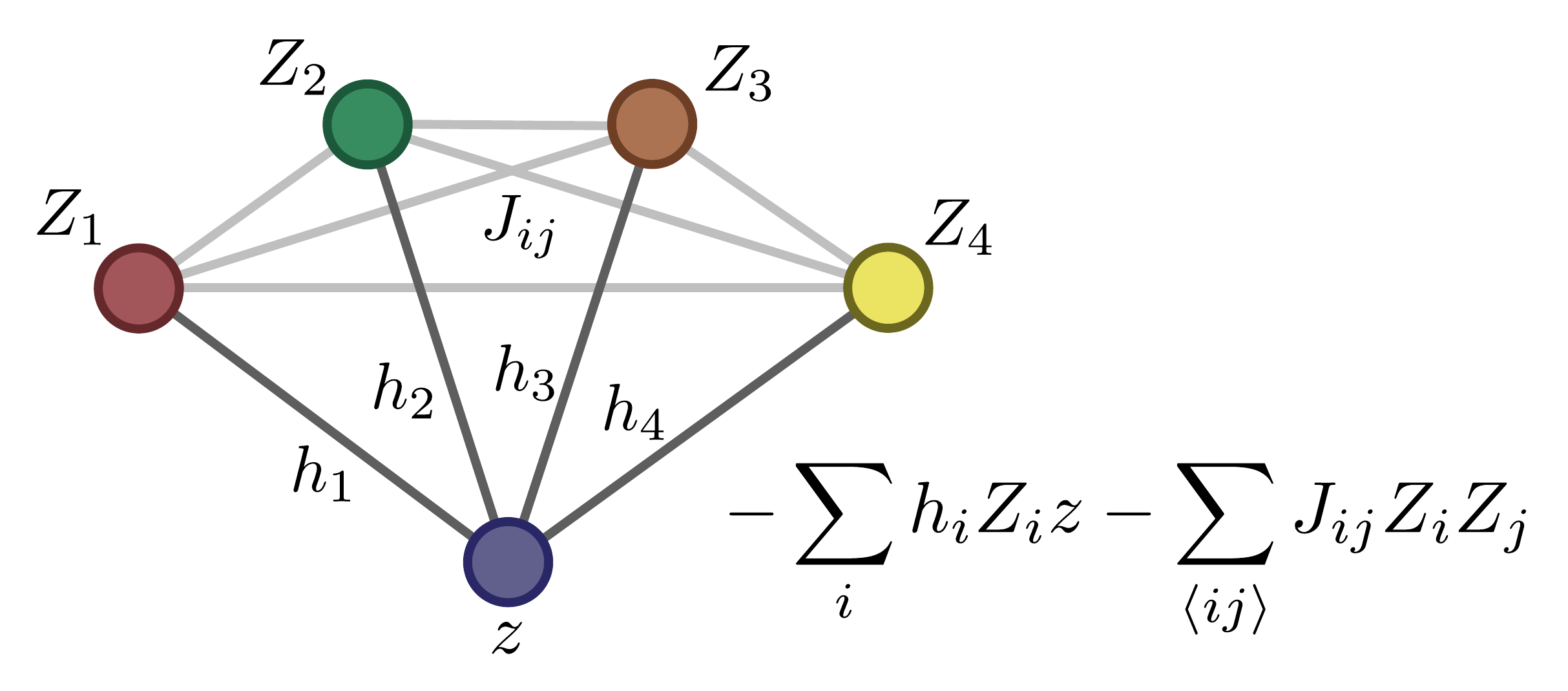}
		\caption{A graph showing how an effective local field can be produced from a~$ZZ$~Ising Hamiltonian (with only 2-local terms). If~$z$ is~$+1$, then the local fields are already effectively applied, whereas if~$z$ is~$-1$ then the other spins can be flipped to effectively apply the local fields, since the~$ZZ$ Ising Hamiltonian has spin flip symmetry.}
		\label{fig:effective-local-fields}\par\medskip
\end{figure}\noindent

\section{Reduction of the satisfiability problem}
\label{sec:SAT}

In this section, the reduction of {\em satisfiability} (SAT) to a subdivision-embeddable input graph satisfying the relevant degree and coupling strength requirements is demonstrated. SAT was the first problem shown to be NP-Complete by Cook and Levin~\cite{cook1971complexity} in 1971, and is a fundamental NP-Complete decision problem (with a yes/no answer) to which many other problems in NP may be reduced in a polynomial number of steps.

A typical SAT problem instance consists of a logical formula $\phi$ containing boolean variables $x_i$ (and their negations $\overline{x}_i$) linked by the logical connectives {\em and} ($\wedge$) and {\em or} ($\vee$). Each variable $x_i$ and its negation $\overline{x}_i$ are distinguished as separate {\em literals} corresponding to the same variable. We will assume that every SAT problem instance is expressed in conjunctive normal form (CNF) -- that is, a form where the variables are grouped into disjunctive clauses $(x_1 \vee x_2 \vee \dots \vee x_k)$, while the clauses themselves are conjunctively linked. An example of such a problem instance is
\begin{align}
  \label{eq:SAT}
  \phi &= (x_1 \vee x_2) \wedge (\overline{x}_1 \vee \overline{x}_2 \vee x_3) \wedge (\overline{x}_1 \vee x_4) \\ \nonumber
       &\ \ \ \wedge (x_2 \vee x_3) \wedge (x_1 \vee x_3 \vee x_4) \wedge (\overline{x}_2 \vee \overline{x}_3 \vee \overline{x}_4),
\end{align}	
which will be used in the remainder of this section as an aide to illustrate the reduction process. The solution to a SAT instance answers the following question. Does there exist a satisfying assignment $x_i \in \lbrace 0,1\rbrace$ for all $i$, such that $\phi$ evaluates to $1$ (i.e. {\em true})? If such an assignment exists, $\phi$ is said to be {\em satisfiable}, otherwise it is {\em unsatisfiable}. Since clauses are conjunctively linked, a SAT instance is satisfiable only when each clause can be simultaneously satisfied. It can be easily verified that the example in Eq.~\ref{eq:SAT} is satisfied by the assignment $(x_1,x_2,x_3,x_4) = (0,1,1,0)$.

A special case of SAT is $k$-SAT, which restricts each clause to have exactly~$k$ variables. One of the most widely studied variants is 3-SAT, which still remains NP-Complete. For large instances, the hardest 3-SAT instances have a clause-to-variable ratio of~$3.52 \leq r \leq 4.51$~\cite{mu2015empirical}. There has been considerable research into fast classical algorithms for solving 3-SAT, and in the worst-case clause-to-variable ratio, state-of-the-art algorithms scale as~$\mathcal{O}(1.008^n)$, where~$n$ is the number of variables. Another closely related variant on SAT, which will be used in our reduction, is NAE-SAT (not-all-equal SAT). This variant is similar to SAT, except that each clause returns true if and only if it contains literals that differ in logical values.

\subsection{Related work}

Due to its ubiquity in complexity theory, there have been several previous results on mapping variants of satisfiability to an adiabatic quantum computer. A straightforward 3-SAT mapping was described in Farhi's~\cite{farhi2000quantum} seminal paper on adiabatic quantum computing, requiring 3-local couplings. Another well-known 3-SAT reduction involves first reducing 3-SAT to the graph-theoretic problem of max independent set~\cite{choi2011different} and then to an Ising Hamiltonian with local fields. There is also some literature on the mapping of NAE-3-SAT to unit-coupling Ising models~\cite{king2014algorithm,Douglass2015}, as well as a comparison between local and global embedding methods for this mapping~\cite{Bian2016}. In this paper, we demonstrate a full reduction of a general SAT instance to a unit-coupling Input graph without local fields. This allows us to take arbitrarily complicated SAT instances and solve them on heavily restricted AQC architectures.

\subsection{Composition of constraint Hamiltonians}

In SAT, each clause imposes a constraint, and each constraint must be simultaneously satisfied if the entire SAT instance is to be satisfiable. This constraint can be mapped on to the problem of Ising energy minimisation (shown in Eq.~\ref{eq:IsingHamiltonian}) with the analogy that each clause corresponds to a constraint Hamiltonian which achieves its ground state if and only if the clause is satisfied. Before presenting the specific reduction, it is useful to first talk in generality about how large constraint Hamiltonians may be constructed from smaller ones.

Let $H_A$ and $H_B$ be Ising Hamiltonians acting on two sets of spins~$Z_A = \lbrace Z_{a_i}\rbrace$ and~$Z_B = \lbrace Z_{b_i}\rbrace$ respectively, with~$Z_{a_i},Z_{b_i} \in \lbrace -1,+1\rbrace$. We will also assume that these sets are not necessarily disjoint, so that there is a set of overlap spins $Z_C := \lbrace Z_{c_i}\rbrace = Z_A \cap Z_B$, on which {\em both}~$H_A$ and~$H_B$ act. Let us also define~$H$ to be the aggregate Hamiltonian,
\begin{align}
  H = H_A + H_B.
\end{align}
If $E_0^A$ and $E_0^B$ are the ground state energies of~$H_A$ and~$H_B$ respectively and $E$ is the Ising energy of $H$, then in general,~$E \geq E_0^A + E_0^B$, with equality only when~$H_A$ and~$H_B$ are simultaneously in their respective ground states. For this to be the case, there must exist a pair of ground states of~$H_A$ and~$H_B$ that agree on the set of overlapping spins~$Z_C$. As a simple illustration, consider the two spin Hamiltonians
\begin{align}
  H_A &= (Z_a + Z_{c_1} + Z_{c_2}) + (Z_aZ_{c_1}+Z_aZ_{c_2} + Z_{c_1}Z_{c_2}) \label{eq:HArevised} \\
  H_B &= - (Z_{b} + Z_{c_1} + Z_{c_2}) + (Z_bZ_{c_1} + Z_bZ_{c_2} + Z_{c_1}Z_{c_2}). \label{eq:HBrevised}
\end{align}
These systems can be described graphically as anti-ferromagnetic triangles as in Fig.~\ref{fig:separateRevised}, with spins biased downwards in~$H_A$ and biased upwards in~$H_B$. In the aggregate Hamiltonian~$H = H_A + H_B$, the local fields of the overlap spins~$\lbrace Z_{c_i}\rbrace$ have cancelled out to zero.
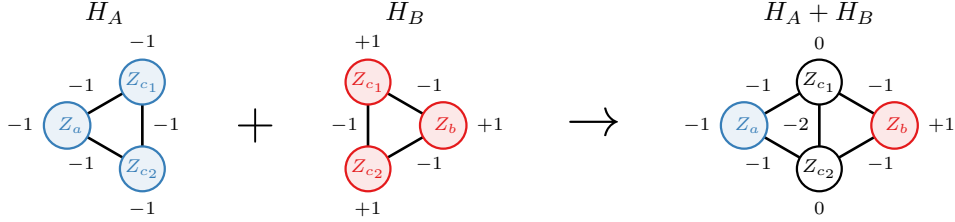
\begin{figure}
  \centering
  \begin{tikzpicture}

  \definecolor{myblue}{RGB}{55, 126, 184}
  \definecolor{myred}{RGB}{228, 26, 28}
  \definecolor{myredblue}{RGB}{142, 76, 106}
  
  \begin{scope}
  \pgftransformcm{1}{0}{0}{1}{\pgfpoint{-1cm}{0cm}};
  \draw[line width=1] (0,0) -- (1, -0.5774) node[midway, below left] {\scriptsize $-1$};
  \draw[line width=1] (1,-0.5774) -- (1, 0.5774) node[midway, right] {\scriptsize $-1$};
  \draw[line width=1] (1, 0.5774) -- (0,0) node[midway, above left] {\scriptsize $-1$};
  \node[draw, circle, color=myblue, fill=myblue!10, label=left:{\scriptsize $-1$}, text width=2.5mm, align=center, inner sep=3pt, line width=0.3mm] at (0,0) {\scriptsize $Z_a$};
  \node[draw, circle, color=myblue, fill=myblue!10, label=below:{\scriptsize $-1$}, text width=2.5mm, align=center, inner sep=4pt, minimum size=8pt, line width=0.3mm] at (1, -0.5774) {};
  \node[color=myblue] at (1, -0.5774) {\scriptsize $Z_{c_2}$};
  
  \node[draw, circle, color=myblue, fill=myblue!10, label=above:{\scriptsize $-1$}, text width=2.5mm, align=center, inner sep=4pt, minimum size=8pt, line width=0.3mm] at (1, 0.5774) {};
  \node[color=myblue] at (1, 0.5774) {\scriptsize $Z_{c_1}$};
  
  \node[] at (0.5,1.5) {$H_A$};
  \end{scope}

    \begin{scope}
    \pgftransformcm{2}{0}{0}{2}{\pgfpoint{1.5cm}{0cm}};
    
    \node[] at (0,0) {\huge{$+$}};
    \end{scope}

  \begin{scope}
  \pgftransformcm{1}{0}{0}{1}{\pgfpoint{3cm}{0cm}};
  \draw[line width=1] (1,0) -- (0, -0.5774) node[midway, below right] {\scriptsize $-1$};
  \draw[line width=1] (0,-0.5774) -- (0, 0.5774) node[midway, left] {\scriptsize $-1$};
  \draw[line width=1] (0, 0.5774) -- (1,0) node[midway, above right] {\scriptsize $-1$};
  \node[draw, circle, color=myred, fill=myred!10, label=right:{\scriptsize $+1$}, text width=2.5mm, align=center, inner sep=3pt, line width=0.3mm] at (1,0) {\scriptsize $Z_b$};
  \node[draw, circle, color=myred, fill=myred!10, label=below:{\scriptsize $+1$}, text width=2.5mm, align=center, inner sep=4pt, minimum size=8pt, line width=0.3mm] at (0, -0.5774) {};
  \node[color=myred] at (0, -0.5774) {\scriptsize $Z_{c_2}$};
  
  \node[draw, circle, color=myred, fill=myred!10, label=above:{\scriptsize $+1$}, text width=2.5mm, align=center, inner sep=4pt, minimum size=8pt, line width=0.3mm] at (0, 0.5774) {};
  \node[color=myred] at (0, 0.5774) {\scriptsize $Z_{c_1}$};
  \node[] at (0.5,1.5) {$H_B$};
  \end{scope}

    \begin{scope}
    \pgftransformcm{2}{0}{0}{2}{\pgfpoint{6cm}{0cm}};
    
    \node[] at (0,0) {\huge{$\to$}};
    \end{scope}
    
  \begin{scope}
    \pgftransformcm{1}{0}{0}{1}{\pgfpoint{9cm}{0cm}};
    \draw[line width=1] (0,-0.5774) -- (0, 0.5774) node[midway, left] {\scriptsize $-2$};
    \draw[line width=1] (1,0) -- (0, -0.5774) node[midway, below right] {\scriptsize $-1$};
    \draw[line width=1] (0, 0.5774) -- (1,0) node[midway, above right] {\scriptsize $-1$};
    \draw[line width=1] (-1,0) -- (0, -0.5774) node[midway, below left] {\scriptsize $-1$};
    \draw[line width=1] (0, 0.5774) -- (-1,0) node[midway, above left] {\scriptsize $-1$};
    \node[draw, circle, color=black, fill=white, label=below:{\scriptsize $0$}, text width=2.5mm, align=center, inner sep=4pt, minimum size=8pt, line width=0.3mm] at (0,-0.5774) {};
    
    \node[] at (0, -0.5774) {\scriptsize $Z_{c_2}$};
    
    \node[draw, circle, color=black, fill=white, label=above:{\scriptsize $0$}, text width=2.5mm, align=center, inner sep=4pt, minimum size=8pt, line width=0.3mm] at (0, 0.5774) {};
    
    \node[] at (0, 0.5774) {\scriptsize $Z_{c_1}$};
    
    \node[draw, circle, color=myred, fill=myred!10, label=right:{\scriptsize $+1$}, text width=2.5mm, align=center, inner sep=3pt, line width=0.3mm] at (1, 0) {\scriptsize $Z_b$};
    \node[draw, circle, color=myblue, fill=myblue!10, label=left:{\scriptsize $-1$}, text width=2.5mm, align=center, inner sep=3pt, line width=0.3mm] at (-1, 0) {\scriptsize $Z_a$};
    \node[] at (0,1.5) {$H_A+H_B$};
  \end{scope}

\end{tikzpicture}
  \caption{\label{fig:separateRevised} Ising graphs corresponding to the Hamiltonians in Eqs.~\ref{eq:HArevised}$-$\ref{eq:HBrevised}, with local field and coupling strengths shown beside their respective vertices and edges. Note that the states of spins $Z_{c_1}$ and $Z_{c_2}$ are common to both subsystems. This means that if there are pairs of ground states of $H_A$ and $H_B$ where the spin states of $Z_{c_1}$ and $Z_{c_2}$ agree, the ground states of $H_A+H_B$ will comprise those pairs, as shown in Figure~\ref{fig:compositionRevised}.}
\end{figure}

The ground states of $H_A$ and $H_B$ are shown on the left and right respectively in Figure~\ref{fig:compositionRevised}. In the aggregate system $H = H_A + H_B$, the ground states are composed of only the pairs of ground states of $H_A$ and $H_B$ which agree on the overlap spins $Z_C$. While this specific example contains ground state pairs which agree on the overlap spins, this will not always be the case. If there does not exist a pair of ground states of $H_A$ and $H_B$ that agree on the overlap, then $E > E_0^A + E_0^B$, which is to say that in the ground state of $H$, at least one of $H_A$ and $H_B$ must be in an excited state. In such a situation, if $H_A$ and $H_B$ were encoding constraints, then at least one of the constraints must have been violated. Conversely, if $E = E_0^A + E_0^B$, all constraints will have been satisfied. 

\begin{figure}
  \centering
  \input{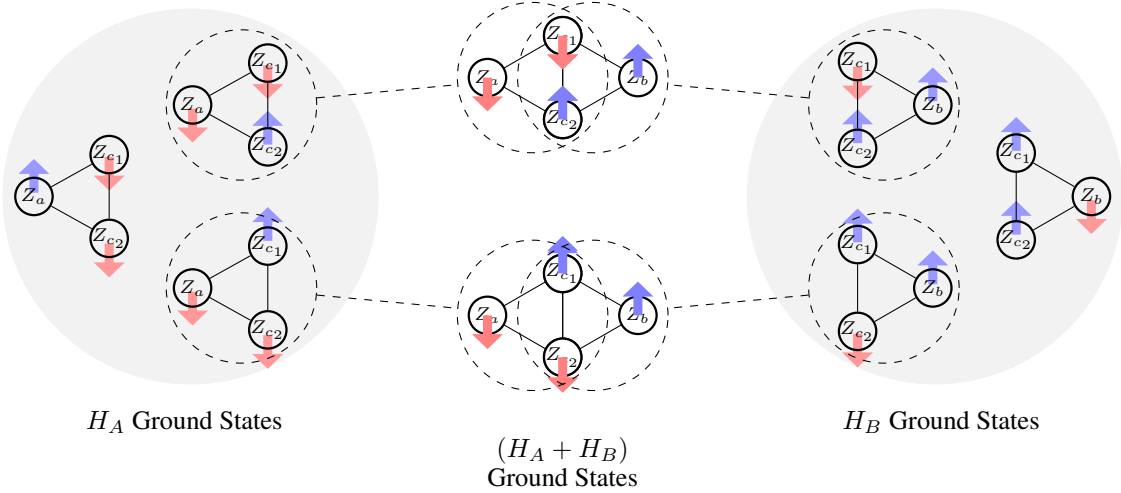}
  \caption{The three ground states of the Ising Hamiltonian $H_A$ (Eq.~\ref{eq:HArevised}) have two spins down ($Z=-1$) and one spin up ($Z=+1$), while the three ground states of the Ising Hamiltonian $H_B$ (Eq.~\ref{eq:HBrevised}) have two spins up ($Z=+1$) and one spin down ($Z=-1$). The ground states of the aggregate Hamiltonian $H_A + H_B$ are the overlaps of individual ground states of $H_A$ and $H_B$. \label{fig:compositionRevised}}
\end{figure}

\subsection{Reduction Step 1: SAT \texorpdfstring{$\to$}{to} NAE-SAT}

The first step of the reduction is to reduce SAT to (not-all-equal) NAE-SAT. The purpose of this reduction is to introduce negation-symmetry, a property that is present in Ising Hamiltonians in the form of spin-flip symmetry, but not present in SAT. 

A NAE clause is satisfied if not all literals in the clause have the same boolean value. For example, $\text{NAE}(1,1,1)$ is an unsatisfied NAE clause, while $\text{NAE}(1,0,1)$ and $\text{NAE}(0,1,0)$ are both satisfied. The reduction from SAT to NAE-SAT involves the introduction of a single global ancillary Boolean variable (denoted $y$) to each clause, i.e.
\begin{align}
  \label{eq:SAT2NAESAT}
  (x_1 \vee x_2 \vee \dots \vee x_k) \mapsto \text{NAE}(x_1, x_2,  \dots , x_k, y)
\end{align}
The original SAT instance $\phi$ and the resultant NAE-SAT instance $\phi'$ are logically equivalent in the sense that $\phi$ is satisfiable if and only if $\phi'$ is satisfiable. Thus, a SAT problem instance can be solved by solving its corresponding NAE-SAT instance. Applying the transformation in Eq.~\ref{eq:SAT2NAESAT} to Eq.~\ref{eq:SAT}, 
\begin{align}
  \phi \to \phi' &= \text{NAE}(x_1,x_2,y) \wedge \text{NAE}(\overline{x}_1, \overline{x}_2, x_3,y) \wedge \text{NAE}(\overline{x}_1 , x_4,y) \\ \nonumber
                 &\ \ \ \ \wedge \text{NAE}(x_2 ,x_3,y) \wedge \text{NAE}(x_1,x_3,x_4,y) \wedge \text{NAE}(\overline{x}_2,\overline{x}_3 ,\overline{x}_4,y).
\end{align}
Note that each SAT clause with $k$ variables becomes a NAE-SAT clause with $(k+1)$ variables.

\subsection{Reduction Step 2: Encoding NAE-SAT in Ising Constraint Hamiltonians}

Once SAT has been reduced to NAE-SAT, it is possible to map it to a unit-coupling Input Hamiltonian made up of constraint sub-Hamiltonians. We will also refer to these sub-Hamiltonian as gadgets. We require two types of gadgets: one for clauses and one for variables. The clause-gadgets are constructed such that their ground states correspond precisely to all the possible satisfying assignments a NAE clause can have. Their topology is represented by linked triangles as shown in Fig.~\ref{fig:clause}-\ref{fig:clause3}, where a $k$-variable NAE clause requires $(k-2)$ triangles. The Hamiltonian for a $k$-variable NAE clause can be written
\begin{align}
H_{\mathrm{clause}} = \sum_{i=1}^{k-2} (Z_{i,1}Z_{i,2} + Z_{i,2}Z_{i,3}+Z_{i,1}Z_{i,3}) + \sum_{i=1}^{k-3}Z_{i,3}Z_{i+1, 1}
\end{align}
where the spin variable $Z_{i,j} \in \lbrace -1,+1\rbrace$ corresponds to the $j$th spin in the $i$th linked triangle (see Figures~\ref{fig:clause}-\ref{fig:clause3}). The spins $Z_{1,1}, Z_{k-2,3}$, and $Z_{i,2}$ for $1 \leq i \leq k$ correspond to logical variables in the NAE clause, while the rest are ancilla variables.

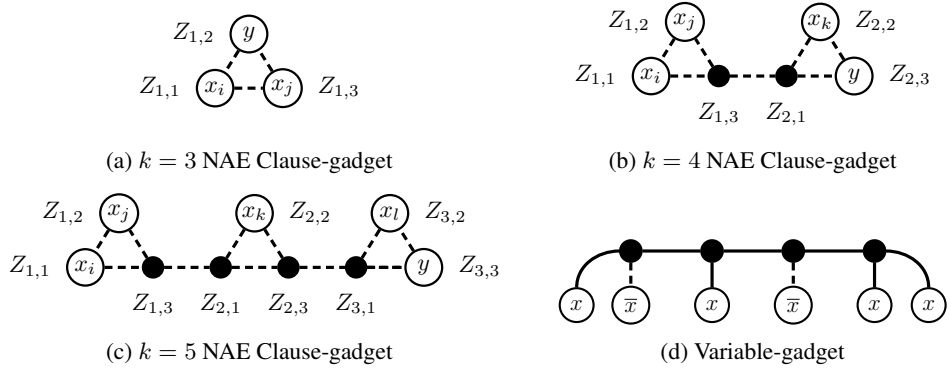
\begin{figure}
  \centering
  \subcaptionbox{$k=3$ NAE Clause-gadget\label{fig:clause}}[0.4\textwidth]{\begin{tikzpicture}[scale=0.9, transform shape]
	\node[draw, circle, text width=5mm, align=center, inner sep=0pt, line width=0.3mm] (a) at (0,0) {$x_i$} ;
	\node[draw, circle, text width=5mm, align=center, inner sep=0pt, line width=0.3mm] (b) at (1cm,0) {$x_j$};
	\node[draw, circle, text width = 5mm, align=center, inner sep=0pt, line width=0.3mm] (c) at ( 0.5cm,0.8cm) {$y$};
	\draw[densely dashed, line width=1.2] (a) -- (b);
	\draw[densely dashed, line width=1.2] (b) -- (c);
	\draw[densely dashed, line width=1.2] (c) -- (a);
	
	\node[left=0.1cm of a] {$Z_{1,1}$};
	\node[left=0.1cm of c] {$Z_{1,2}$};
	\node[right=0.1cm of b] {$Z_{1,3}$};
	
	\node[] at (0, -0.55) {};
	
\end{tikzpicture}}
  \subcaptionbox{$k=4$ NAE Clause-gadget\label{fig:clause2}}[0.4\textwidth]{\begin{tikzpicture}[scale=0.9, transform shape]
	\node[draw, circle, text width = 5mm, align=center, inner sep=0pt, line width=0.3mm] (a) at (0,0) {$x_i$} ;
	\node[fill=black, draw, circle, inner sep=3pt, line width=0.3mm] (b) at (1cm,0) {};
	\node[draw, circle, text width = 5mm, align=center, inner sep=0pt, line width=0.3mm] (c) at ( 0.5cm,0.8cm) {$x_j$};
	\draw[densely dashed, line width=1.2] (a) -- (b);
	\draw[densely dashed, line width=1.2] (b) -- (c);
	\draw[densely dashed, line width=1.2] (c) -- (a);

	\node[fill=black, draw, circle, inner sep=3pt, line width=0.3mm] (d) at (2,0) {} ;
	\node[draw, circle, text width = 5mm, align=center, inner sep=0pt, line width=0.3mm] (e) at (3cm,0) {$y$};
	\node[draw, circle, text width = 5mm, align=center, inner sep=0pt, line width=0.3mm] (f) at ( 2.5cm,0.8cm) {$x_k$};
	\draw[densely dashed, line width=1.2] (d) -- (e);
	\draw[densely dashed, line width=1.2] (e) -- (f);
	\draw[densely dashed, line width=1.2] (f) -- (d);
	
	\draw[densely dashed, line width=1.2] (b) -- (d);
	
	\node[left=0.1cm of a] {$Z_{1,1}$};
	\node[left=0.1cm of c] {$Z_{1,2}$};
	\node[below=0.1cm of b] {$Z_{1,3}$};
	
	\node[below=0.1cm of d] {$Z_{2,1}$};
	\node[right=0.1cm of e] {$Z_{2,3}$};
	\node[right=0.1cm of f] {$Z_{2,2}$};

	\node[] at (0, -0.25) {};
	
\end{tikzpicture}} \\ \vspace{0.2cm}
  \subcaptionbox{$k=5$ NAE Clause-gadget\label{fig:clause3}}[0.4\textwidth]{\begin{tikzpicture}[scale=0.9, transform shape]
	\node[draw, circle, text width = 5mm, align=center, inner sep=0pt, line width=0.3mm] (a) at (0,0) {$x_i$} ;
	\node[fill=black, draw, circle, inner sep=3pt, line width=0.3mm] (b) at (1cm,0) {};
	\node[draw, circle, text width = 5mm, align=center, inner sep=0pt, line width=0.3mm] (c) at ( 0.5cm,0.8cm) {$x_j$};
	\draw[densely dashed, line width=1.2] (a) -- (b);
	\draw[densely dashed, line width=1.2] (b) -- (c);
	\draw[densely dashed, line width=1.2] (c) -- (a);

	\node[fill=black, draw, circle, inner sep=3pt, line width=0.3mm] (d) at (2,0) {} ;
	\node[fill=black, draw, circle, inner sep=3pt, line width=0.3mm] (e) at (3cm,0) {};
	\node[draw, circle, text width = 5mm, align=center, inner sep=0pt, line width=0.3mm] (f) at ( 2.5cm,0.8cm) {$x_k$};
	\draw[densely dashed, line width=1.2] (d) -- (e);
	\draw[densely dashed, line width=1.2] (e) -- (f);
	\draw[densely dashed, line width=1.2] (f) -- (d);
	
	\node[fill=black, draw, circle, inner sep=3pt, line width=0.3mm] (g) at (4,0) {} ;
	\node[draw, circle, text width = 5mm, align=center, inner sep=0pt, line width=0.3mm] (h) at (5cm,0) {$y$};
	\node[draw, circle, text width = 5mm, align=center, inner sep=0pt, line width=0.3mm] (i) at ( 4.5cm,0.8cm) {$x_l$};
	\draw[densely dashed, line width=1.2] (g) -- (h);
	\draw[densely dashed, line width=1.2] (h) -- (i);
	\draw[densely dashed, line width=1.2] (i) -- (g);

	\draw[densely dashed, line width=1.2] (b) -- (d);
	\draw[densely dashed, line width=1.2] (e) -- (h);
	
	\node[left=0.1cm of a] {$Z_{1,1}$};
	\node[left=0.1cm of c] {$Z_{1,2}$};
	\node[below=0.1cm of b] {$Z_{1,3}$};
	
	\node[below=0.1cm of d] {$Z_{2,1}$};
	\node[below=0.1cm of e] {$Z_{2,3}$};
	\node[right=0.1cm of f] {$Z_{2,2}$};

	\node[below=0.1cm of g] {$Z_{3,1}$};
	\node[right=0.1cm of h] {$Z_{3,3}$};
	\node[right=0.1cm of i] {$Z_{3,2}$};
	
	\node[] at (0, -0.25) {};
	
\end{tikzpicture}}
  \subcaptionbox{Variable-gadget\label{fig:var}}[0.4\textwidth]{\begin{tikzpicture}[scale=1.2, transform shape]
\begin{scope}
    \pgftransformcm{0.6}{0}{0}{0.6}{\pgfpoint{0}{0}}
	\node[fill=black, draw, circle, thick, inner sep=4pt] (1) at (0,0) {};
	\node[fill=black, draw, circle, thick, inner sep=4pt] (2) at (1.5,0) {};
	\node[fill=black, draw, circle, thick, inner sep=4pt] (3) at (3,0) {};
	\node[fill=black, draw, circle, thick, inner sep=4pt] (4) at (4.5,0) {};
	\node[draw, circle, thick] (5) at (-1,-1) {\large$x$};
	\node[draw, circle, thick] (6) at (0,-1) {\large$\overline{x}$};
	\node[draw, circle, thick] (7) at (1.5,-1) {\large$x$};
	\node[draw, circle, thick] (8) at (3,-1) {\large$\overline{x}$};
	\node[draw, circle, thick] (9) at (4.5,-1) {\large$x$};
	\node[draw, circle, thick] (10) at (5.5,-1) {\large$x$};
	
	\draw[line width=1.2] (1) -- (2) node[midway, above] {};
	\draw[line width=1.2] (2) -- (3) node[midway, above] {};
	\draw[line width=1.2] (3) -- (4) node[midway, above] {};
	\draw[line width=1.2] (1) to[out=180, in=90] (5) node[midway, above left=-0.1cm and 0.8cm] {};
	
	\draw[line width=1.2, densely dashed] (1) -- (6) node[midway, left] {};
	\draw[line width=1.2] (2) -- (7) node[midway, left] {};
	\draw[line width=1.2, densely dashed] (3) -- (8) node[midway, right] {};
	\draw[line width=1.2] (4) -- (9) node[midway, right] {};
	\draw[line width=1.2] (4) to[out=0, in=90] (10) node[midway, above right=-0.3cm and 5.2cm] {};
	
	\node[] at (-1, 0) {};
	
	
\end{scope}

\end{tikzpicture}} 
  \caption{\label{fig:clausegadgets} Examples of Ising Hamiltonian constraint gadgets. Solid edges \mysolidline correspond to ferromagnetic $J = 1$ couplings, while dashed edges \mydashedline correspond to antiferromagnetic $J = -1$ couplings, and there are no local fields. The shaded vertices correspond to ancilla spins, while the unshaded vertices correspond to logical spins. \textbf{(a)} Clause-gadget for three variable NAE-SAT clause (i.e. two variable SAT clause) \textbf{(b)} Clause-gadget for four variable NAE-SAT clause (i.e. three variable SAT clause) \textbf{(c)} Clause-gadget for five variable NAE-SAT clause (i.e. four variable SAT clause) \textbf{(d)} Variable-gadget for a variable $x$ that appears in a SAT instance four times as the literal $x$ and twice as the literal $\overline{x}$. Each logical spin in the variable gadget connects to the logical spin in the clause gadget to which it belongs.}
\end{figure}

The motivation for the structure of the clause-gadgets is the ability to recursively split up NAE-clauses by interjecting a local ancilla variable, as in
\begin{align}
  \text{NAE}(a,b,c,d) = \text{NAE}(a,b,\overline{w}) \vee \text{NAE}(w,c,d),
\end{align}
until each NAE clause has been broken down to a 3-variable clause. Each 3-variable NAE clause becomes a triangle as shown in Fig.~\ref{fig:clause}-\ref{fig:clause3}, and the ancillary variable links them together. The variable-gadgets are constructed such that their ground states encode the number of occurrences of a variable and its negation, thus enforcing the internal logical consistency of a variable. These have the topology of a comb with bristles, where the sign of the coupling on each bristle corresponds to whether the variable appears as itself or its negation, as in Fig.~\ref{fig:var}. The variable-gadget is designed to have a non-degenerate, non-frustrated ground state, where all the spins along the base of the bristle (the shaded spins in Fig.~\ref{fig:var}) have the same spin state and represent the Boolean value of the variable.

The gadgets are then stitched together as in Fig.~\ref{fig:LogicalIsingGraph} to create a composite Hamiltonian that encodes all of the constraints in the NAE-SAT instance. The pictured reduction is far from unique. For example, by changing the ordering of the variable/clause-gadgets, it may be possible to construct input graphs with lower crossing numbers.

The ground states of the resulting Ising graph will have spin flip symmetry, but only the solution corresponding to the Boolean assignment $y=0$ will directly correspond to a solution to the original SAT instance. If $y=1$, the solution to the SAT instance can be recovered by negating the value of the other variables.

\begin{figure}[htb!]
  \begin{center}
  \begin{tikzpicture}
    \node[] at (0,0) {\includegraphics[scale = 1.2]{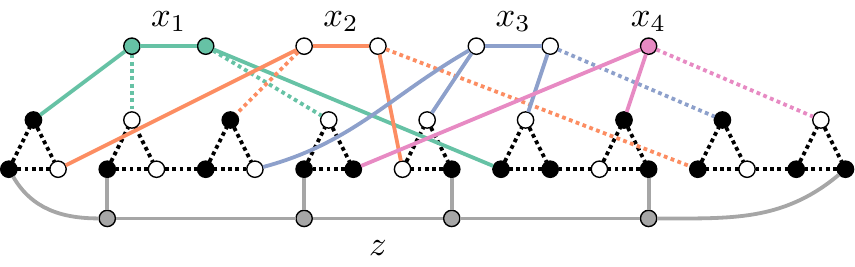}};
    \draw[fill=white, color=white] (-0.9,-1.6) rectangle ++(0.5,0.3);
    \node[] at (-0.5,-1.4) {\large$y$};
  \end{tikzpicture}
  \end{center}
  \caption{An Ising graph representing the SAT instance in Equation~\ref{eq:SAT}. Solid edges \mysolidline correspond to ferromagnetic $J = +1$ coupling strengths, while dotted edges \mydottedline correspond to antiferromagnetic $J = -1$ coupling strengths. The variable-gadgets are colour coded. The Ising graph is shaded according to a particular solution, with shaded vertices corresponding to a spin state of $+1$ (Boolean 0) and unshaded vertices corresponding to a spin state of $-1$ (Boolean 1). Directly reading off the Boolean states for $x_1,x_2,x_3,x_4$ yields $(1,0,0,1)$, but since~$y$ has the Boolean value of~$1$, the solution to the corresponding SAT instance requires the negated Boolean values $(x_1,x_2,x_3,x_4) = (0,1,1,0)$.\label{fig:LogicalIsingGraph}}
\end{figure}

\section{Embedding and scaling}\label{sec:scaling}
The runtime for the quantum adiabatic algorithm tends to increase with the number of qubits. This is because it scales as the inverse square of the gap between the ground and excited states, which generally decreases as the number of qubits increases. To help inform the architecture design process, we look at various architecture layouts and compare embedding efficiency between them. We then quantify how the amount of physical quantum resources of embedding scales with the size of QUBO and SAT problems. Additionally, we develop an efficient and deterministic subdivision-embedding for SAT problems onto grid-like architectures with crossings. This embedding procedure significantly reduces both time and space complexity overheads when compared to heuristic embedding algorithms.

\subsection{Architectures}\label{sec:architectures}
 Hardware graph layouts are required to have some degree of non-planarity in order to embed arbitrary NP-Complete Problems. This can be accomplished using either three-dimensional architectures, or crossing couplings in two-dimensions~\cite{istrail2000statistical, barahona1982computational}. The amount of connectivity within a hardware graph plays a large role in the efficiency of problem embedding sizes. Here, we briefly introduce three hardware graph layouts with distinct connectivity that we perform embedding size scaling analysis on in the next section.

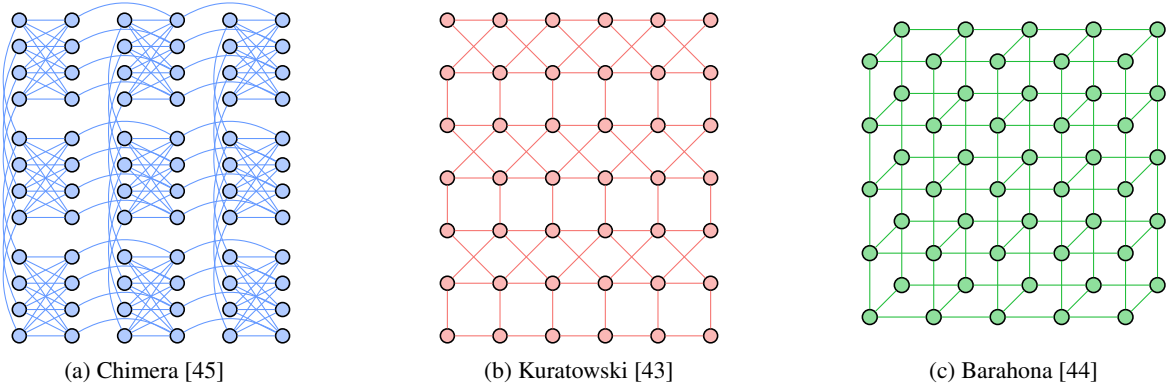
\begin{figure}[htb!]
  \centering
  \subcaptionbox{Chimera~\cite{harris2010experimental}\label{fig:figChimera}}[0.30\textwidth]{	\begin{tikzpicture}[scale=0.7, transform shape]

\definecolor{chimeracolor}{RGB}{99, 152, 255}
	
	\foreach \u in {0,...,2}
	\foreach\v in {0,...,2}
	\foreach \x in {0,...,1}
	\foreach \t in {0,...,3} {
		\node[draw,circle, color=black, fill = chimeracolor!50, scale = 0.8, line width=0.6] (\x\t\v\u) at (\x+\u*2,-.5*\t-\v*2.25) {};
	}

	\foreach \x in {0,...,3}
	\foreach \y in {0,...,2}
	\foreach \z in {0,...,2}
	\foreach \w in {0,...,3}
	{\draw[color=chimeracolor] (0\x\y\z) -- (1\w\y\z); }
	
	\foreach \x in {0,...,3}
	\foreach \y in {0,...,2}
	\foreach \z in {0,...,1}
	{	\pgfmathsetmacro{\w}{int(\z+1)}
		\draw[color=chimeracolor] (1\x\y\z) to[out=30,in=150] (1\x\y\w); }
	
	\foreach \x in {0,...,3}
	\foreach \y in {0,...,2}
	\foreach \z in {0,...,1}
	{	\pgfmathsetmacro{\w}{int(\z+1)}
		\draw[color=chimeracolor] (0\x\z\y) to[out=-115,in=115] (0\x\w\y); }
	\end{tikzpicture}}\hfill
  \subcaptionbox{Kuratowski~\cite{istrail2000statistical}\label{fig:figKuratowski}}[0.30\textwidth]{\begin{tikzpicture}[scale=0.7, transform shape]

\definecolor{kuratowskicolor}{RGB}{246, 113, 109}

	\foreach \x in {0,...,5}
	\foreach \y in {0,...,6} {
		\node[draw, circle, color=black, fill=kuratowskicolor!50, scale = 0.8, line width=0.6] (\x\y) at (\x,-\y) {};
	}
	
	\foreach \x in {0,...,4}
	\foreach \y in {0,...,6}
	{	\pgfmathsetmacro{\w}{int(\x+1)}
		\draw[color=kuratowskicolor] (\x\y) -- (\w\y); }
	
	\foreach \x in {0,...,4}
	\foreach \y in {0,2,4}
	{ \pgfmathsetmacro{\u}{int(\x+1)}
		\pgfmathsetmacro{\v}{int(\y+1)}
		\draw[color=kuratowskicolor] (\x\y) -- (\u\v);
	}
	
		\foreach \x in {0,...,4}
		\foreach \y in {1,3,5}
		{ \pgfmathsetmacro{\u}{int(\x+1)}
			\pgfmathsetmacro{\v}{int(\y-1)}
			\draw[color=kuratowskicolor] (\x\y) -- (\u\v);
		}
		
		\foreach \x in {0,...,5}
		\foreach \y in {1,3,5}
		{ \pgfmathsetmacro{\u}{int(\y+1)}
			\draw[color=kuratowskicolor] (\x\y) -- (\x\u); }
\end{tikzpicture}}\hfill
  \subcaptionbox{Barahona~\cite{barahona1982computational}\label{fig:figTwoLevelGrid}}[0.30\textwidth]{\begin{tikzpicture}[scale = 0.85, transform shape]

\definecolor{barahonacolor}{RGB}{31, 190, 56}

	\foreach \x in {0,...,4}
	\foreach \y in {0,...,4} {
		\node[draw, circle, color=black, fill=barahonacolor!50, scale = 0.7, line width=0.6] (\x\y) at (\x,-\y) {};
	}
	
	\foreach \x in {5,...,9}
	\foreach \y in {5,...,9} {
		\node[draw, circle, color=black, fill=barahonacolor!50, scale = 0.7, line width=0.6] (\x\y) at (\x-5.5,-\y+4.5) {};
	}
	
	\foreach \x in {0,...,3}
	\foreach \y in {0,...,4} {
		\pgfmathsetmacro{\u}{int(\x+1)}
		\draw[color=barahonacolor] (\x\y) -- (\u\y);
	}
	
	\foreach \x in {5,...,8}
	\foreach \y in {5,...,9} {
		\pgfmathsetmacro{\u}{int(\x+1)}
		\draw[color=barahonacolor] (\x\y) -- (\u\y);
	}
	
	\draw[color=barahonacolor] (00) -- (01);
	\draw[color=barahonacolor] (10) -- (11);
	\draw[color=barahonacolor] (20) -- (21);
	\draw[color=barahonacolor] (30) -- (31);
	\draw[color=barahonacolor] (40) -- (41);

	\draw[color=barahonacolor] (01) -- (02);
	\draw[color=barahonacolor] (11) -- (12);
	\draw[color=barahonacolor] (21) -- (22);
	\draw[color=barahonacolor] (31) -- (32);
	\draw[color=barahonacolor] (41) -- (42);

	\draw[color=barahonacolor] (02) -- (03);
	\draw[color=barahonacolor] (12) -- (13);
	\draw[color=barahonacolor] (22) -- (23);
	\draw[color=barahonacolor] (32) -- (33);
	\draw[color=barahonacolor] (42) -- (43);
	
	\draw[color=barahonacolor] (03) -- (04);
	\draw[color=barahonacolor] (13) -- (14);
	\draw[color=barahonacolor] (23) -- (24);
	\draw[color=barahonacolor] (33) -- (34);
	\draw[color=barahonacolor] (43) -- (44);
	
	\draw[color=barahonacolor] (55) -- (56);
	\draw[color=barahonacolor] (65) -- (66);
	\draw[color=barahonacolor] (75) -- (76);
	\draw[color=barahonacolor] (85) -- (86);
	\draw[color=barahonacolor] (95) -- (96);
	
	\draw[color=barahonacolor] (56) -- (57);
	\draw[color=barahonacolor] (66) -- (67);
	\draw[color=barahonacolor] (76) -- (77);
	\draw[color=barahonacolor] (86) -- (87);
	\draw[color=barahonacolor] (96) -- (97);
	
	\draw[color=barahonacolor] (57) -- (58);
	\draw[color=barahonacolor] (67) -- (68);
	\draw[color=barahonacolor] (77) -- (78);
	\draw[color=barahonacolor] (87) -- (88);
	\draw[color=barahonacolor] (97) -- (98);
	
	\draw[color=barahonacolor] (58) -- (59);
	\draw[color=barahonacolor] (68) -- (69);
	\draw[color=barahonacolor] (78) -- (79);
	\draw[color=barahonacolor] (88) -- (89);
	\draw[color=barahonacolor] (98) -- (99);
	
	\draw[color=barahonacolor] (00) -- (55);
	\draw[color=barahonacolor] (01) -- (56);
	\draw[color=barahonacolor] (02) -- (57);
	\draw[color=barahonacolor] (03) -- (58);
	\draw[color=barahonacolor] (04) -- (59);

	\draw[color=barahonacolor] (10) -- (65);
	\draw[color=barahonacolor] (11) -- (66);
	\draw[color=barahonacolor] (12) -- (67);
	\draw[color=barahonacolor] (13) -- (68);
	\draw[color=barahonacolor] (14) -- (69);
	
	\draw[color=barahonacolor] (20) -- (75);
	\draw[color=barahonacolor] (21) -- (76);
	\draw[color=barahonacolor] (22) -- (77);
	\draw[color=barahonacolor] (23) -- (78);
	\draw[color=barahonacolor] (24) -- (79);
	
	\draw[color=barahonacolor] (30) -- (85);
	\draw[color=barahonacolor] (31) -- (86);
	\draw[color=barahonacolor] (32) -- (87);
	\draw[color=barahonacolor] (33) -- (88);
	\draw[color=barahonacolor] (34) -- (89);
	
	\draw[color=barahonacolor] (40) -- (95);
	\draw[color=barahonacolor] (41) -- (96);
	\draw[color=barahonacolor] (42) -- (97);
	\draw[color=barahonacolor] (43) -- (98);
	\draw[color=barahonacolor] (44) -- (99);

	\node[circle] at (0,-4.75) {};

\end{tikzpicture}}
  \caption{Hardware graphs representing the topology of physical architectures.\label{fig:hardwaregraphs} The hardware graphs are assumed to be unbounded, allowing size constraints to be ignored.}
\end{figure}

The first hardware graph layout is the {\em Chimera} graph shown in Fig.~\ref{fig:figChimera}, that represents D-Wave's primary quantum annealing architecture~\cite{harris2010experimental}. It is a~6-regular graph composed of interlinked~$K_{4,4}$ bipartite `unit-cells' and has dynamic range that is significantly larger than potential restricted architectures, such as those with unit coupling strengths. This makes it comparatively straightforward to embed a wide range of input graphs. As such, it is one of the most extensively researched hardware graphs~\cite{choi2011minor,cai2014practical,klymko2014adiabatic,Boothby2016,zaribafiyan2017systematic}, and forms a standard against which other hardware graphs may be compared.

The layouts shown in Fig.~\ref{fig:figKuratowski} and Fig.~\ref{fig:figTwoLevelGrid} are examples of more restricted hardware graphs, which we refer to as the {\em Kuratowski} and {\em Barahona} hardware graphs respectively. They are 5-regular graphs with shorter-range couplings than the Chimera graph, making them more suitable for physical implementations with short-range couplings, such as phosphorus donors in silicon~\cite{hill2015surface} or quantum dots~\cite{veldhorst2017silicon}. Although more restrictive graphs are possible, these particular graphs were chosen because of their historical significance. The Barahona graph was the first to be used to prove that the unit-coupling Ising problem on a three-dimensional grid is NP-Complete~\cite{barahona1982computational}, while the Kuratowski was used to extend this result to general non-planar graphs~\cite{istrail2000statistical}.

\subsection{Direct embedding scaling}\label{sec:architectures-scaling}

In order to gauge the effectiveness of subdivision-embedding on different hardware graphs, randomly-generated cubic input graphs were subdivision-embedded on the Chimera, Barahona and Kuratowski architectures shown in Figure~\ref{fig:hardwaregraphs}. We use cubic graphs as a benchmark due to them having the lowest possible vertex degree such that they are still capable of encoding arbitrarily difficult problems, while also being directly subdivision-embeddable onto arbitrary hardware graph layouts with sufficient connectivity for encoding NP-Hard problems~\cite{istrail2000statistical,barahona1982computational}. In addition, input planar (capable of being embedded on a plane without edges crossing) and non-planar input graphs were compared. Although the planar graphs tested are incapable of encoding NP-Complete problems on unit-coupling architectures, they nonetheless approximate a lower-bound to the topological complexity of an input graph. These embedding studies were performed using an implementation of a minor-embedding heuristic developed by Cai \textit{et al}.~\cite{cai2014practical}. Although this algorithm was specifically intended for minor-embedding, for the case of input graphs with vertices of maximum degree three, the minor and subdivision-embeddings are topologically equivalent.
\begin{figure}
  \begin{center}
  \begin{tikzpicture}
  \node[] at (0,0) {\includegraphics[scale=0.6]{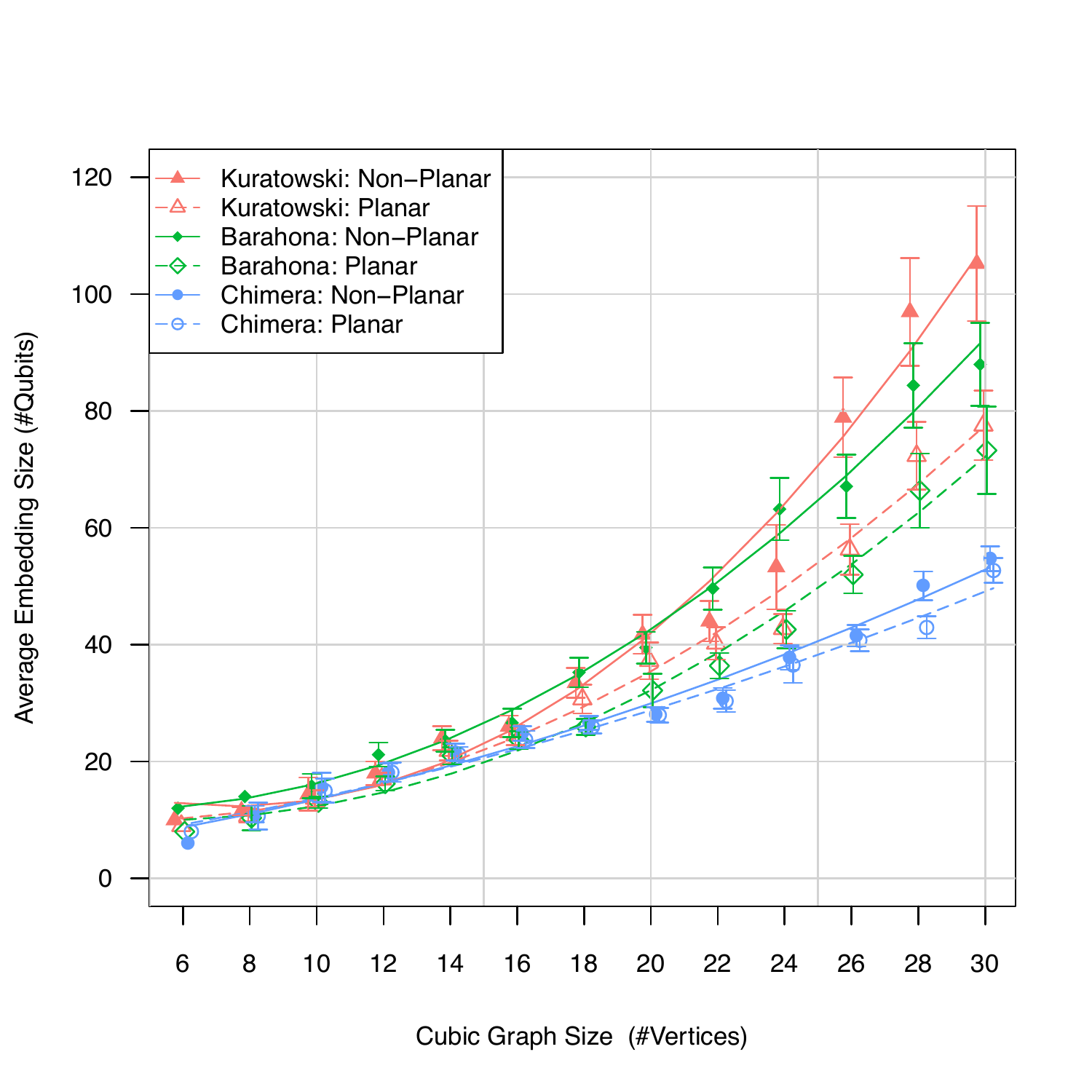}};
  \draw[fill=white, color=white] (-2.8, 3.3) rectangle (-0.4, 1.7);
  \node[] at (-1.6, 2.45)
  {\tiny
  \begin{tabular}{c@{\hspace{1pt}}c@{\hspace{1pt}}c}
  Non-Planar &$\to$& Kuratowski \\
  Planar &$\to$& Kuratowski \\
  Non-Planar & $\to$ & Barahona \\
  Planar & $\to$ & Barahona \\
  Non-Planar & $\to$ & Chimera \\
  Planar & $\to$ & Chimera \\
  \end{tabular}
  };
  \end{tikzpicture}
  \end{center}
  \caption{A scatter plot of resulting embedded graph sizes within Chimera, Barahona and Kuratowski hardware graph layouts with respect to planar and non-planar cubic input graph sizes. These embeddings were determined using the minor-embedding heuristic algorithm~\cite{cai2014practical}. A minor-embedding algorithm can be used for subdivision-embedding here because the embeddings are topologically equivalent for cubic input graphs. For each data point, up to 100 unique planar or non-planar input graph instances were examined (unless all unique graphs of a certain size had already been exhausted). For each input graph, the heuristic algorithm was run on progressively larger hardware graphs for 100 iterations and the smallest embedding was chosen. Each fitted curve shown is quadratic. \label{fig:cubicnewall}}
\end{figure}
For cubic input graphs of size $|V| \in \lbrace 6,8,\dots,30\rbrace$, up to 100 unique planar and non-planar instances were examined (unless all unique graphs of a certain size had already been exhausted). In order to find near-optimal embeddings, the heuristic algorithm was run for 100 iterations on progressively larger hardware graphs for each input graph and the smallest embedding was chosen. The scaling results are shown in Fig.~\ref{fig:cubicnewall} and Table~\ref{table:scaling}. These results indicate that directly embedding cubic graphs onto physical architectures using the heuristic minor-embedding algorithm yields a quadratic overhead on the number of qubits.

\begin{table}
  \centering
            \begin{tabular}{ c|c|c|c }
                & Chimera & Barahona & Kuratowski \\
                \hline \hline
                Planar Cubic Input & $\sim 0.027|V|^2$ & $\sim 0.101|V|^2$ & $\sim 0.101|V|^2$ \\  Non-Planar Cubic Input & $\sim 0.032|V|^2$ & $\sim 0.118|V|^2$ & $\sim 0.191|V|^2$ 
            \end{tabular}
  \vspace{0.3cm}
  \caption{This table shows the leading orders of the fitted quadratic curves for qubit counts required to subdivision-embed cubic input graphs of $|V|$ vertices onto the Chimera, Barahona and Kuratowski hardware graphs corresponding to the data shown in Figure~\ref{fig:cubicnewall}.\label{table:scaling}}
\end{table}

\subsection{Scaling of QUBO reduction and embedding}
\label{subsec:QUBO Scaling}

The complete degree reduction presented in Theorem~\ref{theorem:complete-degree-reduction} is a starting point for the reduction of QUBO problems, but it is likely that more efficient reductions are possible. As a reference point for future comparison, a scaling relation for the size of the transformed problem graph with respect to the input graph is determined, as well as upper bounds for the total number of physical qubits required with respect to complete input graph size.

Given a single vertex representing a spin with degree $D_0$ (see Appendix~\ref{sec:degree-reduction-complexity}), the function for complete degree reduction transforms the vertex into a number of vertices scaling log-linearly as
\begin{equation}
  \mathcal{O}(D_0 \log (D_0)^{\log_{2}5}),
\end{equation}
with proof provided in Appendix~\ref{sec:degree-reduction-complexity}. Since this transformation is recursive in nature, one might initially expect the number of vertices to scale exponentially. However, the number of vertices after each application of the reduction transformation has a square root dependence on the degree prior to transformation, which after repeated applications, effectively cancels the exponential behaviour. An example of this relation is shown in Figure~\ref{fig:scalability-degree1000}.

To compare resource requirements between classical and quantum methods, we examine the relationship between classical input size of a problem and the number of vertices after degree reduction to 3-regular. An arbitrary problem that can be represented as a~ZZ~Ising Hamiltonian has~$N$ boolean variables and~$N(N-1)/2$ interactions, where after scaling, each interaction has integer coupling strengths ranging between~$-M$ and~$M$ and can be represented with~$\lceil \log_2 (2M+1) \rceil$ bits. Thus, a precise encoding of an arbitrary instance of the problem requires
\begin{equation}
  \mathcal{O}\left( N^2\log(M)\right)
\end{equation}
classical bits.

\begin{figure}[h!]
  \centering
  \includegraphics[width=0.60\textwidth]{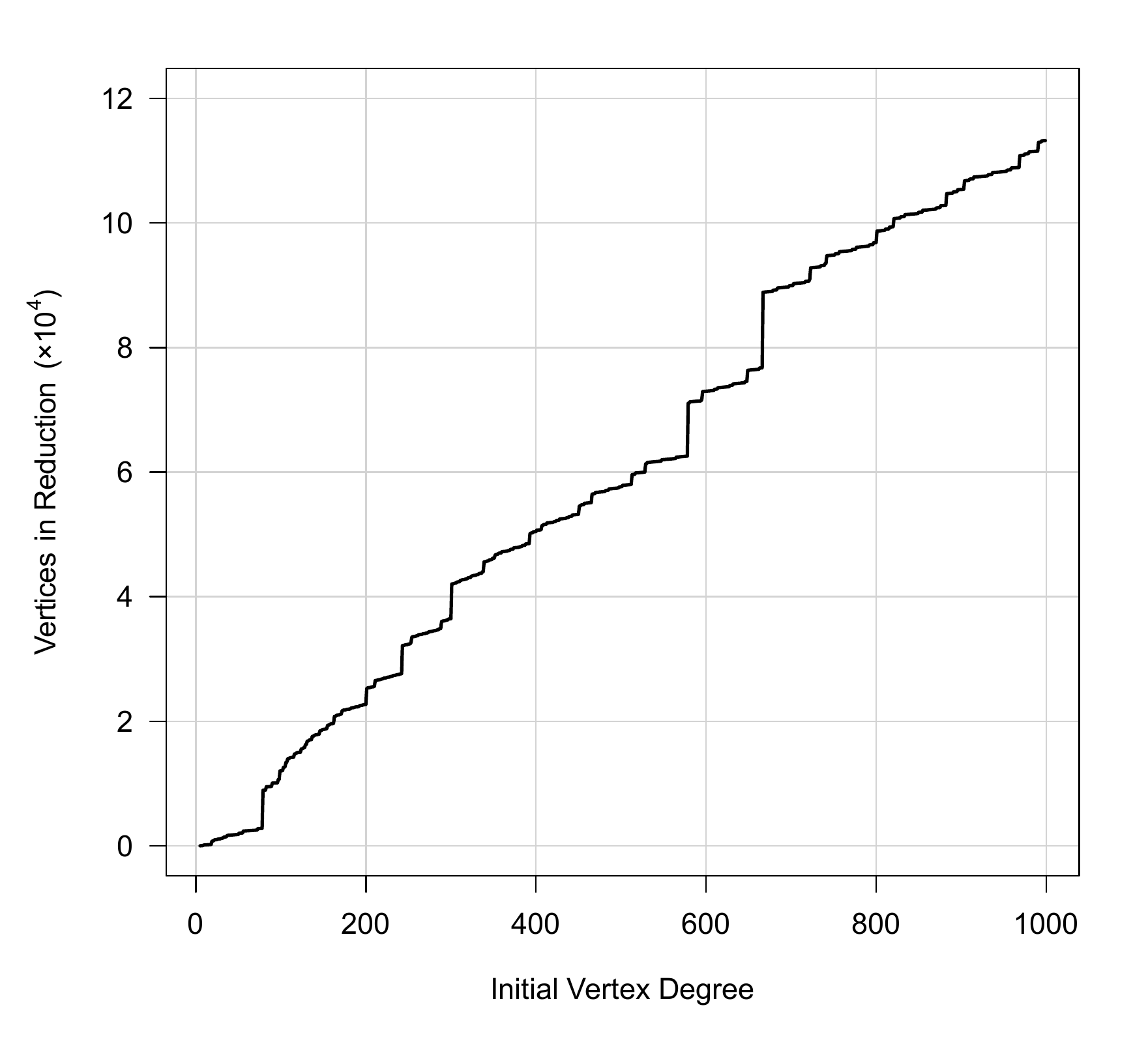}
  \caption{Qubit count upper-bounds after degree reduction of a single spin to a cluster of physical spins of at most degree three. These results were calculated by recursively applying the complete degree reduction transformation and the transformations shown in Fig.~\ref{fig:degree-reduction} to an initial single vertex until all transformed vertices have degree less than or equal to three.}
  \label{fig:scalability-degree1000}\par\medskip
\end{figure}\noindent

We will now show how degree reduction to an input graph of maximum degree three from an arbitrary problem graph with integer coupling strengths scales with respect to $N$ and $M$. We consider two cases relating to the replacement of edges of larger coupling strengths with edges of smaller coupling strengths as discussed in Section~\ref{QUBO}. First, consider the case where the problem graph edges are not required to be reduced (e.g. the hardware graph edges are capable of integer coupling strengths). After applying the degree reductions to all~$N$ vertices, by Eq.~\ref{eq:reduction-complexity}, the problem graph is replaced with an input graph of order
\begin{equation}
  \mathcal{O}\left( N^2 \log\left( N\right)^{\log_2 5} \right).
\end{equation}
Second, consider the case where the problem graph edges are reduced to edges with unit coupling strengths. Each edge is replaced by up to~$M$ edges. So each vertex will be incident to at most~$M(N-1)$ edges. After applying the degree reductions to all $N$ vertices, by Eq.~\ref{eq:reduction-complexity}, the problem graph is replaced with an input graph of order
\begin{equation}
  \mathcal{O}\left( MN^2 \log\left( MN\right)^{\log_2 5} \right).
\end{equation}
We can obtain further scaling relations between the sizes of these problem graphs and the physical embedded hardware graphs by squaring the number of vertices to account for the embedding stage (from the relations shown in Table~\ref{table:scaling}). It is important to note that the order of scaling will likely be improved by using tighter upper bounds in the complexity calculations and by applying optimisations to the algorithm.

To get a sense of how many physical qubits could be required in a worst-case scenario, we will go through an example for reducing a complete problem graph with integer coupling strengths and no local fields to an input graph with unit-coupling strengths, no local fields and vertices of maximum degree three, and then finally embedding it into an architecture with unit-couplings and no local fields. For a complete problem graph with~$k$-bit coupling strengths and~$N$ vertices, after edge splitting to unit coupling strengths, each vertex is reduced to a vertex with at most degree $(N-1)2^{k-1}$, where one bit is reserved for sign. For this example we will let~$k=6$. Let $T_3[D]$ be the number of vertices after a reduction to a graph with at most degree three from a single vertex of degree $D$, calculated using a recursive algorithm and shown in Figure~\ref{fig:scalability-degree1000}. Then we can write the total number of vertices in the cubic input graph as~$NT_3[32(N-1)]$. To then get the number of physical qubits in an embedding, the quadratic best-fits for embedding random non-planar cubic graphs into architectures shown in Table~\ref{table:scaling} and Fig.~\ref{fig:cubicnewall} can be used as approximate worst-case scaling functions. Computations using the Chimera, Kuratowski and Barahona architectures give the results shown in Figure~\ref{fig:upper-bound-of-physical-qubit-scalability}.

Since the aim of this paper is to lay the foundations for mapping~QUBO problems to heavily restricted~AQC architectures, little attention has been given to optimisation. We expect that the resource requirements and scaling functions calculated in this section are sensitive to improvements that could lead to significant resource savings can be expected.

\begin{figure}[h!]
  \centering
  \includegraphics[width=0.60\textwidth]{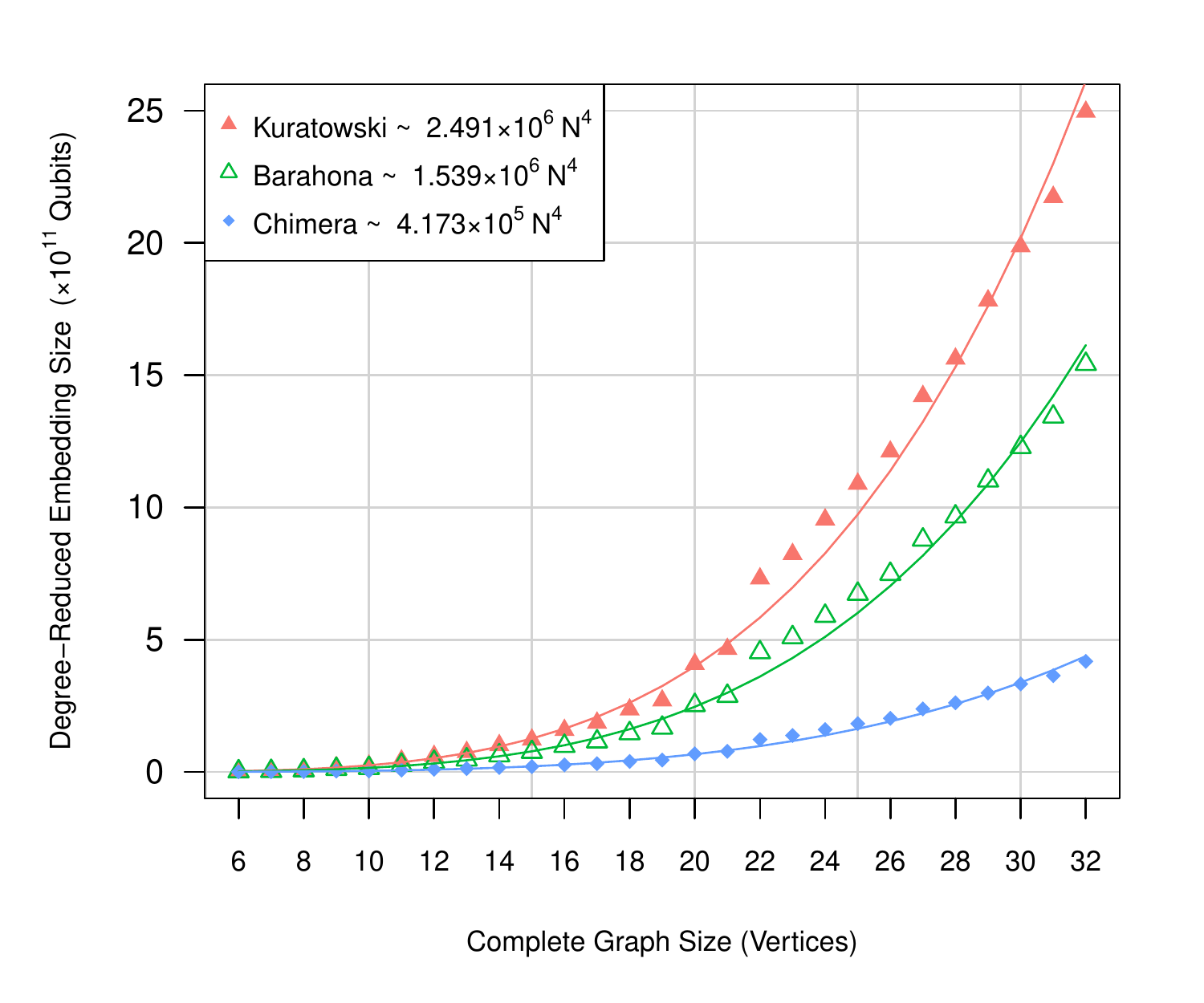}
  \caption{This figure shows the upper-bound for the number of physical qubits required to map an arbitrary fully-connected problem graph with~6-bit precision coupling strengths and no local fields onto the unit-coupling Chimera, Barahona and Kuratowski hardware graph layouts. The embedding sizes are calculated using the non-planar scaling relations for cubic graphs shown in Table~\ref{table:scaling}, which are for Chimera~$0.032 x^2 + 0.7x$, for Barahona~$0.118x^2 - 0.9x$, and for Kuratowski~$0.191x^2 - 3.0x$ where $x$ is the size of the graph after the reduction stage. The value for $x$ is computed as $x = NT_3[2^5(N-1)]$, where~$N$ is the number of qubits in the initial problem graph and $T_3[D_0]$ is the number of qubits resulting from degree reducing a single vertex of degree $D_0$ to a graph of at most degree three. The values for~$T_3[D_0]$ are assigned according to the results shown in Fig.~\ref{fig:scalability-degree1000} which were computed by recursively applying degree reduction.}\label{fig:upper-bound-of-physical-qubit-scalability}
\end{figure}\noindent

\subsection{Scaling of SAT reduction and embedding}
\label{subsec:SAT Scaling}
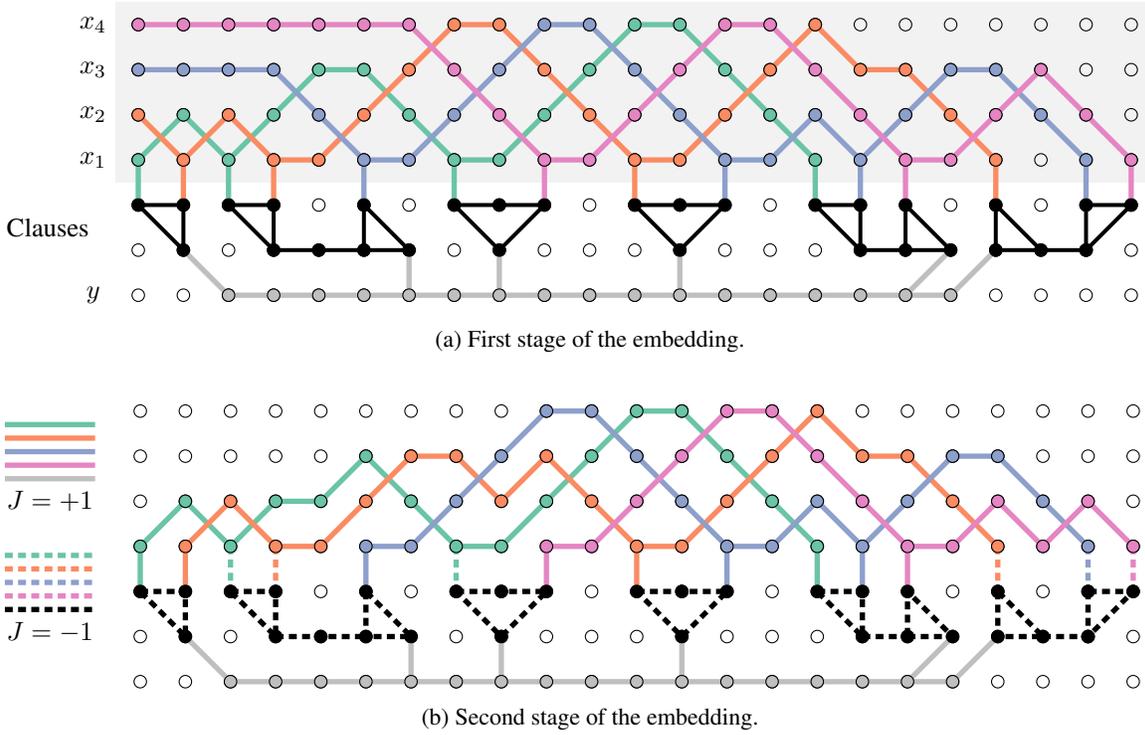
\begin{figure}[p]
\centering
\subcaptionbox{First stage of the embedding.\label{fig:braidembeddingA}}[0.8\textwidth]{\hspace*{-2cm}\definecolor{x1color}{RGB}{106,196,165}
\definecolor{x2color}{RGB}{250,139,99}
\definecolor{x3color}{RGB}{141, 159, 203}
\definecolor{x4color}{RGB}{229,133,195}

\tikzstyle{x1} = [draw, circle, scale=0.5, fill=x1color]
\tikzstyle{x2} = [draw, circle, scale=0.5, fill=x2color]
\tikzstyle{x3} = [draw, circle, scale=0.5, fill=x3color]
\tikzstyle{x4} = [draw, circle, scale=0.5, fill=x4color]
\tikzstyle{clause} = [draw, circle, scale=0.5, fill=black]
\tikzstyle{zv} = [draw, circle, scale=0.5, fill=gray!50]
\tikzstyle{empty} = [draw, circle, scale=0.5]

\begin{tikzpicture}[scale=0.6]

     \fill [gray!10] (-0.5,-0.5) rectangle (22.5,3.5);

    \foreach \x in {0,...,22}
	\foreach \y in {0,...,6} {
	    \pgfmathsetmacro{\u}{int(\y-3)}
		\node[draw, circle, scale = 0.5] (\x\y) at (\x,\u) {};
	}

    \node[] at (-1, 0) {$x_1$};
    \node[] at (-1, 1) {$x_2$};
    \node[] at (-1, 2) {$x_3$};
    \node[] at (-1, 3) {$x_4$};
    \node[] at (-2, -1.5) {Clauses};
    \node[] at (-1, -3) {$y$};

\draw[color=x1color, line width=2] (0, 0) -- (0, -1);
	\draw[color=x1color, line width=2] (2, 0) -- (2, -1);
	\draw[color=x1color, line width=2] (7, 0) -- (7, -1);
	\draw[color=x1color, line width=2] (15, 0) -- (15, -1);
	
	\draw[color=x2color, line width=2] (1, 0) -- (1, -1);
	\draw[color=x2color, line width=2] (3, 0) -- (3, -1);
	\draw[color=x2color, line width=2] (11, 0) -- (11, -1);
	\draw[color=x2color, line width=2] (19, 0) -- (19, -1);
	
	\draw[color=x3color, line width=2] (5, 0) -- (5, -1);
	\draw[color=x3color, line width=2] (13, 0) -- (13, -1);
	\draw[color=x3color, line width=2] (16, 0) -- (16, -1);
	\draw[color=x3color, line width=2] (21, 0) -- (21, -1);

	\draw[color=x4color, line width=2] (9, 0) -- (9, -1);
	\draw[color=x4color, line width=2] (17, 0) -- (17, -1);
	\draw[color=x4color, line width=2] (22, 0) -- (22, -1);

    \draw[color=x1color, line width=2] (0,0) -- (1,1) -- (2,0) -- (3,1) -- (4,2) -- (5,2) -- (6,1) -- (7,0) -- (8,0) -- (9,1) -- (10,2) -- (11,3) -- (12,3) -- (13,2) -- (14,1) -- (15,0);

	\node[x1] at (0,0) {};
	\node[x1] at (1,1) {};
	\node[x1] at (2,0) {};
	\node[x1] at (3,1) {};
	\node[x1] at (4,2) {};
	\node[x1] at (5,2) {};
	\node[x1] at (6,1) {};
	\node[x1] at (7,0) {};
	\node[x1] at (8,0) {};
	\node[x1] at (9,1) {};
	\node[x1] at (10,2) {};
	\node[x1] at (11,3) {};
	\node[x1] at (12,3) {};
	\node[x1] at (13,2) {};
	\node[x1] at (14,1) {};
	\node[x1] at (15,0) {};
	
	\draw[color=x2color, line width=2] (0,1) -- (1,0) -- (2,1) -- (3,0) -- (4,0) -- (5,1) -- (6,2) -- (7,3) -- (8,3) -- (9,2) -- (10,1) -- (11,0) -- (12,0) -- (13,1) -- (14,2) -- (15,3) -- (16,2) -- (17,2) -- (18,1) -- (19,0);
	
	\node[x2] at (0,1) {};
	\node[x2] at (1,0) {};
	\node[x2] at (2,1) {};
	\node[x2] at (3,0) {};
	\node[x2] at (4,0) {};
	\node[x2] at (5,1) {};
	\node[x2] at (6,2) {};
	\node[x2] at (7,3) {};
	\node[x2] at (8,3) {};
	\node[x2] at (9,2) {};
	\node[x2] at (10,1) {};
	\node[x2] at (11,0) {};
	\node[x2] at (12,0) {};
	\node[x2] at (13,1) {};
	\node[x2] at (14,2) {};
	\node[x2] at (15,3) {};
	\node[x2] at (16,2) {};
	\node[x2] at (17,2) {};
	\node[x2] at (18,1) {};
	\node[x2] at (19,0) {};
	
	\draw[color=x3color, line width=2] (0,2) -- (1,2) -- (2,2) -- (3,2) -- (4,1) -- (5,0) -- (6,0) -- (7,1) -- (8,2) -- (9,3) -- (10,3) -- (11,2) -- (13,0) -- (14,0) -- (15,1) -- (16,0) -- (17,1) -- (18,2) -- (19,2) -- (20,1) -- (21,0);
	
	\node[x3] at (0,2) {};
	\node[x3] at (1,2) {};
	\node[x3] at (2,2) {};
	\node[x3] at (3,2) {};
	\node[x3] at (4,1) {};
	\node[x3] at (5,0) {};
	\node[x3] at (6,0) {};
	\node[x3] at (7,1) {};
	\node[x3] at (8,2) {};
	\node[x3] at (9,3) {};
	\node[x3] at (10,3) {};
	\node[x3] at (11,2) {};
	\node[x3] at (12,1) {};
	\node[x3] at (13,0) {};
	\node[x3] at (14,0) {};
	\node[x3] at (15,1) {};
	\node[x3] at (16,0) {};
	\node[x3] at (17,1) {};
	\node[x3] at (18,2) {};
	\node[x3] at (19,2) {};
	\node[x3] at (20,1) {};
	\node[x3] at (21,0) {};
	
	\draw[color=x4color, line width=2] (0,3) -- (1,3) -- (2,3) -- (3,3) -- (4,3) -- (5,3) -- (6,3) -- (7,2) -- (8,1) -- (9,0) -- (10,0) -- (11,1) -- (12,2) -- (13,3) -- (14,3) -- (15,2) -- (16,1) -- (17,0) -- (18,0) -- (19,1) -- (20,2) -- (21,1) -- (22,0);
	
	\node[x4] at (0,3) {};
	\node[x4] at (1,3) {};
	\node[x4] at (2,3) {};
	\node[x4] at (3,3) {};
	\node[x4] at (4,3) {};
	\node[x4] at (5,3) {};
	\node[x4] at (6,3) {};
	\node[x4] at (7,2) {};
	\node[x4] at (8,1) {};
	\node[x4] at (9,0) {};
	\node[x4] at (10,0) {};
	\node[x4] at (11,1) {};
	\node[x4] at (12,2) {};
	\node[x4] at (13,3) {};
	\node[x4] at (14,3) {};
	\node[x4] at (15,2) {};
	\node[x4] at (16,1) {};
	\node[x4] at (17,0) {};
	\node[x4] at (18,0) {};
	\node[x4] at (19,1) {};
	\node[x4] at (20,2) {};
	\node[x4] at (21,1) {};
	\node[x4] at (22,0) {};
	
	\draw[color=gray!50, line width=2] (1, -2) -- (2, -3) -- (3, -3) -- (4, -3) -- (5, -3) -- (6, -3) -- (7, -3) -- (8,-3) -- (9, -3) -- (10, -3) -- (11, -3) -- (12, -3) -- (13, -3) -- (14, -3) -- (15, -3) -- (16, -3) -- (17, -3) -- (18, -3) -- (19, -2);
	\draw[color=gray!50, line width=2] (6, -3) -- (6, -2);
	\draw[color=gray!50, line width=2] (8, -3) -- (8, -2);
	\draw[color=gray!50, line width=2] (12, -3) -- (12, -2);
	\draw[color=gray!50, line width=2] (17, -3) -- (18, -2);
	
	\node[clause] at (0, -1) {};
	\node[clause] at (1, -1) {};
	\node[clause] at (1, -2) {};
	\draw[line width=1.5] (0, -1) -- (1, -1) -- (1, -2) -- cycle;
	
	\node[clause] at (2, -1) {};
	\node[clause] at (3, -1) {};
	\node[clause] at (3, -2) {};
	\node[clause] at (4, -2) {};
	\node[clause] at (5, -2) {};
	\node[clause] at (5, -1) {};
	\node[clause] at (6, -2) {};
	\draw[line width=1.5] (2, -1) -- (3, -1) -- (3, -2) -- cycle;
	\draw[line width=1.5] (5, -1) -- (5, -2) -- (6, -2) -- cycle;
	\draw[line width=1.5] (3, -2) -- (4,-2) -- (5, -2);
	
	\node[clause] at (7, -1) {};
	\node[clause] at (8, -1) {};
	\node[clause] at (9, -1) {};
	\node[clause] at (8, -2) {};
	\draw[line width=1.5] (7, -1) -- (8, -1) -- (9, -1) -- (8, -2) -- cycle;
	
	\node[clause] at (11, -1) {};
	\node[clause] at (12, -1) {};
	\node[clause] at (13, -1) {};
	\node[clause] at (12, -2) {};
	\draw[line width=1.5] (11, -1) -- (12, -1) -- (13, -1) -- (12, -2) -- cycle;

	\node[clause] at (15, -1) {};
	\node[clause] at (16, -1) {};
	\node[clause] at (16, -2) {};
	\draw[line width=1.5] (15, -1) -- (16, -1) -- (16, -2) -- cycle;
	\node[clause] at (17, -2) {};
	\draw[line width=1.5] (16, -2) -- (17, -2);
	\node[clause] at (17, -1) {};
	\node[clause] at (18, -2) {};
	\draw[line width=1.5] (17, -2) -- (18, -2) -- (17, -1) -- cycle;
	
	\node[clause] at (19, -1) {};
	\node[clause] at (19, -2) {};
	\node[clause] at (20, -2) {};
	\draw[line width=1.5] (19, -1) -- (19, -2) -- (20, -2) -- cycle;
	
	\node[clause] at (21, -2) {};
	\draw[line width=1.5] (20, -2) -- (21, -2);
	
    \node[clause] at (21, -1) {};
    \node[clause] at (21, -2) {};
    \node[clause] at (22, -1) {};
    \draw[line width=1.5] (21, -1) -- (21, -2) -- (22, -1) -- cycle;
    
    \node[zv] at (2, -3) {};
    \node[zv] at (3, -3) {};
    \node[zv] at (4, -3) {};
    \node[zv] at (5, -3) {};
    \node[zv] at (6, -3) {};
    \node[zv] at (7, -3) {};
    \node[zv] at (8, -3) {};
    \node[zv] at (9, -3) {};
    \node[zv] at (10, -3) {};
    \node[zv] at (11, -3) {};
    \node[zv] at (12, -3) {};
    \node[zv] at (13, -3) {};
    \node[zv] at (14, -3) {};
    \node[zv] at (15, -3) {};
    \node[zv] at (16, -3) {};
    \node[zv] at (17, -3) {};
    \node[zv] at (18, -3) {};

\end{tikzpicture}} \\ 
\subcaptionbox{Second stage of the embedding.\label{fig:braidembeddingB}}[0.8\textwidth]{\hspace*{-2cm}\definecolor{x1color}{RGB}{106,196,165}
\definecolor{x2color}{RGB}{250,139,99}
\definecolor{x3color}{RGB}{141, 159, 203}
\definecolor{x4color}{RGB}{229,133,195}

\tikzstyle{x1} = [draw, circle, scale=0.5, fill=x1color]
\tikzstyle{x2} = [draw, circle, scale=0.5, fill=x2color]
\tikzstyle{x3} = [draw, circle, scale=0.5, fill=x3color]
\tikzstyle{x4} = [draw, circle, scale=0.5, fill=x4color]
\tikzstyle{clause} = [draw, circle, scale=0.5, fill=black]
\tikzstyle{zv} = [draw, circle, scale=0.5, fill=gray!50]
\tikzstyle{empty} = [draw, circle, scale=0.5]

\begin{tikzpicture}[scale=0.6]

    \node[] at (0,4) {};
    
    \draw[line width=2, color=x1color] (-3, 2.7) -- (-1, 2.7);
    \draw[line width=2, color=x2color] (-3, 2.4) -- (-1, 2.4);
    \draw[line width=2, color=x3color] (-3, 2.1) -- (-1, 2.1);
    \draw[line width=2, color=x4color] (-3, 1.8) -- (-1, 1.8);
    \draw[line width=2, color=gray!50] (-3, 1.5) -- (-1, 1.5);
    \node[] at (-2, 1) {$J = +1$};
    
    \draw[densely dashed, line width=2, color=x1color] (-3, -0.2) -- (-1, -0.2);
    \draw[densely dashed, line width=2, color=x2color] (-3, -0.5) -- (-1, -0.5);
    \draw[densely dashed, line width=2, color=x3color] (-3, -0.8) -- (-1, -0.8);
    \draw[densely dashed, line width=2, color=x4color] (-3, -1.1) -- (-1, -1.1);
    \draw[densely dashed, line width=2] (-3, -1.4) -- (-1, -1.4);
    \node[] at (-2, -1.9) {$J = -1$};

    \foreach \x in {0,...,22}
	\foreach \y in {0,...,6} {
	    \pgfmathsetmacro{\u}{int(\y-3)}
		\node[draw, circle, scale = 0.5] (\x\y) at (\x,\u) {};
	}

\draw[color=x1color, line width=2] (0, 0) -- (0, -1);
	\draw[color=x1color, line width=2, densely dashed] (2, 0) -- (2, -1);
	\draw[color=x1color, line width=2, densely dashed] (7, 0) -- (7, -1);
	\draw[color=x1color, line width=2] (15, 0) -- (15, -1);
	
	\draw[color=x2color, line width=2] (1, 0) -- (1, -1);
	\draw[color=x2color, line width=2, densely dashed] (3, 0) -- (3, -1);
	\draw[color=x2color, line width=2] (11, 0) -- (11, -1);
	\draw[color=x2color, line width=2, densely dashed] (19, 0) -- (19, -1);
	
	\draw[color=x3color, line width=2] (5, 0) -- (5, -1);
	\draw[color=x3color, line width=2] (13, 0) -- (13, -1);
	\draw[color=x3color, line width=2] (16, 0) -- (16, -1);
	\draw[color=x3color, line width=2, densely dashed] (21, 0) -- (21, -1);

	\draw[color=x4color, line width=2] (9, 0) -- (9, -1);
	\draw[color=x4color, line width=2] (17, 0) -- (17, -1);
	\draw[color=x4color, line width=2, densely dashed] (22, 0) -- (22, -1);

    \draw[color=x1color, line width=2] (0,0) -- (1,1) -- (2,0) -- (3,1) -- (4,1) -- (5,2) -- (6,1) -- (7,0) -- (8,0) -- (9,1) -- (10,2) -- (11,3) -- (12,3) -- (13,2) -- (14,1) -- (15,0);
    
    \node[] at (-1, 0) {};
    \node[] at (-1, 1) {};
    \node[] at (-1, 2) {};
    \node[] at (-1, 3) {};
    \node[] at (-2, -1.5) {};
    \node[] at (-1, -3) {};

	\node[x1] at (0,0) {};
	\node[x1] at (1,1) {};
	\node[x1] at (2,0) {};
	\node[x1] at (3,1) {};
	\node[x1] at (4,1) {};
	\node[x1] at (5,2) {};
	\node[x1] at (6,1) {};
	\node[x1] at (7,0) {};
	\node[x1] at (8,0) {};
	\node[x1] at (9,1) {};
	\node[x1] at (10,2) {};
	\node[x1] at (11,3) {};
	\node[x1] at (12,3) {};
	\node[x1] at (13,2) {};
	\node[x1] at (14,1) {};
	\node[x1] at (15,0) {};
	
	\draw[color=x2color, line width=2] (1,0) -- (2,1) -- (3,0) -- (4,0) -- (5,1) -- (6,2) -- (7,2) -- (8,1) -- (9,2) -- (10,1) -- (11,0) -- (12,0) -- (13,1) -- (14,2) -- (15,3) -- (16,2) -- (17,2) -- (18,1) -- (19,0);
	
	\node[x2] at (1,0) {};
	\node[x2] at (2,1) {};
	\node[x2] at (3,0) {};
	\node[x2] at (4,0) {};
	\node[x2] at (5,1) {};
	\node[x2] at (6,2) {};
	\node[x2] at (7,2) {};
	\node[x2] at (8,1) {};
	\node[x2] at (9,2) {};
	\node[x2] at (10,1) {};
	\node[x2] at (11,0) {};
	\node[x2] at (12,0) {};
	\node[x2] at (13,1) {};
	\node[x2] at (14,2) {};
	\node[x2] at (15,3) {};
	\node[x2] at (16,2) {};
	\node[x2] at (17,2) {};
	\node[x2] at (18,1) {};
	\node[x2] at (19,0) {};
	
	\draw[color=x3color, line width=2] (5,0) -- (6,0) -- (7,1) -- (8,2) -- (9,3) -- (10,3) -- (11,2) -- (13,0) -- (14,0) -- (15,1) -- (16,0) -- (17,1) -- (18,2) -- (19,2) -- (20,1) -- (21,0);
	
	\node[x3] at (5,0) {};
	\node[x3] at (6,0) {};
	\node[x3] at (7,1) {};
	\node[x3] at (8,2) {};
	\node[x3] at (9,3) {};
	\node[x3] at (10,3) {};
	\node[x3] at (11,2) {};
	\node[x3] at (12,1) {};
	\node[x3] at (13,0) {};
	\node[x3] at (14,0) {};
	\node[x3] at (15,1) {};
	\node[x3] at (16,0) {};
	\node[x3] at (17,1) {};
	\node[x3] at (18,2) {};
	\node[x3] at (19,2) {};
	\node[x3] at (20,1) {};
	\node[x3] at (21,0) {};
	
	\draw[color=x4color, line width=2] (9,0) -- (10,0) -- (11,1) -- (12,2) -- (13,3) -- (14,3) -- (15,2) -- (16,1) -- (17,0) -- (18,0) -- (19,1) -- (20,0) -- (21,1) -- (22,0);
	
	\node[x4] at (9,0) {};
	\node[x4] at (10,0) {};
	\node[x4] at (11,1) {};
	\node[x4] at (12,2) {};
	\node[x4] at (13,3) {};
	\node[x4] at (14,3) {};
	\node[x4] at (15,2) {};
	\node[x4] at (16,1) {};
	\node[x4] at (17,0) {};
	\node[x4] at (18,0) {};
	\node[x4] at (19,1) {};
	\node[x4] at (20,0) {};
	\node[x4] at (21,1) {};
	\node[x4] at (22,0) {};
	
	\draw[color=gray!50, line width=2] (1, -2) -- (2, -3) -- (3, -3) -- (4, -3) -- (5, -3) -- (6, -3) -- (7, -3) -- (8,-3) -- (9, -3) -- (10, -3) -- (11, -3) -- (12, -3) -- (13, -3) -- (14, -3) -- (15, -3) -- (16, -3) -- (17, -3) -- (18, -3) -- (19, -2);
	\draw[color=gray!50, line width=2] (6, -3) -- (6, -2);
	\draw[color=gray!50, line width=2] (8, -3) -- (8, -2);
	\draw[color=gray!50, line width=2] (12, -3) -- (12, -2);
	\draw[color=gray!50, line width=2] (17, -3) -- (18, -2);
	
	\node[clause] at (0, -1) {};
	\node[clause] at (1, -1) {};
	\node[clause] at (1, -2) {};
	\draw[densely dashed, line width=2] (0, -1) -- (1, -1) -- (1, -2) -- cycle;
	
	\node[clause] at (2, -1) {};
	\node[clause] at (3, -1) {};
	\node[clause] at (3, -2) {};
	\node[clause] at (4, -2) {};
	\node[clause] at (5, -2) {};
	\node[clause] at (5, -1) {};
	\node[clause] at (6, -2) {};
	\draw[densely dashed, line width=2] (2, -1) -- (3, -1) -- (3, -2) -- cycle;
	\draw[densely dashed, line width=2] (5, -1) -- (5, -2) -- (6, -2) -- cycle;
	\draw[densely dashed, line width=2] (3, -2) -- (4,-2) -- (5, -2);
	
	\node[clause] at (7, -1) {};
	\node[clause] at (8, -1) {};
	\node[clause] at (9, -1) {};
	\node[clause] at (8, -2) {};
	\draw[densely dashed, line width=2] (7, -1) -- (8, -1) -- (9, -1) -- (8, -2) -- cycle;
	
	\node[clause] at (11, -1) {};
	\node[clause] at (12, -1) {};
	\node[clause] at (13, -1) {};
	\node[clause] at (12, -2) {};
	\draw[densely dashed, line width=2] (11, -1) -- (12, -1) -- (13, -1) -- (12, -2) -- cycle;

	\node[clause] at (15, -1) {};
	\node[clause] at (16, -1) {};
	\node[clause] at (16, -2) {};
	\draw[densely dashed, line width=2] (15, -1) -- (16, -1) -- (16, -2) -- cycle;
	\node[clause] at (17, -2) {};
	\draw[densely dashed, line width=2] (16, -2) -- (17, -2);
	\node[clause] at (17, -1) {};
	\node[clause] at (18, -2) {};
	\draw[densely dashed, line width=2] (17, -2) -- (18, -2) -- (17, -1) -- cycle;
	
	\node[clause] at (19, -1) {};
	\node[clause] at (19, -2) {};
	\node[clause] at (20, -2) {};
	\draw[densely dashed, line width=2] (19, -1) -- (19, -2) -- (20, -2) -- cycle;
	
	\node[clause] at (21, -2) {};
	\draw[densely dashed, line width=2] (20, -2) -- (21, -2);
	
    \node[clause] at (21, -1) {};
    \node[clause] at (21, -2) {};
    \node[clause] at (22, -1) {};
    \draw[densely dashed, line width=2] (21, -1) -- (21, -2) -- (22, -1) -- cycle;

    \node[zv] at (2, -3) {};
    \node[zv] at (3, -3) {};
    \node[zv] at (4, -3) {};
    \node[zv] at (5, -3) {};
    \node[zv] at (6, -3) {};
    \node[zv] at (7, -3) {};
    \node[zv] at (8, -3) {};
    \node[zv] at (9, -3) {};
    \node[zv] at (10, -3) {};
    \node[zv] at (11, -3) {};
    \node[zv] at (12, -3) {};
    \node[zv] at (13, -3) {};
    \node[zv] at (14, -3) {};
    \node[zv] at (15, -3) {};
    \node[zv] at (16, -3) {};
    \node[zv] at (17, -3) {};
    \node[zv] at (18, -3) {};

\end{tikzpicture}}
\caption{Braid subdivision-embedding stages for the SAT instance in Equation~\ref{eq:SAT}. \textbf{(a)} In the first stage, the top $n$ rows are populated entirely by variable-gadgets which are interwoven according to a priority queue based on their order of appearance. The clause-gadgets and ancilla variable-gadget are fairly simple to attach afterwards. For any given clause-gadget, the order in which its corresponding variable-gadgets are attached does not matter. \textbf{(b)} In the second stage, the variable-gadgets are trimmed according to the number of clauses they actually span, and couplings are introduced. They may also be braided downwards to fill gaps, resulting in a more compact embedding. Further optimisations could also be applied to reduce qubit resource requirements.\label{fig:braidembedding}}
\end{figure}

\begin{algorithm}[p]
  \SetAlgoLined
  
  Create queue $Q$ of ordered variable appearances in $\phi$.

  \ForEach{grid column, left to right}{
    Sort logical variables $\lbrace x_i\rbrace $ by priority in $Q$.

    \eIf{first column}{Place sorted logical variables vertically.
    }{
      \For{sorted logical variables $\lbrace x_i\rbrace$}{
        \uIf{able to extend diagonally down-and-right}{Extend $x_i$ diagonally down to next column.}
        \uElseIf{able to extend horizontally right}{Extend $x_i$ horizontally right to next column.}
        \uElse{Extend $x_i$ diagonally up to next column.}
      }
    }
    If variable of highest priority reaches lowest row, pop from queue $Q$.
  }
  Attach clause and ancilla variable gadgets. (see Fig.~\ref{fig:braidembeddingA}).
  
  Trim excess edges, add couplings and apply minor optimisations (see Fig.~\ref{fig:braidembeddingB}).
  \caption{\label{algo:braidalgo}Braid Embedding a SAT instance $\phi$}
\end{algorithm}
After reducing a $C$-clause $k$-SAT instance to a cubic input graph for subdivision-embedding, it can be shown that the number of vertices required is at most
\begin{align}
  |V| = 4Ck - 2C - 4
\end{align}
using the constructions shown in Section~\ref{sec:SAT}. To subdivision-embed the cubic input graph into a hardware graph, one option is to use a minor-embedding heuristic, which for the specific algorithm used in this work~\cite{cai2014practical} has resource scaling shown in Figure~\ref{fig:cubicnewall}. However, for particular hardware graph layouts, such as a regular grid with next-nearest neighbour interactions, it is possible to deterministically subdivision-embed SAT in a way that fully utilises the available crossings by braiding the variable-gadgets, as shown in Figure~\ref{fig:braidembedding}. In this embedding, the details of which are given in Algorithm~\ref{algo:braidalgo}, the variable-gadgets are interwoven according to an assigned priority, while the clause-gadgets are horizontally spaced to match the placement of the variable-gadgets. Since each variable-gadget occupies the equivalent of a single row and the clause-gadgets occupy a total of two rows, the maximum required vertical dimension of the grid is $(n + 3)$, where $n$ is the number of distinct variables in the SAT instance, and accounting for the extra ancilla variable. The horizontal length of the embedding of a $C$-clause~$k$-SAT instance is required to be at most twice the number of variable appearances in the~SAT instance, i.e. $2Ck$. This is because braiding is only a nearest-neighbour pairwise operation, so two braiding operations are required for every variable-gadget to have a chance to braid up or down (see for example, columns six to eight in Figure~\ref{fig:braidembeddingA}). Thus, a loose upper-bound for the number of qubits required to embed a $C$-clause, $n$-variable, $k$-SAT instance on a regular grid-like architecture with crossings is
\begin{align}
  \# \mathrm{qubits} \leq 2Ck(n+3).
\end{align}
For an arbitrary SAT instance where the number of variables per clause is not constant, $Ck$ may be replaced with the total number of variable appearances.

The algorithm for the braided subdivision-embedding shown in Fig.~\ref{fig:braidembeddingB} begins by ordering all variable appearances within a SAT instance, and creating a priority queue which gives a variable higher priority if it appears earlier on in a SAT instance. Starting from the left-most column, the grid is then scanned left to right, column by column. For each column, the variables, ordered by priority, are either braided down, unbraided, or braided up in order of preference, and depending on available space in the next column. This all occurs within the shaded region of Figure~\ref{fig:braidembeddingA}. It is only after this that the clause-gadgets and ancilla variable-gadget are attached in the three rows below. In the second stage of the embedding, the edges ending in degree-one vertices are trimmed until the embedding is topologically equivalent to the Ising input graph, and couplings are added.

In addition to providing more efficient embeddings, another advantage of this deterministic embedding is that the computation time required to find an embedding typically scales better than heuristic embedding methods. Due to requiring only one scan of the grid based on a priority queue, the time-complexity of the algorithm scales as the area of the grid, i.e. $\mathcal{O}(nCk)$. It is also generalisable (albeit less efficiently), to grid-like architectures with sparser crossings, such as the Kuratowski hardware graph. It is worth noting that clause ordering may affect qubit resource requirements, such as choosing an ordering that reduces the average horizontal distance each variable-gadget spans.

\section{Conclusion}
\label{Conclusion}
We developed techniques for mapping NP-Hard QUBO and SAT problems onto architectures with restrictive constraints where each spin couples to a small number of other spins, couplings consist of only unit coupling strengths i.e. $\{0, \pm J\}$, and spins have no local magnetic fields. This was done by separating the mapping problem into reduction and embedding stages. For the embedding stage, subdivision-embedding was developed and the notion of effective coupling strength was introduced for the Ising energy minimisation problem, which is used to help set compatible hardware coupling strength parameters for subdivision-embeddings. For the reduction stage, the techniques used for QUBO and SAT differ. Complete degree reduction was introduced for QUBO that reduces the number of couplings required for a spin by replacing it with a fully connected cluster of spins in such a way that the cluster acts as a single logical spin. Constraint Hamiltonians were described for SAT that encode problem instances in cubic Ising graphs requiring only unit couplings and no local fields. For both QUBO and SAT reduction stage techniques, we examined how the physical qubit count scales with problem size on Chimera, Barahona and Kuratowski hardware graph layouts, each of which have distinct connectivity. Additionally, an efficient deterministic embedding method for SAT was introduced for an architecture layout of a grid with crossings. For mapping QUBO problems onto heavily restricted physical hardware, we determined upper-bounds for the required number of physical qubits. For realistic problem sizes, we expect optimisations (not considered here) could significantly reduce these resource costs. Additionally, due to the quadratic qubit scaling overhead of the heuristic minor-embedding method used, a more resource efficient embedding method could provide further reductions of embedding size. For SAT, embedding size reductions may be achieved through finding optimal clause orderings. The underlying theory for the techniques developed in this work could aid in transforming Ising graphs for the purposes of Ising energy minimisation and could potentially assist in the development and optimisation of current and future embedding methods.


\newpage{}
\section*{Appendices}
\appendix
\renewcommand\thefigure{\thesection.\arabic{figure}}
\setcounter{figure}{0} 
\renewcommand\thedefinition{\thesection.\arabic{definition}}
\setcounter{definition}{0} 
\renewcommand\thelemma{\thesection.\arabic{lemma}}
\setcounter{lemma}{0}
\renewcommand\thetheorem{\thesection.\arabic{theorem}}
\setcounter{theorem}{0}
\renewcommand\thecorollary{\thesection.\arabic{corollary}}
\setcounter{corollary}{0}
\section{A proof of the correctness of subdivision-embedding} \label{sec:spin-chain-analysis}
In the following section, spin layout concepts will be defined in the language of graph theory. A graph is a convenient representation 
that can describe the spin-spin topology and interaction strengths within a layout. We will show an equivalence property for 
spin chains based on an \emph{effective coupling strength}, 
which will, in turn, be used to show four transformations 
that could potentially be used as steps in a subdivision-embedding algorithm. Most of the following definitions are summarised 
as examples in Figures~\ref{fig:spin-chain-examples} and~\ref{fig:spin-chain-combined-example}.

\begin{definition}[Trimmed Path -- Figure~\ref{fig:trimmed-path}]
  Let~$P$ be a non-trivial path or cycle with leaf vertices~$v_1$ and~$v_2$ or some vertex~$v_1=v_2$ respectively. We define the trimmed path of~$P$ to be the path remaining after removing $v_1$ and $v_2$, or more precisely, the induced subgraph of~$P$ defined on the vertex set~$V(P) - \{v_1, v_2\}$, denoted~$T$.
\end{definition}

\begin{definition}[Ising Graph]
  An Ising graph is a graph representation of an Ising spin glass, where vertices represent spins and edges represent interactions. Local field strengths are represented by vertex weights and coupling strengths are represented by edge weights.
\end{definition}

\begin{definition}[Outer Graph -- Figure~\ref{fig:outer-lattice}]
  Let $P$ be a non-trivial path or cycle that is a subgraph of an Ising graph $I$ and let~$T$ be the trimmed path of~$P$, then we define the \emph{outer graph} of~$P$ with respect to~$I$ to be the graph that remains after removing all edges and trimmed path vertices of $P$, or more precisely, the subgraph of~$I$ defined on the vertex set~$V(I) - V(T)$ and edge set~$E(I) - E(P)$, denoted~$O$.
\end{definition}

\begin{definition}[Spin Chain -- Figure~\ref{fig:spin-chain}]
  Let $Q$ be a non-trivial path or cycle that is an induced subgraph of an Ising graph $I$. The path~$Q$ is called a \emph{spin chain} if and only if all the vertices within its trimmed path have precisely degree two with respect to $I$.\label{def:spin-chain}
\end{definition}

\begin{figure}[h!]
  \centering
  \subcaptionbox{Spin chain\label{fig:spin-chain}}[0.32\textwidth]{\includegraphics[width=0.32\textwidth]{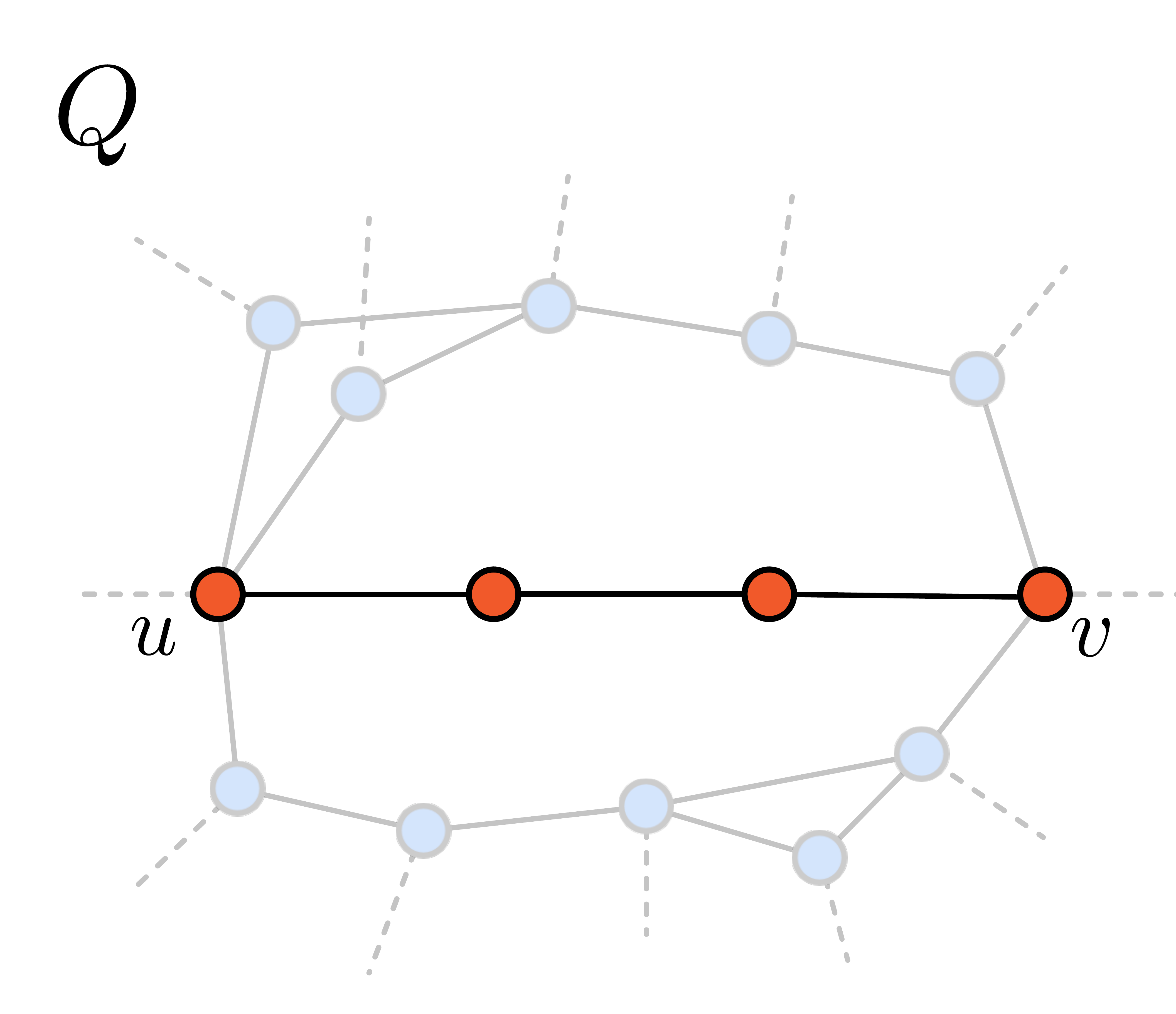}}    
  \subcaptionbox{Trimmed path\label{fig:trimmed-path}}[0.32\textwidth]{\includegraphics[width=0.32\textwidth]{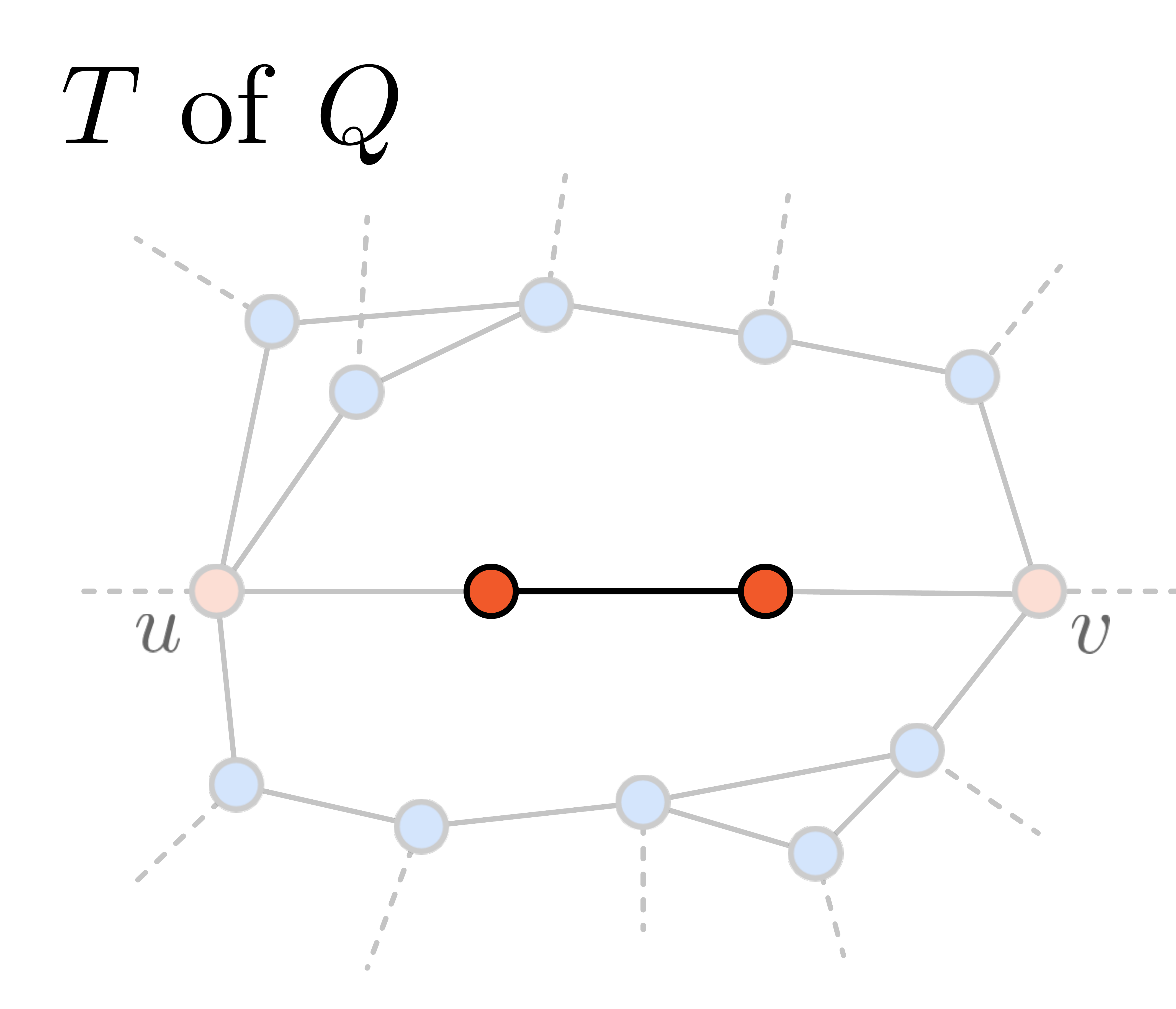}}
  \subcaptionbox{Outer lattice\label{fig:outer-lattice}}[0.32\textwidth]{\includegraphics[width=0.32\textwidth]{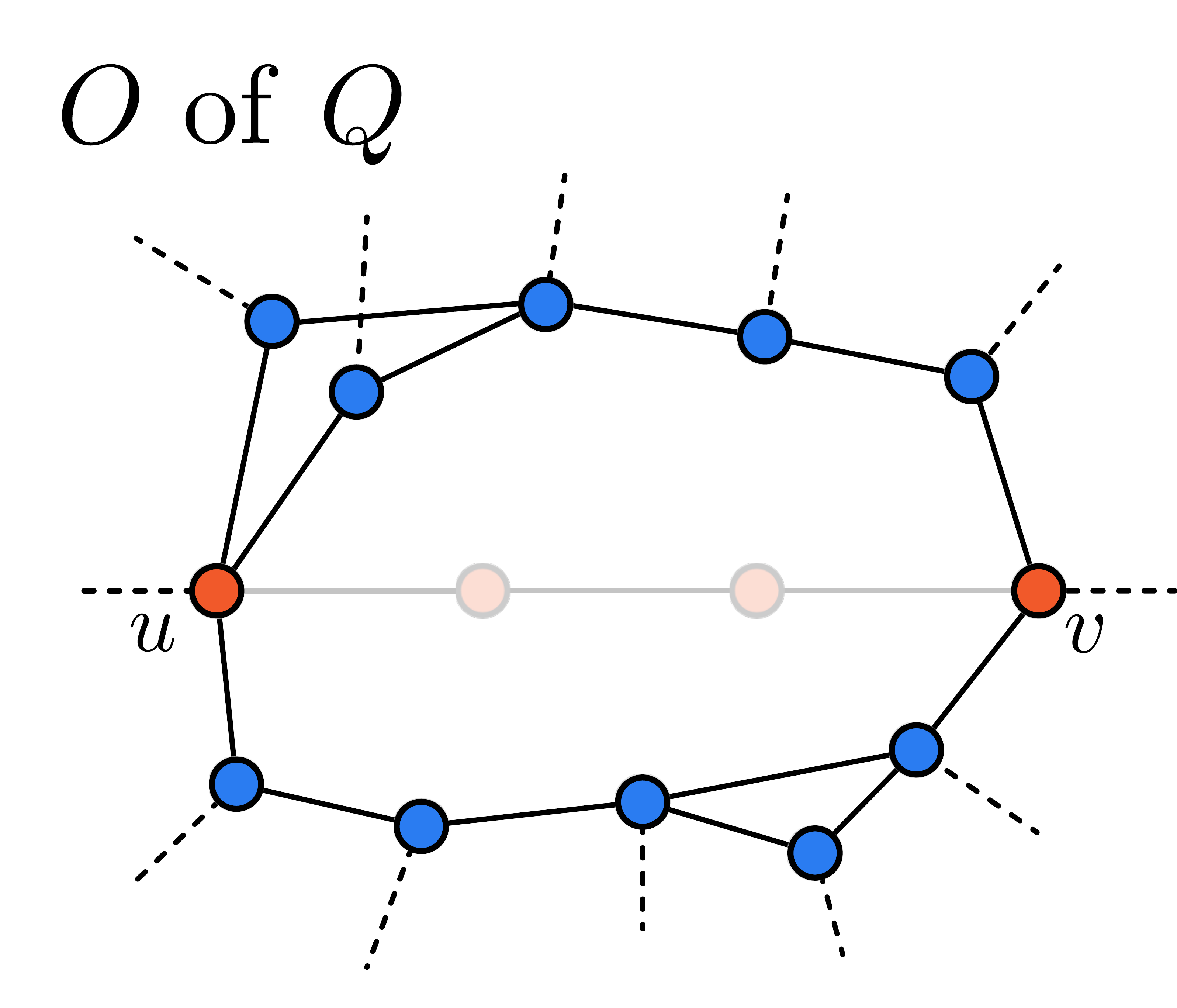}}
  \par\bigskip
  \subcaptionbox{Inner spin states\label{fig:inner-states}}[0.32\textwidth]{\includegraphics[width=0.32\textwidth]{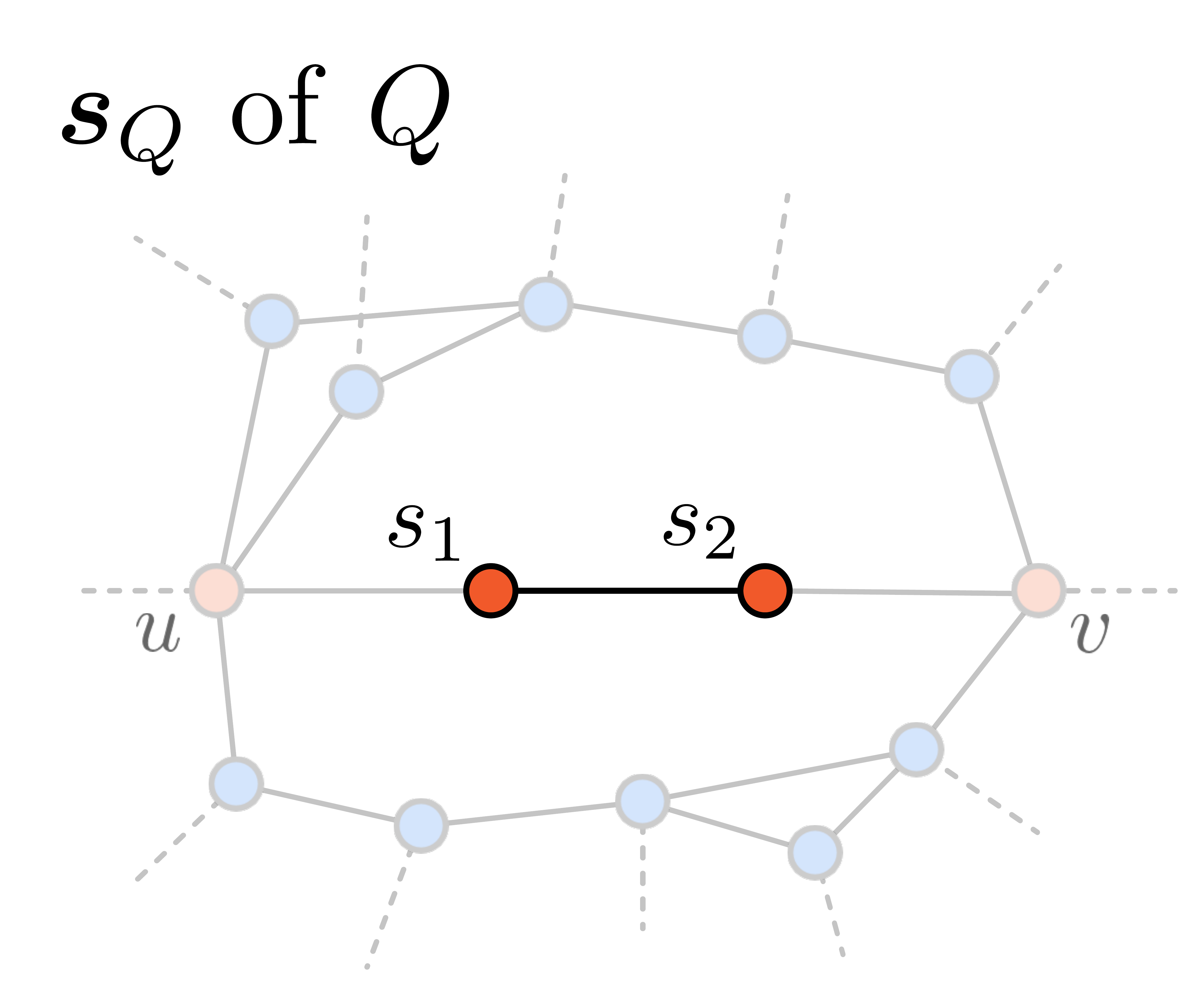}}
  \subcaptionbox{Outer spin states\label{fig:outer-states}}[0.32\textwidth]{\includegraphics[width=0.32\textwidth]{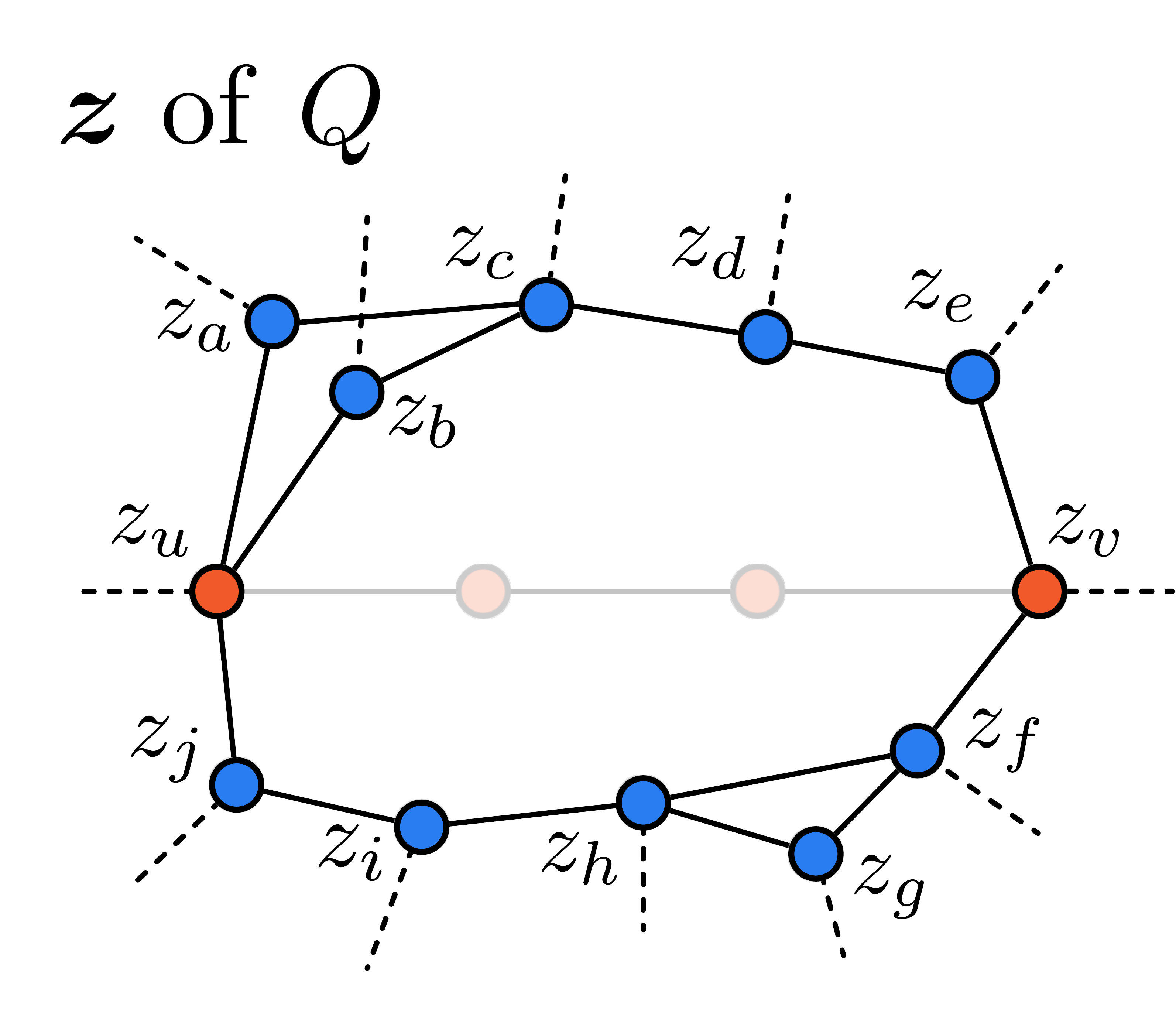}}
  \begin{minipage}[a]{0.32\textwidth}
    \centering
    \vspace{-1.70cm}
    \subcaptionbox{Satisfied\label{fig:satisfied}}[1\textwidth]{\includegraphics[scale=0.13]{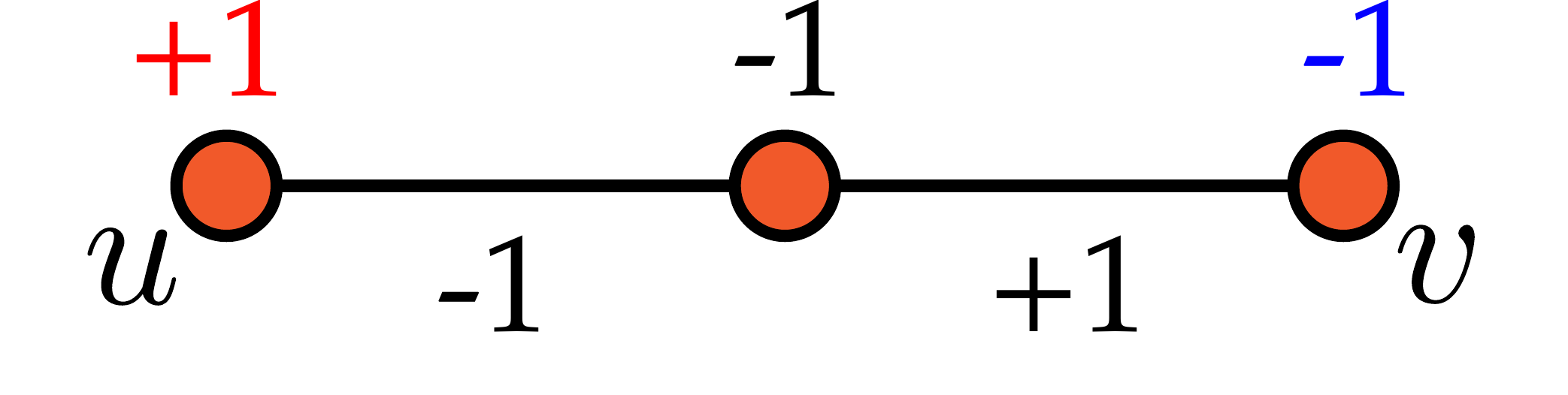}}
    \par\bigskip
    \subcaptionbox{Unsatisfied\label{fig:unsatisfied}}[1\textwidth]{\includegraphics[scale=0.13]{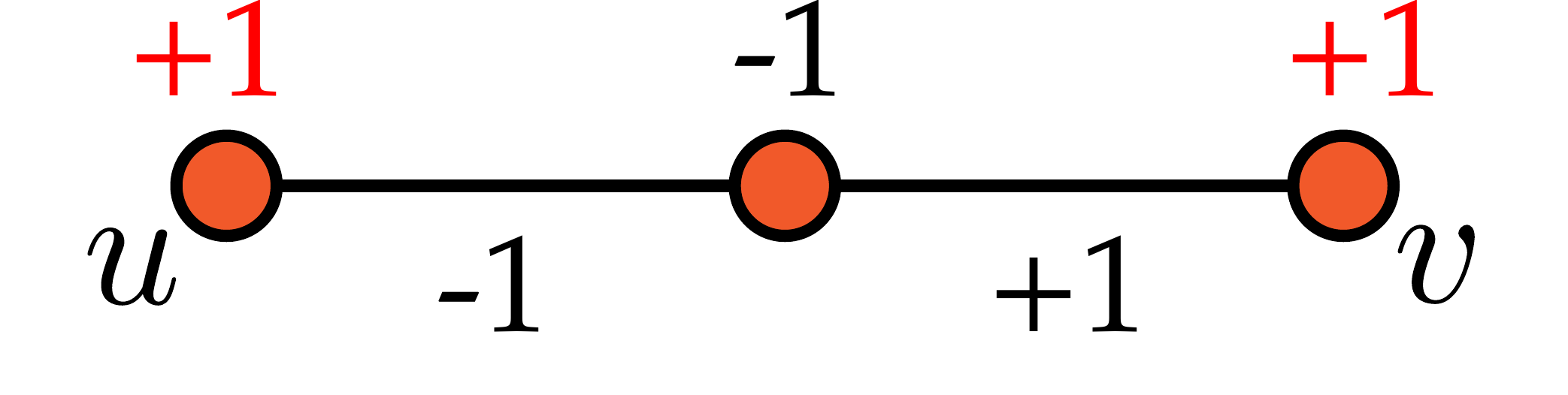}}
  \end{minipage}
  \caption{Examples of spin chain related definitions. \textbf{(a)}~A spin chain between vertices~$u$ and~$v$ shown in red. \textbf{(b)}~A trimmed path consisting of all non-leaf vertices of the path. \textbf{(c)}~An outer lattice consisting of all vertices in the lattice except for vertices of the trimmed path. \textbf{(d)}~The inner spin states are the spin states of the trimmed path. \textbf{(e)}~The outer spin states are the spin states of the outer lattice. \textbf{(f)}~An example of a satisfied spin chain. \textbf{(g)}~An example of an unsatisfied spin chain.}
    \label{fig:spin-chain-examples}
\end{figure}

\newpage

For the following definitions, let $Q$ be a spin chain of an Ising graph~$I$.

\begin{definition}[Inner and Outer Spin States -- Figures~\ref{fig:inner-states} and~\ref{fig:outer-states}]
  Let~$T$ and~$O$ be the trimmed path and outer graph of spin chain~$Q$ respectively. The spin states of~$T$ and~$O$ are called the \emph{inner and outer spin states} of~$Q$ respectively, denoted by~$\boldsymbol{s}_{Q}$ and~$\boldsymbol{z}$. Note that the spin states $\mathbf{Z}$ over the full Ising graph $I$ are equal to the outer spin states combined with the inner spin states, i.e. $\boldsymbol{Z} = \boldsymbol{z} \cup \boldsymbol{s}_{Q}$.
\end{definition}

\begin{definition}[Spin Chain Hamiltonian]
  A \emph{spin chain Hamiltonian} of a spin chain~$Q$, denoted~$H_{Q}$, is defined as its Ising Hamiltonian. 

In the case of a~ZZ~Ising Hamiltonian (i.e. with only 2-local interaction terms), the spin chain Hamiltonian of $Q$ can be written as
\begin{equation}
  H_{Q}(z_u, z_v, \boldsymbol{s}_{Q}) = -J_{u1}z_u s_1 - J_{12}s_1 s_{2} - \ldots - J_{(k-1)v}s_{k-1}z_v,
\end{equation}
where $u$ and $v$ are the leaf vertices,~$J_{ij}$ are the edge weights,~$s_i \in \boldsymbol{s}_{Q}$ are the inner spin states and $k$ is the number of edges in $Q$. 
\end{definition}

\begin{definition}[Satisfied/Unsatisfied -- Figures~\ref{fig:satisfied} and~\ref{fig:unsatisfied}]
  We say that a spin chain $Q$ is \emph{satisfied} when the spin states of the leaf vertices~$z_u$ and~$z_v$ are chosen such that inner spin states can be assigned to reach a ground state of the Ising Hamiltonian~$H_{Q}$ on~$Q$. That is,
  \begin{equation}
   (z_u, z_v) = \underset{(z^\prime_u, z^\prime_v)}{\text{argmin}} \left[ \min \limits_{\boldsymbol{s}_{Q}} [H_{Q}(z^\prime_u, z^\prime_v, \boldsymbol{s}_{Q})] \right].
  \end{equation}
The state of $Q$ is \emph{unsatisfied} or \emph{broken} otherwise. The act of choosing~$z_u, z_v \in \{-1,1\}$ such that~$Q$ is broken is called \emph{breaking}~$Q$.
The ground state energy of~$H_{Q}$ is defined to be
\begin{equation}
	E_{\text{Q}}^{\text{gs}} := \min \limits_{(z^\prime_u, z^\prime_v, \boldsymbol{s}_{Q})} [H_{Q}(z^\prime_u, z^\prime_v, \boldsymbol{s}_{Q})].
\end{equation}
\end{definition}\noindent

The goal for this section is to show how spin chains can be modified such that their outer spin states in the ground state remain unchanged. We will additionally show that these modifications do not change the energy gap between the ground and first excited states. This energy gap is called the \emph{breaking penalty}.

\begin{definition}[Breaking Penalty]
  Let $E_{\text{Q}}^{\text{gs}}$ be the ground state energy of spin chain Hamiltonian $H_{Q}$. The \emph{breaking penalty} of the corresponding spin chain $Q$, denoted $\text{BP}_{Q}$, is the smallest difference between $E_{\text{Q}}^{\text{gs}}$ and the minimum Ising Hamiltonian value among broken states of $Q$. That is,
  \begin{equation}
    \text{BP}_{Q} = \min \limits_{(z_u, z_v, \boldsymbol{s}_{Q})} [H_{Q}(z_u, z_v, \boldsymbol{s}_{Q})] - E_{\text{Q}}^{\text{gs}},
  \end{equation}
where $z_u$ and $z_v$ range over the spin states from unsatisfied states of $Q$.
\end{definition}

\begin{definition}[Equivalent]
  Let~$Q$ be a spin chain of an Ising graph~$I$ and~$Q^{\prime}$ be a spin chain of an Ising graph~$I^\prime$ such that their outer graphs are isomorphic, where the leaf vertices of~$Q$ are the image of the leaf vertices of~$Q^\prime$ under isomorphism.\\
  The two chains are \emph{equivalent} (denoted $Q \leftrightarrow Q^\prime$), if and only if both of the following conditions are satisfied:
  \begin{itemize}
  \item \textbf{Condition 1}: If $\boldsymbol{z}^\text{gs}$ and $\boldsymbol{z}^{\prime \text{gs}}$ are sets of distinct spin state configurations of outer graph $O$ in (possibly degenerate) ground states of $I$ and $I^\prime$ respectively, then 
    \begin{equation}
      \boldsymbol{z}^\text{gs} = \boldsymbol{z}^{\prime \text{gs}}.
    \end{equation}
  \item \textbf{Condition 2}: The breaking penalties for $Q$ and $Q^\prime$ are equal.
  \end{itemize}\label{def:equivalent}
\end{definition}

\begin{figure}[h!]
  \centering
  \includegraphics[scale=0.13]{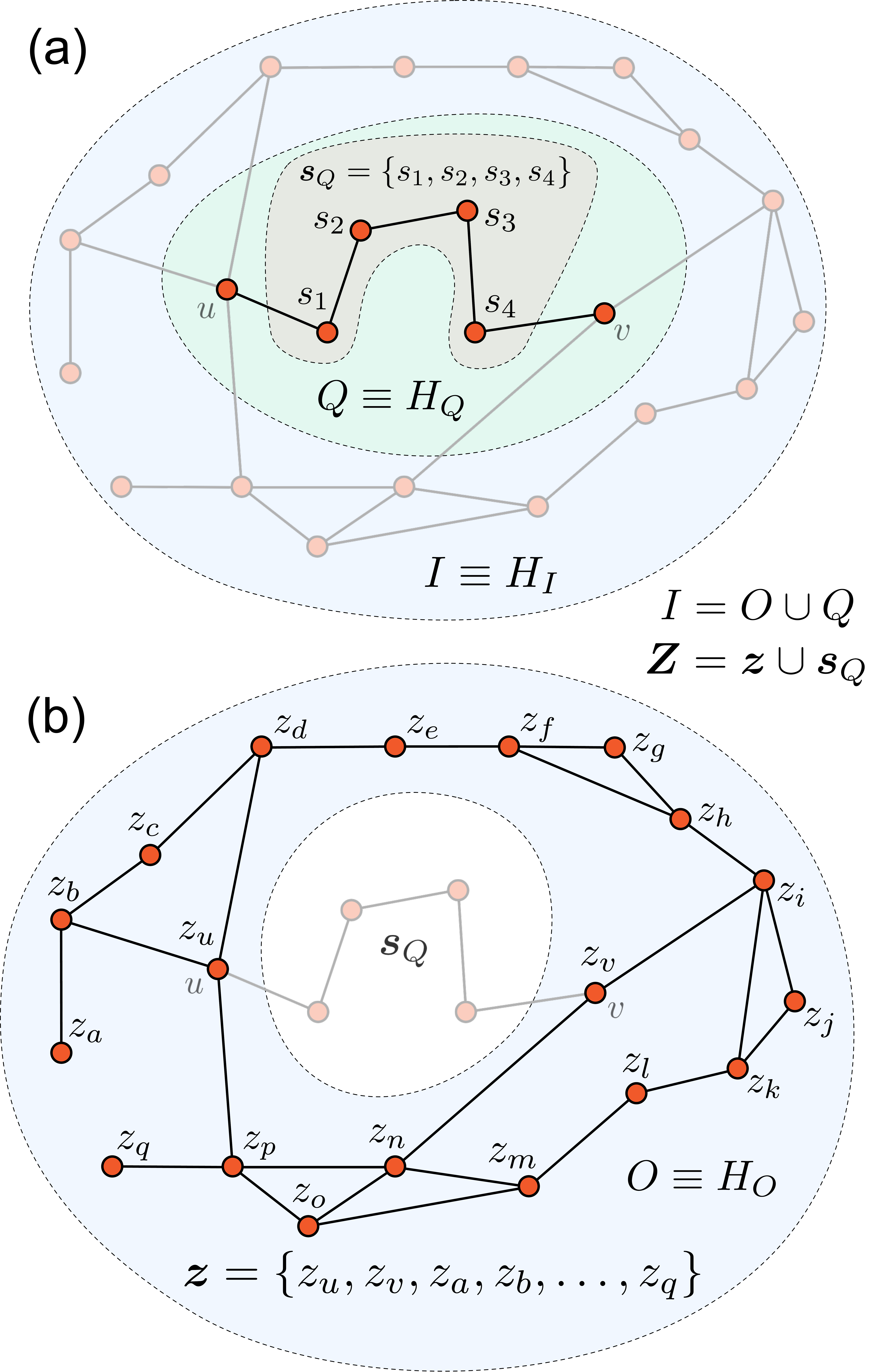}
  \caption{An example spin chain and outer graph, where ``$\equiv$" indicates equivalence between graph and Hamiltonian representations. Note that the full Ising graph is $I = O \cup Q$ and the spin states over $I$ is $\mathbf{Z} = \mathbf{z} \cup \mathbf{s}_Q$. \textbf{(a)}~A spin chain $Q$ with inner spin states $\mathbf{s}_Q$ of an Ising graph $I$. \textbf{(b)}~An outer graph $O$ of $Q$ with outer spin states $\mathbf{z}$.}
  \label{fig:spin-chain-combined-example}
\end{figure}

\begin{lemma} \label{theorem:qubit-chain-equivalence}
	Let~$H_Q$ and $H_{Q^\prime}$ be the spin chain Hamiltonians of spin chains $Q$ and $Q^\prime$ with isomorphic outer graphs where leaf vertices~$u$ and~$v$ of~$Q$ are the image of the leaf vertices of $Q^\prime$ under isomorphism.

  If there exists a constant $c$ such that for all $z_{u}, z_{v} \in \{-1, 1\}$,

  \begin{equation}\label{eq:assumption-1}
    \min\limits_{\boldsymbol{s}_{Q}} \left[ H_{Q}(z_{u}, z_{v}, \boldsymbol{s}_{Q}) \right] = \min\limits_{\boldsymbol{s}_{Q^\prime}} \left[ H_{Q^\prime} (z_{u}, z_{v}, \boldsymbol{s}_{Q^\prime})\right] + c,
  \end{equation}
  then $Q$ and $Q^{\prime}$ are equivalent as defined in Definition~\ref{def:equivalent}.
\end{lemma}

\begin{proof}
	Let $O$ be an Ising graph that is isomorphic to the outer graphs of~$Q$ and~$Q^\prime$ with Ising Hamiltonian~$H_O$. Let~$H_I$ and~$H_{I^\prime}$ be the~ZZ Ising Hamiltonians of~$I := O \cup Q$ and~$I^\prime := O \cup Q^\prime$ respectively. 
  We can write
  \begin{align*}
    (\boldsymbol{z}^\text{gs}, \boldsymbol{s}^\text{gs}_{Q}) &:= \underset{(\boldsymbol{z}, \boldsymbol{s}_{Q})}{\mathrm{argmin}} [H_I(\boldsymbol{z}, \boldsymbol{s}_{Q})], \text{ and}\\
    (\boldsymbol{z}^{\prime \text{gs}}, \boldsymbol{s}^{\text{gs}}_{Q^\prime}) &:= \underset{(\boldsymbol{z}, \boldsymbol{s}_{Q^\prime})}{\mathrm{argmin}}[H_{I^\prime}(\boldsymbol{z}, \boldsymbol{s}_{Q^\prime})],
  \end{align*}
  where~$\boldsymbol{s}_{Q}$ and~$\boldsymbol{s}_{Q^\prime}$ are the inner spin states of~$Q$ and~$Q^\prime$, and~$\boldsymbol{z}$ and~$\boldsymbol{z}^\prime$ are their respective outer spin states.
  
  Assume that there exists a constant $c$ such that for all~$z_u, z_v \in \{-1, 1\}$, 
\begin{equation}
\min\limits_{\boldsymbol{s}_{Q}} \left[ H_{Q}(z_u, z_v, \boldsymbol{s}_{Q}) \right] = \min\limits_{\boldsymbol{s}_{Q^\prime}} \left[ H_{Q^\prime} (z_u, z_v, \boldsymbol{s}_{Q^\prime})\right] + c.\label{eq-lemma-assumption}
\end{equation}  
  To Show:
  \begin{align*}
    \text{\textbf{Condition 1:}} &~\boldsymbol{z}^\text{gs} = \boldsymbol{z}^{\prime \text{gs}}.\\
    \text{\textbf{Condition 2:}} &\text{ If BP}_{Q}\text{ is the breaking penalty of }Q\text{ and}~\text{BP}_{Q^\prime}\text{ is}\\
                                 &\text{ the breaking penalty of}~Q^\prime \text{ then}~\text{BP}_{Q} = \text{BP}_{Q^\prime}.
  \end{align*}
  \textbf{Condition 1}: We observe that $H_I = H_O + H_Q$ and $H_{I^\prime} = H_O + H_{Q^\prime}$, where $H_O$ is the ZZ Ising Hamiltonian of the outer graph $O$. Thus by assumption Equation~\ref{eq-lemma-assumption}
  \begin{align*}
	\min \limits_{\boldsymbol{s}_{Q}} \left[ H_I(\boldsymbol{z}, \boldsymbol{s}_{Q}) \right] &= \min \limits_{\boldsymbol{s}_{Q}} \left[ H_O(\boldsymbol{z}) + H_Q(z_u, z_v, \boldsymbol{s}_{Q}) \right] \\
	&= H_O(\boldsymbol{z}) + \min \limits_{\boldsymbol{s}_{Q}} \left[H_Q(z_u, z_v, \boldsymbol{s}_{Q}) \right]\\
	&= H_O(\boldsymbol{z}) + \min \limits_{\boldsymbol{s}_{Q^\prime}} \left[H_{Q^\prime}(z_u, z_v, \boldsymbol{s}_{Q^\prime}) + c\right]\\
	&= \min \limits_{\boldsymbol{s}_{Q^\prime}} \left[ H_{I^\prime}(\boldsymbol{z}, \boldsymbol{s}_{Q^\prime}) \right] + c
  \end{align*}
Thus,
    \begin{align*}
      \boldsymbol{z}^\text{gs} &= \underset{\boldsymbol{z}}{\mathrm{argmin}} \left[ \min \limits_{\boldsymbol{s}_{Q}} \left[ H_I(\boldsymbol{z}, \boldsymbol{s}_{Q}) \right] \right] = \underset{\boldsymbol{z}}{\mathrm{argmin}} \left[ \min \limits_{\boldsymbol{s}_{Q^\prime}} \left[ H_{I^\prime}(\boldsymbol{z}, \boldsymbol{s}_{Q^\prime}) \right] + c \right]\\
                                &= \boldsymbol{z}^{\prime \text{gs}}.
    \end{align*}
  \\
  \textbf{Condition 2}: Let $(z_u, z_v)$ range over $\{-1, 1\}^2$ and let $(z^*_u, z^*_v)$ range over the spin states from unsatisfied states of~$Q$,
    \begin{align*}
      \text{BP}_{Q} &= \min \limits_{(z^*_u, z^*_v, \boldsymbol{s}_{Q})} [H_{Q}(z^*_u, z^*_v, \boldsymbol{s}_{Q})] - E_{\text{Q}}^{\text{gs}}\\
      &= \min \limits_{(z^*_u, z^*_v, \boldsymbol{s}_{Q})} [H_{Q}(z^*_u, z^*_v, \boldsymbol{s}_{Q})] - \min \limits_{(z_u, z_v, \boldsymbol{s}_{Q})}[H_{Q}(z_u, z_v, \boldsymbol{s}_{Q})] \\
                     &= \min \limits_{(z^*_u, z^*_v, \boldsymbol{s}_{Q^\prime})} [H_{Q^\prime}(z^*_u, z^*_v, \boldsymbol{s}_{Q^\prime})] + c - \min \limits_{(z_u, z_v, \boldsymbol{s}_{Q^\prime})} [H_{Q^\prime}(z_u, z_v, \boldsymbol{s}_{Q^\prime})] - c \\
                     &= \min \limits_{(z^*_u, z^*_v, \boldsymbol{s}_{Q^\prime})} [H_{Q^\prime}(z^*_u, z^*_v, \boldsymbol{s}_{Q^\prime})] + c - E^{\text{gs}}_{\text{Q}^\prime} - c \\
                     &= \text{BP}_{Q^\prime},
    \end{align*}
where $E_{\text{Q}}^{\text{gs}}$ and $E^{\text{gs}}_{\text{Q}^\prime}$ are the ground state energies of $H_{Q}$ and $H_{Q^\prime}$.\\

Therefore the spin chains $Q$ and $Q^\prime$ are equivalent.
\end{proof}

The following theorem describes how any spin chain $Q$ without local fields is equivalent to an effective spin chain consisting of a single edge with coupling strength equal to the spin chain $Q$'s effective coupling strength as defined in~Definition~\ref{def:effective-coupling-strength}. A consequence of this is that any two spin chains are equivalent when they share the same effective coupling strength.

\begin{theorem}[Equivalence with Effective Spin Chain]
  Let $Q$ be a spin chain with leaf vertices $u$ and $v$, $k$ couplings, $p \leq k$ of which are negative, so that its corresponding Hamiltonian is
  \begin{align}
    H_{Q}(z_u,z_v,\mathbf{s}_{Q}) = -J_{u1}z_us_1 - J_{12}s_1s_2 - \dots - J_{(k-1)v} s_{k-1}z_v. \label{eq:thm1hamil}
  \end{align}
Then $Q$ is equivalent to an effective spin chain $Q_\mathrm{eff}$ defined by the Ising Hamiltonian 
\begin{equation}
H_{Q_\mathrm{eff}}(z_u, z_v) = -J^\mathrm{eff}_{Q} z_u z_v,
\end{equation}
where $J^\mathrm{eff}_{Q}$ is the effective coupling strength of $Q$ as defined in Definition~\ref{def:effective-coupling-strength}.

\end{theorem}

\begin{proof}
  We aim to show that
  \begin{align}
    \min_{\mathbf{s}_Q} H_Q(z_u,z_v,\mathbf{s}_Q) = \left( \min_{ij}|J_{ij}| - \sum_{ij} |J_{ij}|\right)+H_{\mathrm{eff}}(z_u,z_v)
  \end{align}
  which by Lemma \ref{theorem:qubit-chain-equivalence} would imply that $Q$ and $Q_{\mathrm{eff}}$ are equivalent spin chains. Note that $z_uz_v \in \lbrace -1,+1\rbrace$, so an expression for the minimum energy must be found for both cases. For convenience, define $s_0 := z_u$ and $s_k := z_v$, so that Eq.~\ref{eq:thm1hamil} may be written as
  \begin{align}
    H_Q(z_u, z_v, \mathbf{s}_Q) = -\sum_{i=1}^k J_{(i-1)i}s_{i-1}s_i
  \end{align}

  \begin{itemize}
  \item \textbf{Case 1}: All edges satisfied. If we assign spins along the spin chain recursively as $s_i = \text{sign}(J_{(i-1)i})s_{i-1}$ for $1 \leq i \leq k$, leaving $s_0=z_u$ undetermined, then all edges along the spin chain are satisfied, since
    \begin{align}
      -J_{(i-1)i} s_{i-1}s_i = -J_{(i-1)i} \text{sign}(J_{(i-1)i})(s_{i-1})^2 =  - |J_{(i-1)i}|.
    \end{align}
    So the spin chain attains its global minimum energy
    \begin{align}
      E_Q^{\mathrm{gs}} = -\sum_{ij} |J_{ij}|,
    \end{align}
    and the relation between $z_v$ and $z_u$ is,
    \begin{align}
      z_v = \left(\prod_{ij}\text{sign}(J_{ij})\right)z_u = (-1)^p z_u. \label{eq:parity}
    \end{align}

    \
    
  \item \textbf{Case 2:} Smallest magnitude edge unsatisfied. If we start from the previous spin assignment but break a single edge along the spin chain by picking $i'$ such that $|J_{(i'-1)i'}| = \min_{ij} |J_{ij}|$ and setting
    \begin{align}
      s_i &= \text{sign}(J_{(i-1)i})s_{i-1},\ \ \ \ \ \ \ \ \ i \neq i' \\
      s_{i'} &= -\text{sign}(J_{(i'-1)i'})s_{i'-1},
    \end{align}
    then all edges are still satisfied except for $J_{(i'-1)i'}$, which will contribute $(+\min_{ij}|J_{ij}|)$ to the energy. With this assignment, the energy of the spin chain is
    \begin{align}
      E = E_{Q}^{\mathrm{gs}} + 2\min_{ij}|J_{ij}|, \label{eq:minunsat}
    \end{align}
    and the relation between $z_v$ and $z_u$ is $z_v = -(-1)^p z_u$. Note that by construction, the assignment in Case 1: $s_i = \text{sign}(J_{(i-1)i})s_{i-1}\ ,\ 1 \leq i \leq k$ is the {\em only} spin assignment that can simultaneously satisfy every edge, and it implies Equation~\ref{eq:parity}. Thus, by the contrapositive, if $z_v \neq (-1)^p z_u$, then the spin assignment must break at least one edge, and since we have chosen the edge with the minimum coupling strength, $E_Q^{\mathrm{gs}} + 2\min_{ij}|J_{ij}|$ must be the minimum energy attainable with $z_v = -(-1)^p z_u$.
  \end{itemize}

  These two cases cover all possible assignments to $(z_u,z_v)$, and can be summarized as
  \begin{align}
    \min_{\mathbf{s}_Q} H_Q(z_u,z_v,\mathbf{s}_Q) &= \begin{cases} \displaystyle -\sum_{ij} |J_{ij}| & z_v = (-1)^p z_u \\ \displaystyle - \sum_{ij} |J_{ij}| + 2\min_{ij} |J_{ij}| & z_v = -(-1)^p z_u\end{cases} \\
      &= -\sum_{ij} |J_{ij}| + \min_{ij}|J_{ij}| - (-1)^p \min_{ij}|J_{ij}|z_uz_v \\
      &= \left(\min_{ij}|J_{ij}|-\sum_{ij} |J_{ij}|\right) + H_{\mathrm{eff}}(z_u,z_v)
  \end{align}
  So by Lemma \ref{theorem:qubit-chain-equivalence}, $Q$ and $Q_{\mathrm{eff}}$ are equivalent spin chains.
\end{proof}

This theorem allows one to verify that two spin chains are equivalent by showing that their effective coupling strengths are equal. In particular, it can be used to prove the following properties of spin chains.
{\setlength{\belowcaptionskip}{-15pt}
  \begin{figure}[h!]
    \centering
    \includegraphics[scale=0.60]{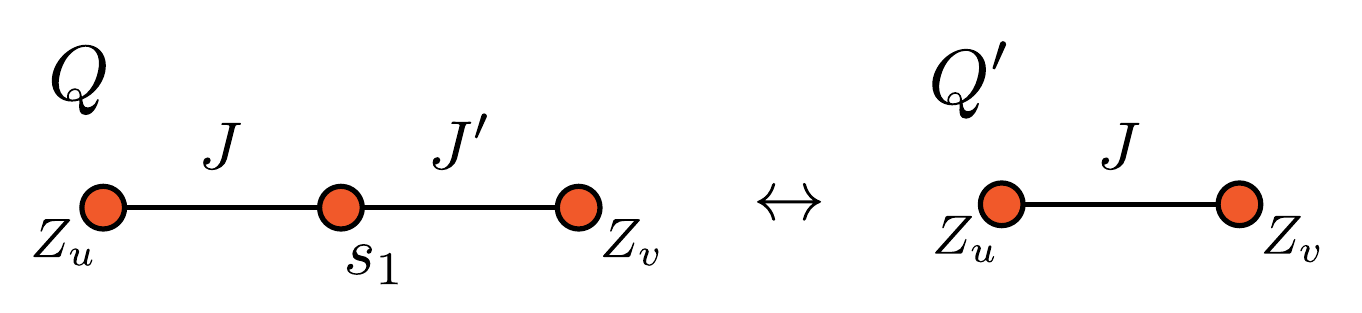}
    \caption{Spin chain identity property. A spin chain with a single coupling of strength $J$ and another with two couplings of strengths $J$ and $J^\prime$, where $J^\prime \geq |J|$, are equivalent under Definition~\ref{def:equivalent} (denoted ``$\leftrightarrow$'').}
    \label{fig:identity-edges}\par\medskip
  \end{figure}}\noindent
\begin{corollary}[Identity Property -- Figure~\ref{fig:identity-edges}] \label{theorem:chain-extension}
  Let~$Q$ and~$Q^\prime$ be spin chains with corresponding spin chain Hamiltonians
  \begin{align}
    H_{Q}(z_u, z_v, \boldsymbol{s}_{Q}) &= -J z_u s_1 - J^\prime  s_1 z_v, \text{ and}\\
    H_{Q^\prime}(z_u, z_v, \boldsymbol{s}_{Q^\prime}) &= -J z_u z_v,
  \end{align}
  where $J^\prime \geq |J|$. Then $Q$ and $Q^\prime$ are equivalent.
\end{corollary}
\begin{proof}
  For both $H_{Q}$ and $H_{Q^\prime}$, the effective coupling strength is $J_{\mathrm{eff}} = J$.
\end{proof}

  \begin{figure}[h!]
    \centering
    \includegraphics[scale=0.60]{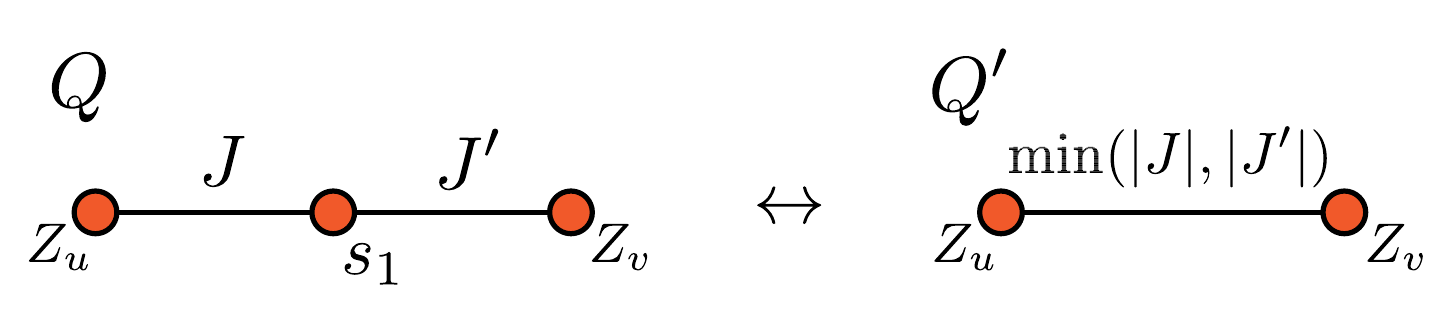}
    \caption{Spin chain inverse property. A spin chain with two couplings of strengths $J$ and $J^\prime$ and another with a single coupling of strength $\min{(|J|,|J^\prime|)}$, where $J$, $J^\prime \leq 0$, are equivalent under Definition~\ref{def:equivalent} (denoted ``$\leftrightarrow$'').}
    \label{fig:inverse-edges}
    \par\medskip
  \end{figure}\noindent
\begin{corollary}[Inverse Property -- Figure~\ref{fig:inverse-edges}] \label{theorem:inverse-edges}
  Let~$Q$ and~$Q^\prime$ be spin chains with corresponding spin chain Hamiltonians
  \begin{align}
    H_{Q}(z_u, z_v, \boldsymbol{s}_{Q}) &= -J z_u s_1 - J^\prime s_1 z_v, \text{ and}\\
    H_{Q^\prime}(z_u, z_v, \boldsymbol{s}_{Q^\prime}) &= -\min(|J|,|J^\prime|) z_u z_v,
  \end{align}
  where $J,J^\prime \leq 0$, then~$Q$ and~$Q^\prime$ are equivalent.
\end{corollary}
\begin{proof}
  For both $H_{Q}$ and $H_{Q^\prime}$, the effective coupling strength is $J_{\mathrm{eff}}= \min(|J|,|J^\prime|)$.
\end{proof}

  \begin{figure}[h!]
    \centering
    \includegraphics[scale=0.60]{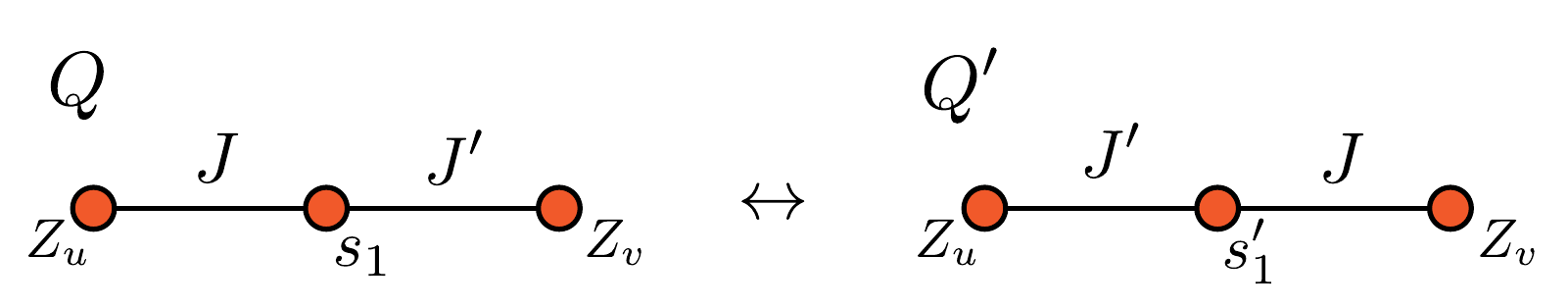}
    \caption{Spin chain commutativity property. A spin chain with two couplings of strengths $J$ and $J^\prime$ and another with two couplings of reversed strengths are equivalent under Definition~\ref{def:equivalent} (denoted ``$\leftrightarrow$'').}
    \label{fig:commuting-edges}\par\medskip
  \end{figure}\noindent
\begin{corollary}[Commutativity Property -- Figure~\ref{fig:commuting-edges}]
\label{theorem:commuting-edges}
  Let~$Q$ and~$Q^\prime$ be spin chains with corresponding spin chain Hamiltonians
  \begin{align}
    H_{Q}(z_u, z_v, \boldsymbol{s}_{Q}) &= -J z_u s_1 - J^\prime s_1 z_v, \text{ and}\\
    H_{Q^\prime}(z_u, z_v, \boldsymbol{s}_{Q^\prime}) &= -J^\prime z_u s^\prime_1 - J s^\prime_1 z_v,
  \end{align}
  then $Q$ and $Q^\prime$ are equivalent.
\end{corollary}
\begin{proof}
  For both $H_{Q}$ and $H_{Q^\prime}$, the effective coupling strength is
  \begin{align}
    J_{\mathrm{eff}} = \text{sign}(J)\text{sign}(J^\prime)\min(|J|,|J^\prime|).
  \end{align}
\end{proof}

\begin{figure}[htb!]
    \centering
    \includegraphics[scale=0.60]{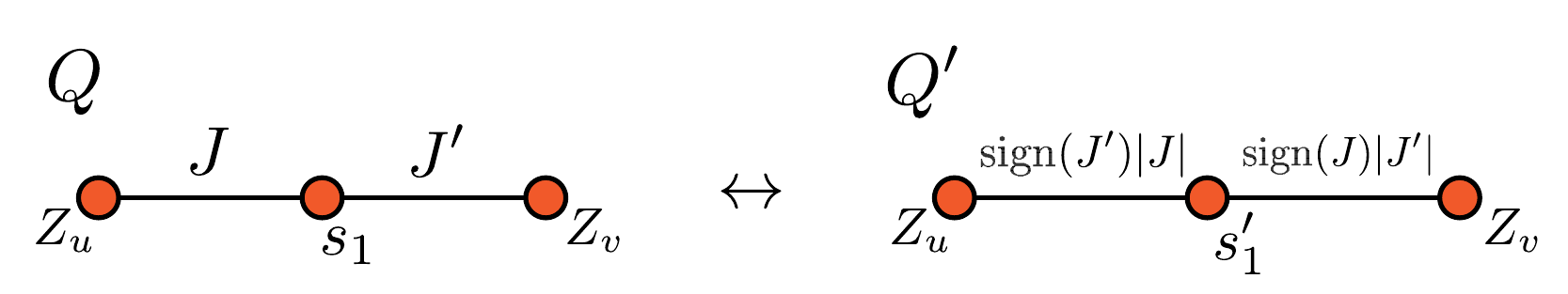}
    \caption{Spin chain parity commutativity property. A spin chain with two couplings of strengths $J$ and $J^\prime$ and another with couplings of strengths $\text{sign}(J^\prime)|J|$ and $\text{sign}(J)|J^\prime|$ are equivalent under Definition~\ref{def:equivalent} (denoted ``$\leftrightarrow$'').}
    \label{fig:commuting-parity-edges}\par\medskip
  \end{figure}
  
\begin{corollary}[Parity Commutativity Property -- Figure~\ref{fig:commuting-parity-edges}] \label{theorem:commuting-parity-edges}
  Let~$Q$ and~$Q^\prime$ be spin chains with corresponding spin chain Hamiltonians
  \begin{align}
    H_{Q}(z_u, z_v, \boldsymbol{s}_{Q}) &= -J z_u s_1 - J^\prime s_1 z_v, \text{ and}\\
    H_{Q^\prime}(z_u, z_v, \boldsymbol{s}_{Q^\prime}) &= - \mathrm{sign}(J^\prime)|J| z_u s^\prime_1 - \mathrm{sign}(J)|J^\prime| s^\prime_1 z_v,
  \end{align}
  then $Q$ and $Q^\prime$ are equivalent.
\end{corollary}
\begin{proof}
  For both $H_{Q}$ and $H_{Q^\prime}$, the effective coupling strength is
  \begin{align}
    J_{\mathrm{eff}} = \text{sign}(J)\text{sign}(J^\prime)\min(|J|,|J^\prime|).
  \end{align}
\end{proof}

\newpage

\section{A proof of the correctness of degree reduction} \label{sec:proof-of-degree-reduction}
\renewcommand\thefigure{\thesection.\arabic{figure}}
\setcounter{figure}{0} 
\setcounter{definition}{0} 
\setcounter{lemma}{0}
\setcounter{theorem}{0}
\setcounter{corollary}{0}
In this section, we will proceed to prove the correctness of complete degree reduction described in Sec.~\ref{QUBO}, which enables an arbitrary problem graph to be 
reduced to an input graph suitable for subdivision-embedding. For the following definitions, let~$G$ be an induced subgraph 
of an Ising graph~$I$ and let~$V(\cdot)$ and~$E(\cdot)$ denote vertex and edge sets of a graph respectively, and let an edge $e$ incident to vertices $u$ and $v$ be denoted by $uv$.

\begin{figure}[h!]
  \centering
  \subcaptionbox{Induced subgraph\label{fig:subgraph}}[0.32\textwidth]{\includegraphics[width=0.32\textwidth]{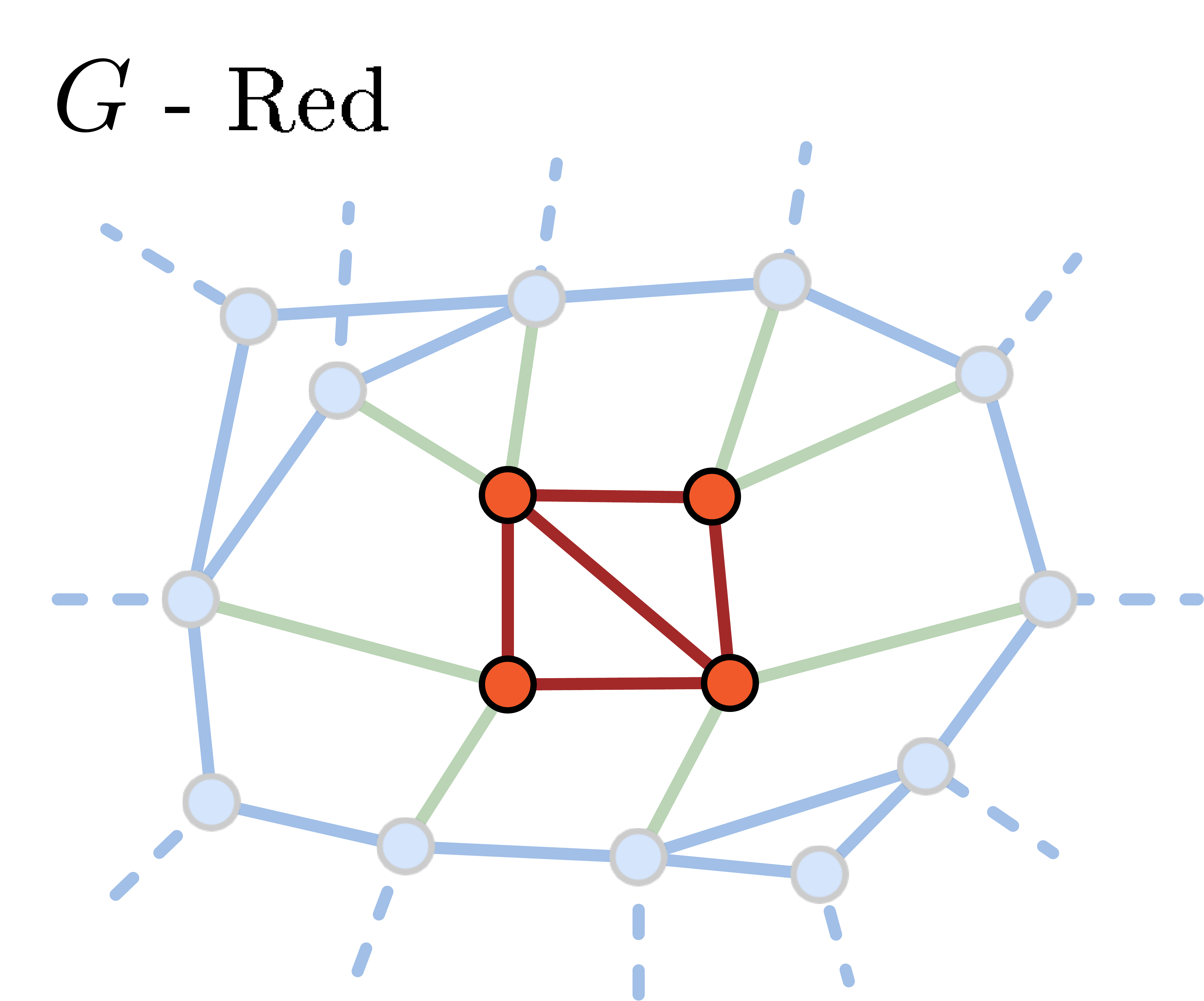}}
  \subcaptionbox{Environment of $G$\label{fig:environment}}[0.32\textwidth]{\includegraphics[width=0.32\textwidth]{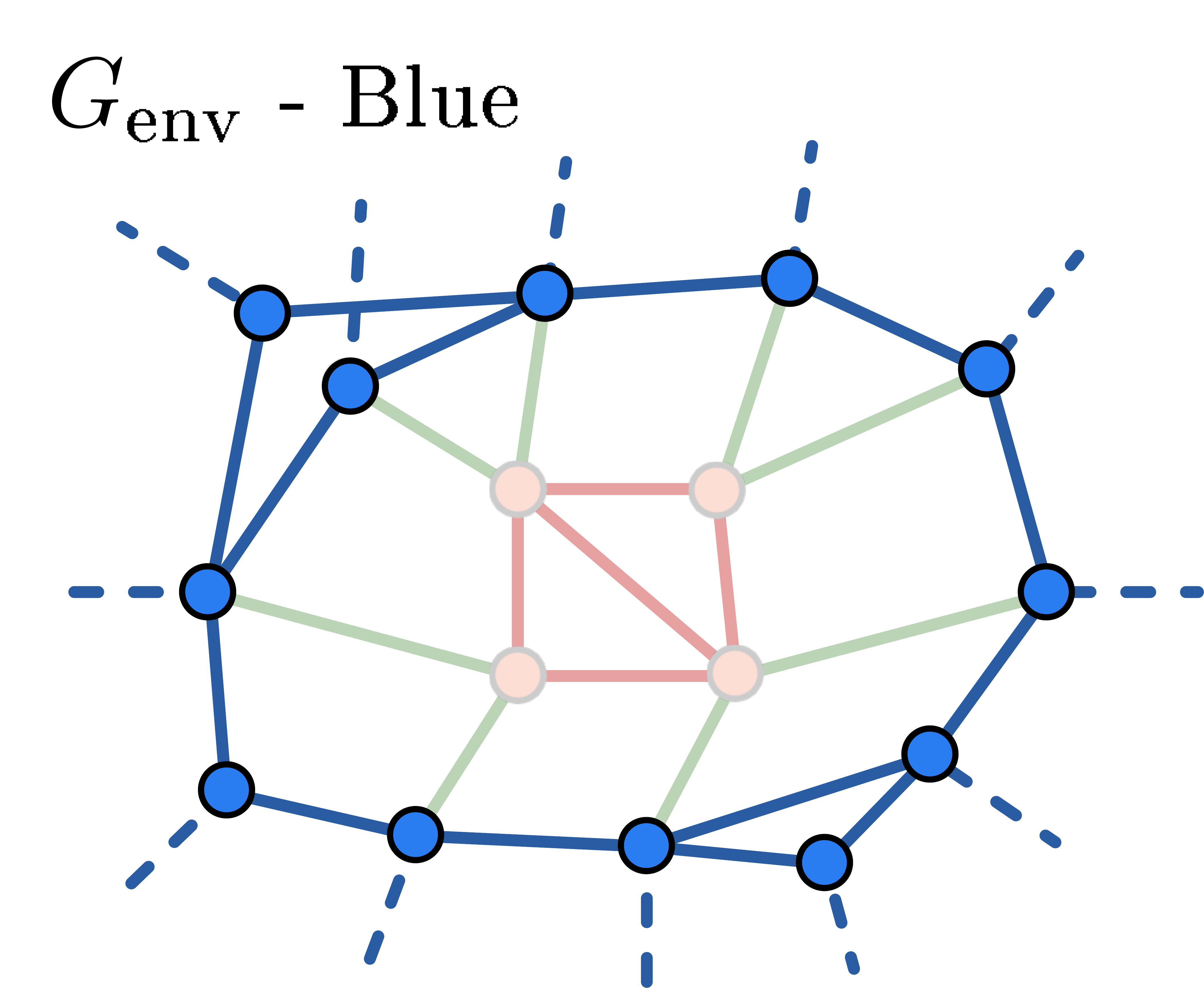}}
  \subcaptionbox{External edges of $G$\label{fig:adjacent-edges}}[0.32\textwidth]{\includegraphics[width=0.32\textwidth]{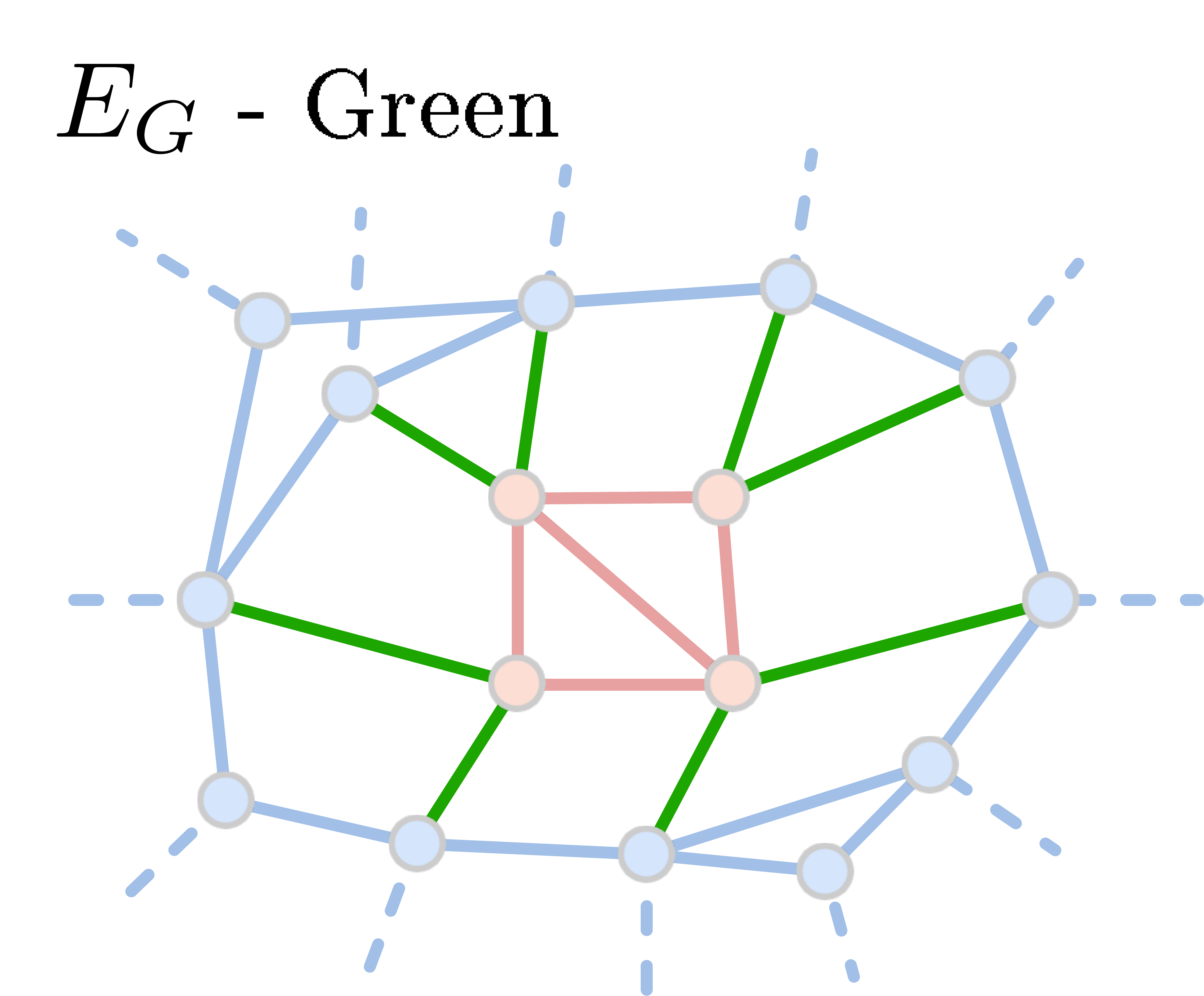}}
  \caption{\textbf{(a)}~An induced subgraph~$G$ (shown in red) of an Ising graph~$I$ is a vertex subset of~$I$ along with
  all edges of~$I$ which have both endpoints in $G$. \textbf{(b)}~The environment of~$G$ (shown in blue) is an Ising graph 
  consisting of all vertices of~$I$ that are not in~$G$ and all edges of~$I$ with no endpoint in~$G$. \textbf{(c)}~The 
  external edges of~$G$ (shown in green) are the edges of $I$ that have only a single endpoint within $G$. \label{fig:subgraph-definitions}}
\end{figure}

\begin{definition}[Environment ($G_{\text{env}}$) -- Figure~\ref{fig:environment}]
  The \emph{environment} of $G$ with respect to an Ising graph $I$ is the induced subgraph of $I$ defined on the vertex set $V(I) - V(G)$ and is denoted by $G_\text{env}$. \label{def:environment}
\end{definition}

\begin{definition}[External Edges -- Figure~\ref{fig:adjacent-edges}]
  The \emph{external edges} of~$G$ with respect to an Ising graph~$I$, denoted $E_G$, are defined as the 
  set of edges~$\{uv \in E(I) \;|\; u \in V(G) \text{ and } v \in V(G_\text{env})\}$.
  The \emph{external edges of a vertex} $u$ of $G$, denoted $E_u$, are the edges of $E_G$ that include $u$ as an endpoint.
\end{definition}

\begin{definition}[Internal Edges]
  The edges of $G$ are called the \emph{internal edges} of $G$, in contrast to external edges.
\end{definition}

When minimising the energy of an Ising graph, spin states of any induced subgraph can be influenced by its environment. We formalise this concept with the following definitions.

\begin{figure}[h!]
  \centering
  \subcaptionbox{Edge influence\label{fig:edge-influence}}[0.45\textwidth]{\hspace{1.2cm}\includegraphics[width=0.35\textwidth]{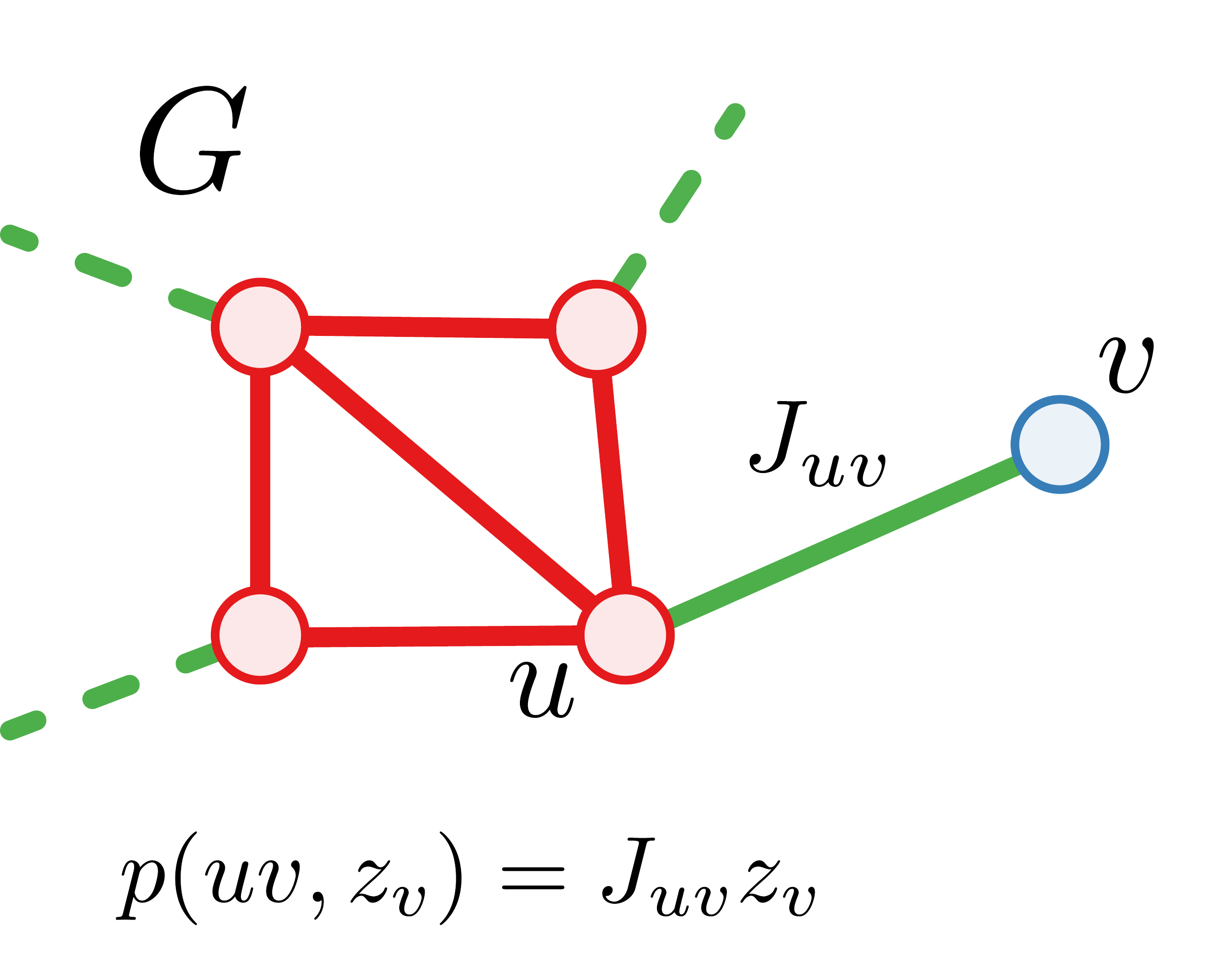}}
  \subcaptionbox{Single vertex example\label{fig:influence-example}}[0.45\textwidth]{\includegraphics[width=0.35\textwidth]{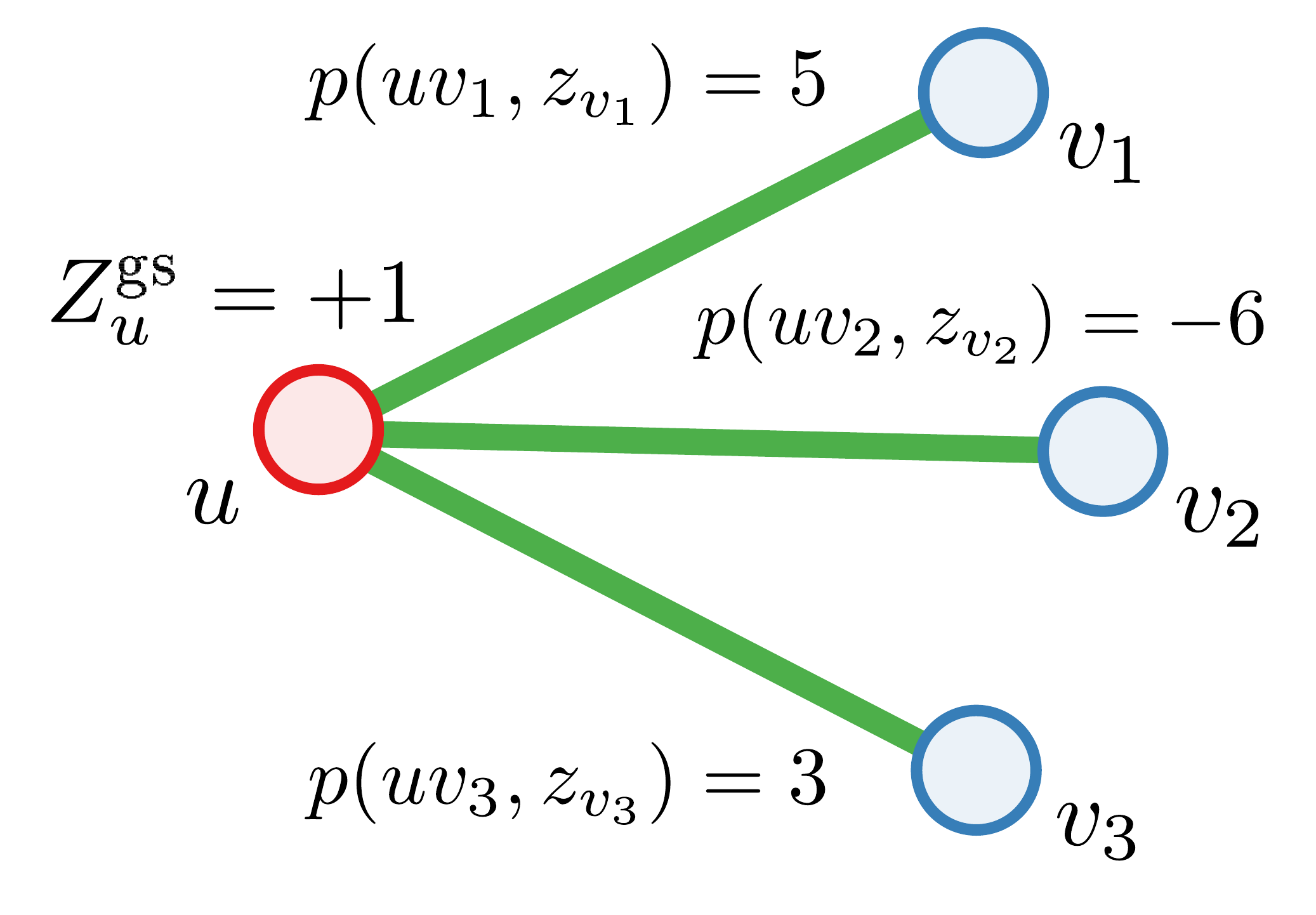}}
  \caption{Graphics to illustrate how edge influence concepts are calculated. \textbf{(a)}~An example of how to calculate the edge influence $p(uv, z_v)$. Let~$uv$ be an external edge of~$G$ and~$z_v$ be the spin state of the environment vertex~$v$, then~$p(uv, z_v)$ is $J_{uv}$ multiplied by $z_v$.
 \textbf{(b)}~An example of a single vertex logical device in the ground state. The net influence of~$u$
  is~$P_u(z_{v_1}, z_{v_2}, z_{v_3}) = p(uv_1, z_{v_1}) + p(uv_2, z_{v_2}) + p(uv_3, z_{v_3}) = 5 - 6 + 3 = 2$ resulting in a positive net majority influence. Hence the majority sign is $+1$ and the minority sign is $-1$ which means the majority edges are~$\{v_1, v_3\}$ and the minority edge is~$\{v_2\}$. Thus the net majority influence is~$P_u^{\text{maj}}(z_{v_1}, z_{v_2}, z_{v_3}) = 8$ and the net minority influence is~$P_u^{\text{min}}(z_{v_1}, z_{v_2}, z_{v_3}) = -6$. Since there is a positive net majority influence and the logical device consists of only a single vertex, the spin state $Z_u$ equal to $+1$ minimises the Ising energy.\label{fig:influence}}
\end{figure}

\begin{definition}[Edge Influence ($p$) -- Figure~\ref{fig:edge-influence}]
  Let~$uv$ be an external edge of $G$ with coupling strength $J_{uv}$, such that~$v$ is a vertex of the environment 
  with spin state $z_v$. The \emph{edge influence} $p$ of the edge~$uv$ with respect to~$G$ is defined as
  \begin{equation}
  p(uv, z_v) := J_{uv} z_v \in \{J_{uv}, -J_{uv}\}.
  \end{equation}
\end{definition}
\begin{definition}[Net Influence ($P_G$) -- Figure~\ref{fig:influence-example}]
  The \emph{net influence} of $G$ is the sum over the edge influences of $G$, that is,
    \begin{equation}
      P_G(\boldsymbol{z}) := \sum_{uv \in E_G} p(uv, z_v),
    \end{equation} 
    where $E_G$ are the external edges of $G$, $z_v \in \boldsymbol{z}$ are the spin states of the environment of $G$ and $p(e)$ is the edge influence of $e$ with respect to $G$. Similarly the \emph{net influence of a vertex} $u$ is defined as
    \begin{equation}
      P_u(\boldsymbol{z}) := \sum_{uv \in E_u} p(uv, z_v),
    \end{equation} 
    where $E_u$ are the external edges of $u$.
\end{definition}
\begin{definition}[Majority and Minority Sign -- Figure~\ref{fig:influence-example}]
Let $P_G(\boldsymbol{z})$ be the net influence of $G$. If $P_G(\boldsymbol{z}) \neq 0$ then~$\text{sign}[P_G(\boldsymbol{z})]$ is the \emph{majority sign} and~$-\text{sign}[P_G(\boldsymbol{z})]$ is the \emph{minority sign}. If $P_G(\boldsymbol{z})=0$ then there is no majority or minority sign.
\end{definition}

\begin{definition}[Majority and Minority Edges -- Figure~\ref{fig:influence-example}]
Let $P_G(\boldsymbol{z})$ be the net influence of $G$ where $P_G(\boldsymbol{z}) \neq 0$. 
The \emph{majority and minority edges} of $G$, denoted $E_G^\text{maj}(\boldsymbol{z})$ and $E_G^\text{min}(\boldsymbol{z})$, are the sets of edges where the signs of their edge 
influences are the majority and minority signs respectively. That is,
\begin{align*}
E_G^\text{maj}(\boldsymbol{z}) &:= \{uv \in E_G \;|\; \text{sign}[p(uv, z_v)] = \text{sign}[P_G(\boldsymbol{z})] \}, \text{ and}\\
E_G^\text{min}(\boldsymbol{z}) &:= \{uv \in E_G \;|\; \text{sign}[p(uv, z_v)] = -\text{sign}[P_G(\boldsymbol{z})] \},
\end{align*}
where $E_G$ are the external edges of $G$ and $p$ is the edge influence. 
The \emph{majority and minority edges of a vertex} $u$, denoted $E_u^\text{maj}$ and~$E_u^\text{min}$, are the edges of~$E_G^\text{maj}$ and~$E_G^\text{min}$ incident 
to vertex~$u$ respectively. In the case where~$P_G(\boldsymbol{z})=0$, 
the majority and minority edges are defined to be the null set, i.e. $E_G^\text{maj}(\boldsymbol{z}) = E_G^\text{min}(\boldsymbol{z}) = \emptyset$.
\end{definition}

\begin{definition}[Net Majority and Minority Influence -- Figure~\ref{fig:influence-example}]
The \emph{net majority and minority influences} are the sums of the majority and minority edge influences respectively, that is,
  \begin{align}
    P_G^\text{maj}(\boldsymbol{z}) &:= \sum_{uv \in E_G^{\text{maj}}(\boldsymbol{z})} p(uv, z_v), &P_G^\text{min}(\boldsymbol{z}) &:= \sum_{uv \in E_G^{\text{min}}(\boldsymbol{z})} p(uv, z_v),
  \end{align}
  where $E_G^{\text{maj}}$ and $E_G^{\text{min}}$ are the majority and minority edges of $G$ respectively and $p(e)$ is the edge influence of $e$ with respect to $G$.
  Similarly, the \emph{net majority and minority influences of a vertex} $u$, denoted $P^\text{maj}_{u}$ and $P^\text{min}_{u}$, is the sum 
  of edge influences corresponding to majority and minority edges incident to $u$ respectively, that is
  \begin{align}
    P_u^\text{maj}(\boldsymbol{z}) &:= \sum_{uv \in E_u^{\text{maj}}(\boldsymbol{z})} p(uv, z_v), &P_u^\text{min}(\boldsymbol{z}) &:= \sum_{uv \in E_u^{\text{min}}(\boldsymbol{z})} p(uv, z_v),
  \end{align}
  where $E_u^{\text{maj}}$ and $E_u^{\text{min}}$ are the majority and minority edges of $u$ respectively.
  Note that the net influence of a graph or vertex is equal to the sum of the net majority and minority influences, i.e. $P_G(\boldsymbol{z}) = P_G^\text{maj}(\boldsymbol{z}) + P_G^\text{min}(\boldsymbol{z})$ and $P_u(\boldsymbol{z}) = P_u^\text{maj}(\boldsymbol{z}) + P_u^\text{min}(\boldsymbol{z})$.
\end{definition}

\begin{definition}[Logical Device -- Figure~\ref{fig:logical-device}]\label{def:logical-device}
A \emph{logical device} is an ordered pair~$(G, E_G)$ where $E_G$ are the external edges of an induced subgraph~$G$ of Ising graph~$I$.
\end{definition}

Note that the edge and vertex sets of $I$ are $E(I) = E(G) \cup E_G \cup E(G_{\text{env}})$ and $V(I) = V(G) \cup V(G_{\text{env}})$ respectively, where $G_{\text{env}}$ is the environment 
of~$G$.

\begin{definition}[Comparable -- Figure~\ref{fig:comparable}]
Two logical devices on distinct Ising graphs are \textit{comparable} if they isomorphically have the same environment and edge influences. More precisely, 
let $(G, E_G)$ and $(G^\prime, E_{G^\prime})$ be logical devices of Ising graphs~$I$ and~$I^\prime$ such that their environments~$G_{\text{env}}$ 
and~$G^\prime_{\text{env}}$ are isomorphic with isomorphism~$\phi : G_{\text{env}} \rightarrow G^\prime_{\text{env}}$. The logical 
devices~$(G, E_G)$ and~$(G^\prime, E_{G^\prime})$ are \emph{comparable} if and only if for all vertices $v$ in the environment $G_\text{env}$, 
the edge sets~$\{u_iv_j\in E_G \;|\; v_j = v\}$ and~$\{u^\prime_i v^\prime_j \in E_{G^\prime} \;|\; v^\prime_j=\phi(v)\}$, including their weights, are isomorphic.
\end{definition}

\begin{figure}[h!]
  \centering
  \subcaptionbox{Logical device example 1\label{fig:logical-device}}[0.49\textwidth]{\includegraphics[width=0.45\textwidth]{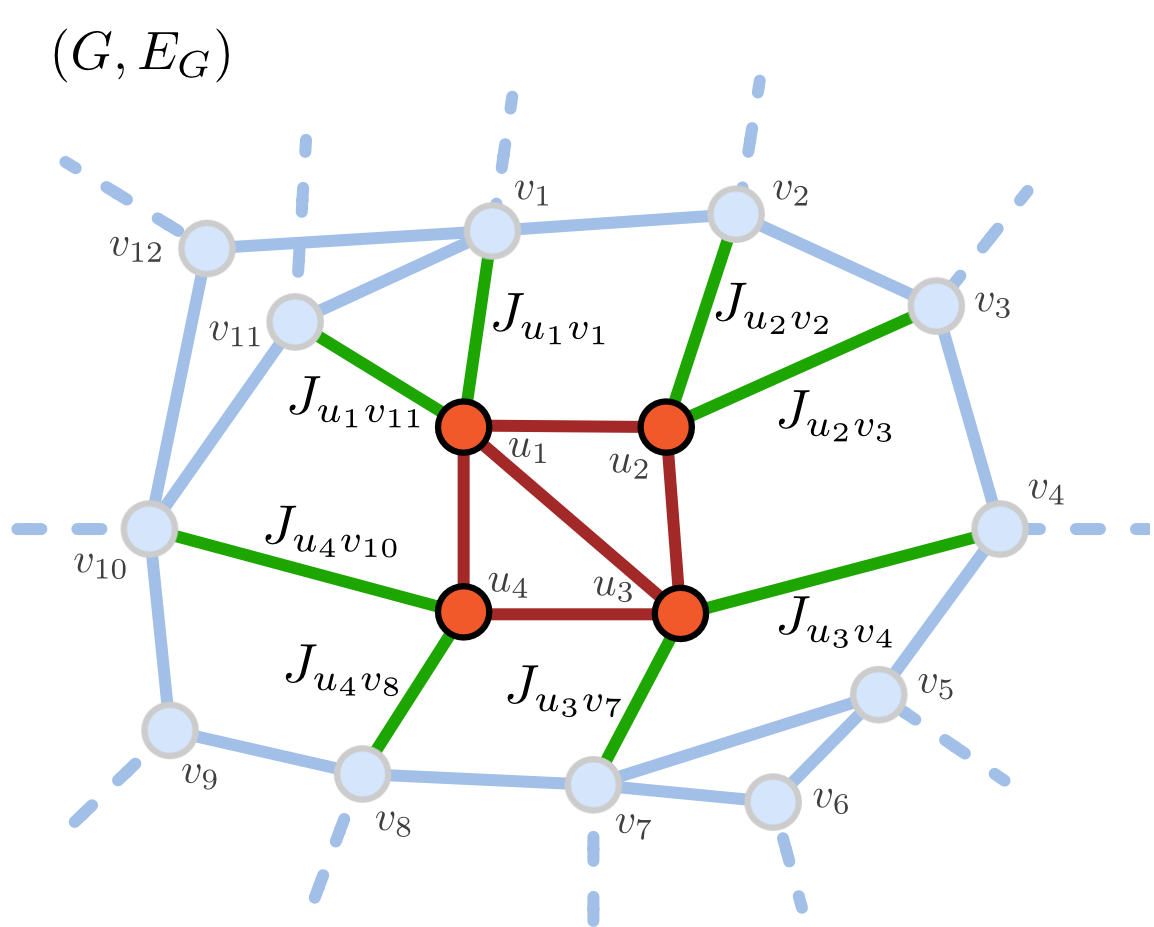}}
  \subcaptionbox{Logical device example 2}[0.49\textwidth]{\includegraphics[width=0.45\textwidth]{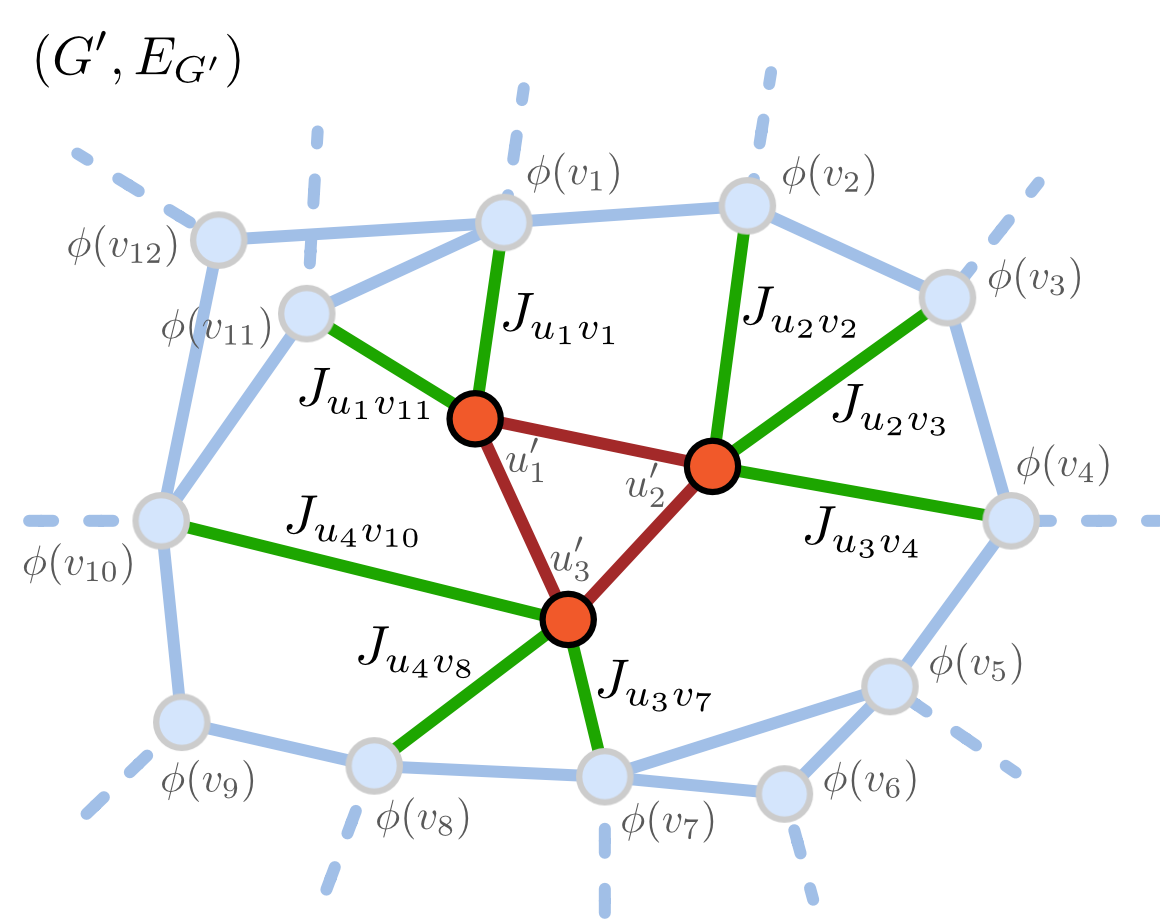}}
  \caption{Two comparable logical devices. These devices are comparable because~$G$ and~$G^\prime$ have isomorphic environments and the external edge 
  sets~$E_G$ and~$E_{G^\prime}$, including their weights, are the same under the isomorphism up to vertices within~$G$ and~$G^\prime$. 
  \textbf{(a)} An example of a logical device~$(G, E_G)$ consisting of the graph~$G$ of four vertices and its external edges~$E_G$. \textbf{(b)} An 
  example of a logical device~$(G^\prime, E_{G^\prime})$ consisting of the graph~$G^\prime$ of three vertices and its external edges~$E_{G^\prime}$ with 
  the same corresponding coupling strengths as $E_G$.\label{fig:comparable}}
\end{figure}

\begin{definition}[Logically Equivalent] \label{def:logically-equivalent}
  Let $(G, E_G)$ and $(G^\prime, E_{G^\prime})$ be comparable logical devices of Ising graphs~$I$ and~$I^\prime$. Let $u$ and $u^\prime$ be single vertices of $G$ and $G^\prime$ respectively. The 
  spin states of vertices within $I$ and $I^\prime$ are denoted as follows. Let~$Z_u$ and~$Z_{u^\prime}$ be spin states of $u$ and $u^\prime$, 
  let~$\mathbf{Z}_{\overline{u}}$ and~$\mathbf{Z}_{\overline{u}^\prime}$ be spin states of~$V(G) - \{u\}$ and~$V(G^\prime) - \{u^\prime\}$ and let~$\mathbf{z}$ and $\mathbf{z}^\prime$
  be spin states of~$~V(G_{\text{env}})$ and $V(G^\prime_{\text{env}})$ respectively.

  The two vertices $u$ and $u^\prime$ of logical devices~$(G, E_G)$ and~$(G^\prime, E_{G^\prime})$ respectively are \emph{logically equivalent} if and only if 
  \begin{equation}
    (Z_u^\text{gs}, \mathbf{z}^\text{gs}) = (Z_{u^{\prime}}^\text{gs}, \mathbf{z}^{\prime \text{gs}}),
  \end{equation}
  where 
  \begin{align}
    (Z_u^\text{gs}, \mathbf{Z}_{\overline{u}}^\text{gs}, \mathbf{z}^\text{gs}) &:= \underset{(Z_u, \mathbf{Z}_{\overline{u}}, \mathbf{z})}{\mathrm{argmin}}[H_I(Z_u, \mathbf{Z}_{\overline{u}}, \mathbf{z})], \text{ and}\\
    (Z_{u^\prime}^{\text{gs}}, \mathbf{Z}_{\overline{u}^\prime}^{\text{gs}}, \mathbf{z}^{\prime \text{gs}}) &:= \underset{(Z_{u^\prime}, \mathbf{Z}_{\overline{u}^\prime}, \mathbf{z}^\prime)}{\mathrm{argmin}}[H_{I^\prime} (Z_{u^\prime}, \mathbf{Z}_{\overline{u}^{\prime }}, \mathbf{z}^\prime)],
  \end{align}
  where $H_I$ and $H_{I^\prime}$ are Ising Hamiltonians of $I$ and~$I^\prime$ respectively.
\end{definition}

The following theorem describes a degree reduction transformation. The intuitive idea is that the spins in the logical devices are coupled strongly enough such that the spin states act as a single logical spin.

\begin{theorem}[Complete Degree Reduction -- Figure~\ref{fig:complete-degree-reduction-hamiltonian}]\label{theorem:complete-degree-reduction}
  Let $(G, E_G)$ be a logical device of a ZZ Ising graph (i.e. with only 2-local interaction terms) where the largest magnitude of 
  all coupling strengths of external edges is denoted by~$K_\text{ext}$. If~$G$ is a complete graph of order~$N$, with each edge coupling 
  strength greater than or equal to~$K_\text{ext}$ and with each vertex incident with up to~$t \leq N/2$ external edges, then all vertices in $(G, E_G)$ are 
  logically equivalent to the vertex of a comparable logical device $(G^\prime, E_{G^\prime})$ consisting of a single vertex.
\end{theorem}

\begin{proof}
Let $H_I$ and $H_{I^\prime}$ be ZZ Ising Hamiltonians of Ising graphs~$I$ and $I^\prime$ and let $G$ and $G^\prime$ be induced 
subgraphs of $I$ and $I^\prime$ respectively, as shown in Fig.~\ref{fig:complete-degree-reduction-hamiltonian}. For readability, we use a convention of labelling vertices of $G$ with $u$, the single vertex 
of $G^\prime$ with $u^\prime$, and 
vertices of the environments of $G$ and $G^\prime$ with $v$ and $v^\prime$ respectively, where subscripts are used to distinguish between them when needed. The Hamiltonians of $I$ and $I^\prime$ can be written as
  \begin{align}
    H_I(\boldsymbol{Z}_G, \boldsymbol{z}) &= H_G(\boldsymbol{Z}_G) + H_{E_G}(\boldsymbol{Z}_G, \boldsymbol{z}) + H_{G_{\text{env}}}(\boldsymbol{z}), \text{ and}\\
    H_{I^\prime} (Z_{u^\prime}, \boldsymbol{z}^\prime) &= H_{E_{G^\prime}}(Z_{u^\prime}, \boldsymbol{z}^\prime) + H_{G^\prime_\text{env}}(\boldsymbol{z}^\prime),
  \end{align}
  where~$H_G$,~$H_{E_G}$,~$H_{E_{G^\prime}}$,~$H_{G_\text{env}}$, and~$H_{G^\prime_\text{env}}$ are the~ZZ~Ising Hamiltonians of~$G$, the external edges~$E_G$ and~$E_{G^\prime}$, and 
  the environments $G_\text{env}$ and $G^\prime_\text{env}$ respectively,~$\boldsymbol{Z}_G$ and~$Z_{u^\prime}$ are spin states of the vertices of $G$ and the single 
  vertex $u^\prime$ of~$G^\prime$ respectively, and~$\boldsymbol{z}$ and~$\boldsymbol{z}^\prime$ are 
  the spin states of the environments $G_\text{env}$ and $G^\prime_\text{env}$ respectively. 
  
  Referring to Fig.~\ref{fig:complete-degree-reduction-hamiltonian}, we can expand~$H_I$ and $H_{I^\prime}$ as (Noting that $H_{G_\text{env}}(\boldsymbol{z}) = H_{G^\prime_\text{env}}(\boldsymbol{z})$ since the environments are isomorphic.)
  \begin{align}
    H_I(\boldsymbol{Z}_G, \boldsymbol{z}) &= -\underbrace{\sum_{u_iu_j\in E(G) } J_{u_iu_j} Z_{u_i} Z_{u_j}}_{\text{internal edges}} - \underbrace{\sum_{i=1}^{N} Z_{u_i} \sum_{v \in n_{\overline{G}}(u_i)} J_{u_i v}z_{v}}_{\text{external edges}} + \underbrace{H_{G_{\text{env}}}(\boldsymbol{z})}_{\text{environment}},\nonumber \\
    H_{I^\prime}(Z_{u^\prime}, \boldsymbol{z^\prime}) &= - \underbrace{ Z_{u^\prime} \sum_{v^\prime \in n(u^\prime)} J_{u^\prime v^\prime}z_{v^\prime}}_{\text{external edges}} + \underbrace{H_{G_{\text{env}}}(\boldsymbol{z}^\prime)}_{\text{environment}},\label{eq:degree-reduction-hamiltonians}
  \end{align}
  where~$N$ is the number of vertices of $G$,~$J_{ij}$ is the coupling strength between spins $i$ and $j$,~$Z_{u_i} \in \boldsymbol{Z}_G$, $z_{v} \in \boldsymbol{z}$, 
  $z_{v^\prime} \in \boldsymbol{z^\prime}$,~$n(u^\prime)$ are the neighbouring vertices of~$u^\prime$, and~$n_{\overline{G}}(u_i)$ are neighbouring 
  vertices of~$u_i$ that are not within $G$.

  By considering the edge influence of $G^\prime$ we can determine an expression for the form of its ground state Ising Hamiltonian. Let~$P_{G^\prime}(\boldsymbol{z}^\prime)$ be 
  the net influence of~$G^\prime$,~$|P_{G^\prime}^{\text{maj}}(\boldsymbol{z}^\prime)|$ be the magnitude of the 
  net majority influence and~$|P_{G^\prime}^{\text{min}}(\boldsymbol{z}^\prime)|$ be the magnitude of the net minority influence. In the case of~$P_{G^\prime}(\boldsymbol{z}^\prime)=0$, we 
  set~$|P_{G^\prime}^{\text{maj}}(\boldsymbol{z}^\prime)| = |P_{G^\prime}^{\text{min}}(\boldsymbol{z}^\prime)|$ and note that both spin states $Z_{u^\prime} \in \{-1, 1\}$ of $G^\prime$ are ground 
  states. In the case of~$P_{G^\prime}(\boldsymbol{z}^\prime)\neq 0$, the ground state of~$G^\prime$ will always assume the majority sign (example shown in 
  Fig.~\ref{fig:influence-example}). In both cases we can write the form of the ground state Ising Hamiltonian of $I^\prime$ as
  \begin{align}
    \min_{Z_{u^\prime}}[H_{I^\prime} (Z_{u^\prime}, \boldsymbol{z}^\prime)] &= \min_{Z_{u^\prime}}[-Z_{u^\prime} \sum \limits_{v^\prime \in n(u^\prime)} J_{u^\prime v^\prime} z_{v^\prime}] + H_{G_\text{env}}(\boldsymbol{z}^\prime)\\
    &= \min_{Z_{u^\prime}}[-Z_{u^\prime} P_{G^\prime}(\boldsymbol{z}^\prime)] + H_{G_\text{env}}(\boldsymbol{z}^\prime)\\
    &= \min_{Z_{u^\prime}}[-Z_{u^\prime} P_{G^\prime}^\text{min}(\boldsymbol{z}^\prime) - Z_{u^\prime} P_{G^\prime}^\text{maj}(\boldsymbol{z}^\prime)] + H_{G_\text{env}}(\boldsymbol{z}^\prime)\\
    &= |P_{G^\prime}^\text{min}(\boldsymbol{z}^\prime)| - |P_{G^\prime}^\text{maj}(\boldsymbol{z}^\prime)| + H_{G_\text{env}}(\boldsymbol{z}^\prime).
  \end{align}
  
  To compare this with the form of the ground state Ising Hamiltonian of $G$ we will set the spin states $\boldsymbol{Z}_G$ of~$G$ to 
  be equal to the majority sign (where we arbitrarily set the majority sign to $+1$ for $P_G(\boldsymbol{z})=0$) and then show that for any assignment of environmental spin states $\boldsymbol{z}$ minimises the Ising Hamiltonian~$H_I$. 
  We will show that this configuration for~$P_G(\boldsymbol{z}) \neq 0$, gives a unique minimum 
  of~$H_I(\boldsymbol{Z}_G, \boldsymbol{z})$, and for~$P_G(\boldsymbol{z})=0$, gives a non-unique minimum.
  The Hamiltonian~$H_I(\boldsymbol{Z}_G, \boldsymbol{z})$, after assigning each element of $\boldsymbol{Z}_G$ to the majority sign is
    \begin{equation} \label{eq:proof_E0}
      H_I^{(0)}(\boldsymbol{z}) := - \sum_{u_iu_j\in E(G) }J_{u_iu_j} + |P_{G}^\text{min}(\boldsymbol{z})| - |P_G^{\text{maj}}(\boldsymbol{z})| + H_{G_\text{env}}(\boldsymbol{z}),
    \end{equation}\noindent
    where $H_I^{(F)}$ represents the Ising Hamiltonian of $I$ with $F$ spin states of $G$ differing from the majority sign.
    To analyse this expression in more depth, it is useful to define the following. Let $P_{u_i}^{\text{maj}}(\boldsymbol{z})$ and $P_{u_i}^{\text{min}}(\boldsymbol{z})$ be the 
    net majority edges of $u_i$. Then, 
    \begin{equation} \label{eq:aplusb}
    |P_{u_i}^\text{maj}(\boldsymbol{z})| + |P_{u_i}^\text{min}(\boldsymbol{z})|  = \sum_{v \in n_{\overline{G}}(u_i)}|p(u_iv, z_v)| \leq t K_\text{ext}
    \end{equation}
    where each vertex of $G$ is incident with up to $t$ external edges such that $t\leq N/2$.
    By analysing how flipping spins away from the majority sign changes the value of $H_I$, we will show that $H_I^{(0)}$ is 
    smaller than or equal to any $F$ number of spin flips away from the majority sign and hence is a ground state. Without loss of 
    generality, we order the vertex labels of $G$ such that spin states corresponding to 
  vertices~$u_i = u_1, \ldots, u_F$ are flipped.
  The Ising Hamiltonian~$H_I(\boldsymbol{Z}_G, \boldsymbol{z})$ with~$F$ flipped spin states of $G$ from the majority sign can be expressed as

  \begin{equation}\label{eq:spin-flip-energy}
    H_I^{(F)}(\boldsymbol{z}) := H_I^{(0)}(\boldsymbol{z}) + \Delta_I^{(F)}(\boldsymbol{z}),
  \end{equation} 
  where $\Delta_I^{(F)}(\boldsymbol{z})$ is the energy change from $H_I^{(0)}(\boldsymbol{z})$ caused by the flipped states. We observe that (see Figure~\ref{fig:energy-change-example})
    \begin{align}
      \Delta_I^{(F)}(\boldsymbol{z}) &= 2\sum_{i=1}^F \sum_{j=F+1}^{N}J_{u_iu_j} + 2\sum_{i = 1}^{F}|P_{u_i}^\text{maj}(\boldsymbol{z})| - 2 \sum_{i=1}^{F} |P_{u_i}^\text{min}(\boldsymbol{z})| \\ \label{eq:proof_DeltaE}
                 &\geq 2 \left[ K_\text{ext} F(N-F) + \sum_{i = 1}^{F}|P_{u_i}^\text{maj}(\boldsymbol{z})| - \sum_{i=1}^{F} |P_{u_i}^\text{min}(\boldsymbol{z})| \right],
    \end{align}\noindent
  where the first term corresponds to the energy change due to internal edges and the second and third terms correspond to the 
  energy change due to the external edge influences, and~$K_\text{ext}$ is the largest magnitude of external edge coupling strengths and 
  by construction is less than or equal to all internal edge coupling strengths. 

  \begin{figure}
    \centering
    \subcaptionbox{All majority sign spin states\label{fig:energy-change-example-1}}[0.45\textwidth]{\includegraphics[width=0.30\textwidth]{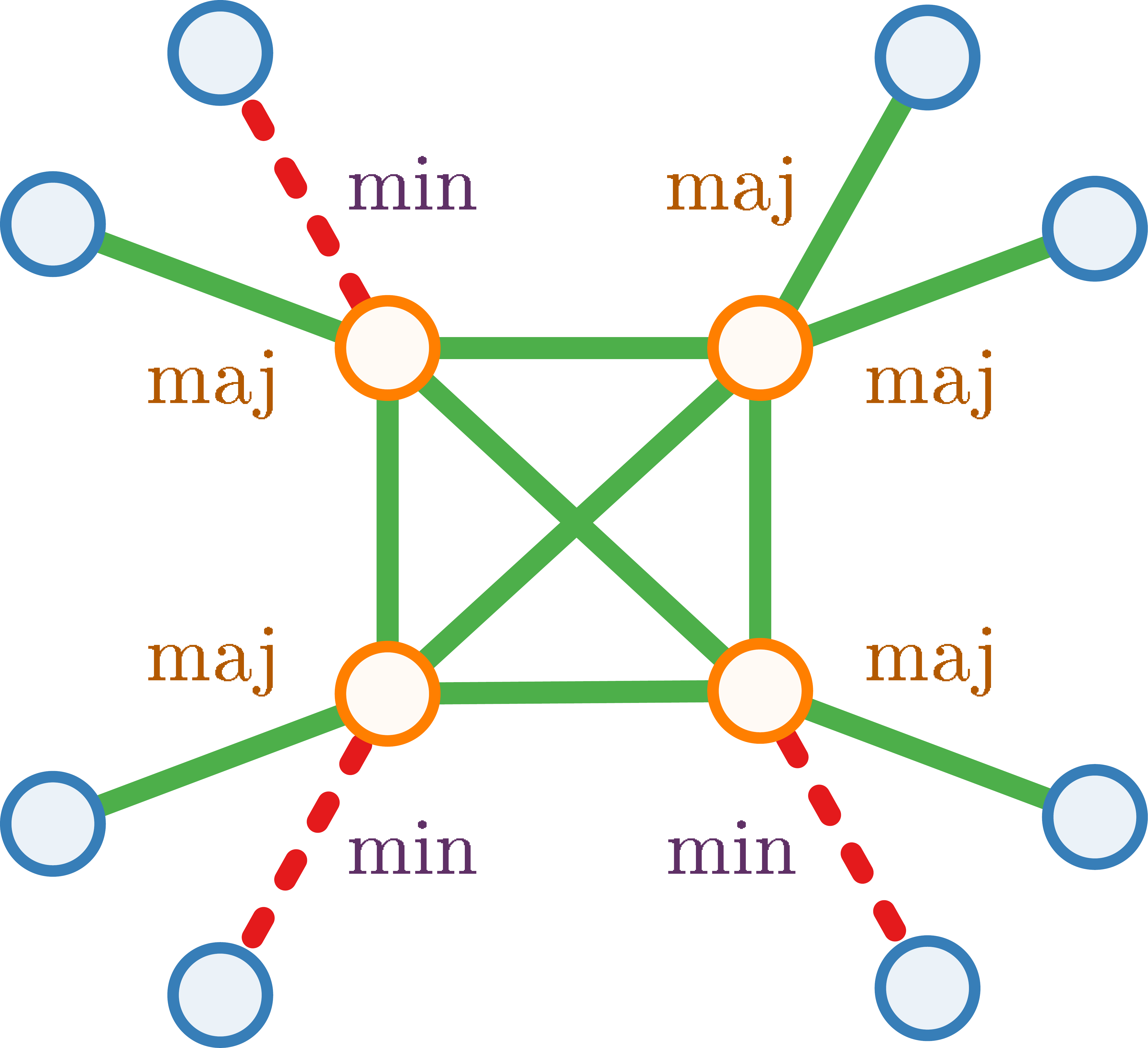}}
    \subcaptionbox{Two flipped spin states\label{fig:energy-change-example-2}}[0.45\textwidth]{\includegraphics[width=0.30\textwidth]{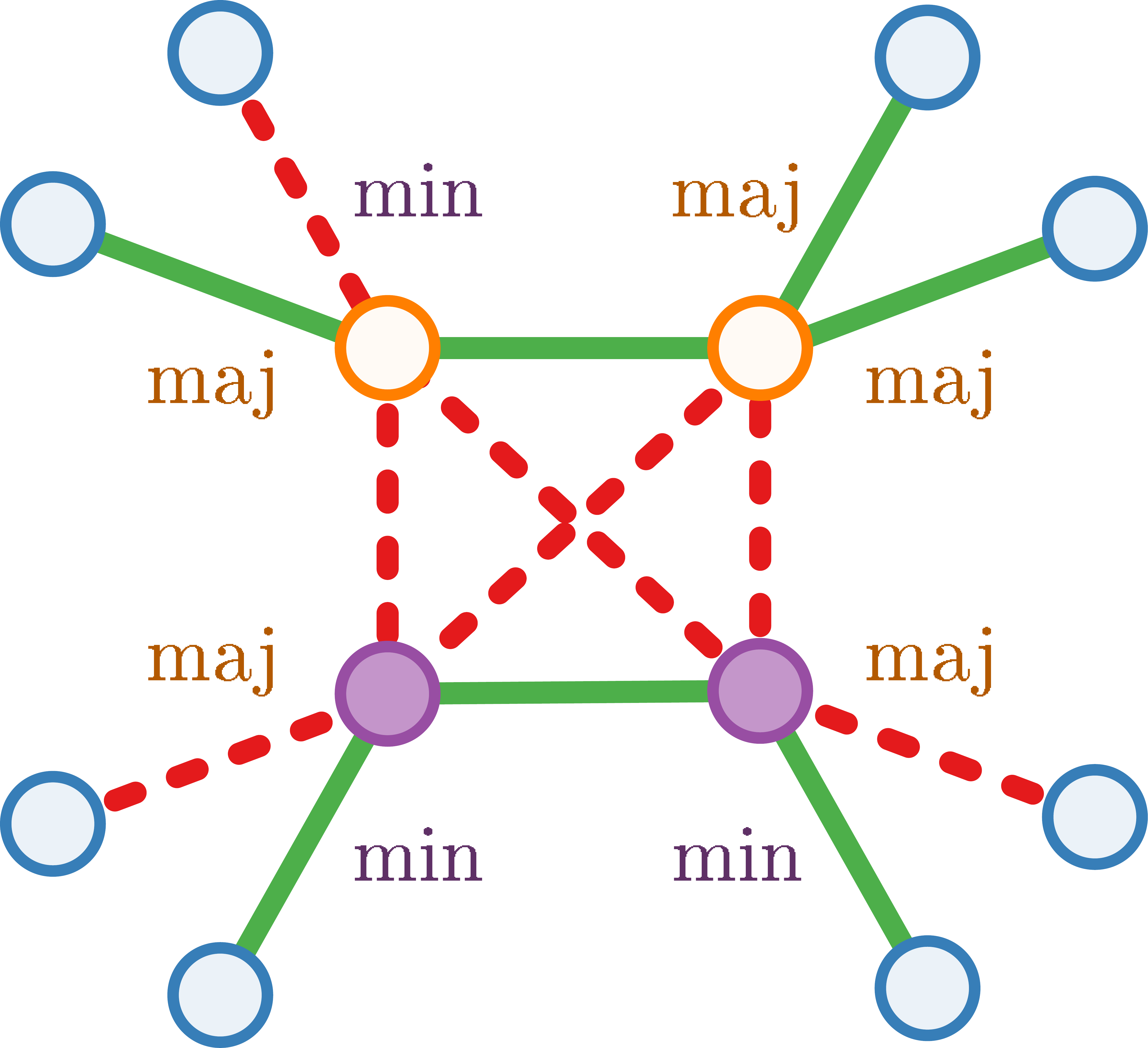}}
    \vspace{0.5 cm}
    \caption{An example of satisfied (green solid) and unsatisfied (red dashed) edges for two spin state configurations (a) and (b), orange (light) vertices within the complete graph indicate spin states in the majority sign while purple (dark) vertices indicate the minority sign. The edge 
    labels ``min" and ``maj" indicate whether the edge is a majority 
    or minority edge respectively. An edge~$uv$ contributes~$-|J_{uv}|$ energy when it is satisfied and~$|J_{uv}|$ when it is unsatisfied. Thus an energy difference of $2|J_{uv}|$ between the spin satisfying the edge and unsatisfying it.
    \textbf{(a)} All spin states in the complete graph are set to the majority sign of the edge influences. \textbf{(b)} Two spin states in the complete graph are flipped to the minority sign.}
    \label{fig:energy-change-example}\par\medskip
  \end{figure}\noindent
  To better understand the ground states of $H_I^{(F)}(\boldsymbol{z})$ for all $0 < F \leq N$, we will show that for $P_G(\boldsymbol{z}) \neq 0$,
  \begin{align}\label{eq:p-not-equal-to-zero}
    H_I^{(F)}(\boldsymbol{z}) > H_I^{(0)}(\boldsymbol{z}),
  \end{align}
  and for $P_G(\boldsymbol{z}) = 0$,
  \begin{align}
    H_I^{(F)}(\boldsymbol{z}) &> H_I^{(0)}(\boldsymbol{z}), \;\;\;\;\;\text{ for }F \neq N/2 \text{ and } F \neq N, \label{eq:p-equal-zero-1}\\
    H_I^{(F)}(\boldsymbol{z}) &\geq H_I^{(0)}(\boldsymbol{z}), \;\;\;\;\;\text{ for }F = N/2, \label{eq:p-equal-zero-2}\\
    H_I^{(F)}(\boldsymbol{z}) &= H_I^{(0)}(\boldsymbol{z}), \;\;\;\;\;\text{ for }F = N. \label{eq:p-equal-zero-3}
  \end{align}
  To show the above statements, the following cases of $F$ and $P_G(\boldsymbol{z})$ are considered separately:
  \begin{align*}
    &\text{\textbf{Case 1}: } 0 < F < N/2, & &\text{ for }P_G(\boldsymbol{z}) \neq 0 \text{ and }P_G(\boldsymbol{z}) = 0,\\
    &\text{\textbf{Case 2}: } N/2 \leq F \leq N, & &\text{ for }P_G(\boldsymbol{z}) \neq 0,\\
    &\text{\textbf{Case 3}: } N/2 < F < N, & &\text{ for }P_G(\boldsymbol{z}) = 0,\\
    &\text{\textbf{Case 4}: } F = N/2, & &\text{ for }P_G(\boldsymbol{z}) = 0, \\
    &\text{\textbf{Case 5}: } F = N & &\text{ for }P_G(\boldsymbol{z}) = 0\\
  \end{align*}
  
  \noindent\textbf{Case 1}: Assume~$0 < F < N/2$. Since $K_\text{ext} > 0$ we can write
  \begin{equation}
    0 < K_{\text{ext}}F(N/2 - F).
  \end{equation}
  Now since each vertex of $G$ is incident with up to $t$ external edges such that~$t\leq N/2$, it follows
  \begin{align}
    0 &< K_{\text{ext}}F(N-F) - K_{\text{ext}}FN/2 \\
    &\leq K_{\text{ext}}F(N-F) - tK_{\text{ext}}F.
  \end{align}
  By Eq~\ref{eq:aplusb} and Eq.~\ref{eq:proof_DeltaE} we can deduce
  \begin{equation}
    0 < K_{\text{ext}}F(N-F) - \sum_{i=1}^F |P_{u_i}^\text{min}(\boldsymbol{z})| < \Delta_I^{(F)}(\boldsymbol{z}).
  \end{equation}
  Therefore by~Eq~\ref{eq:spin-flip-energy}, we have shown~$H_I^{(F)}(\boldsymbol{z}) > H_I^{(0)}(\boldsymbol{z})$ for any value of~$P_G(\boldsymbol{z})$.\\
  \\
  \textbf{Case 2}: Assume that $N/2 \leq F \leq N$ and $P_G(\boldsymbol{z}) \neq 0$. Since $|P_G^\text{min}(\boldsymbol{z})| < |P_G^\text{maj}(\boldsymbol{z})|$, $0 < t \leq N/2$ and $0 < K_{\text{ext}}$, we can justify the following
  \begin{align}
    0 &< K_{\text{ext}}(F-t)(N-F) + |P_G^\text{maj}(\boldsymbol{z})| - |P_G^\text{min}(\boldsymbol{z})| \nonumber \\
    &= K_{\text{ext}}F(N-F) - tK_{\text{ext}}(N-F) + |P_G^\text{maj}(\boldsymbol{z})| - |P_G^\text{min}(\boldsymbol{z})|. \label{eq:case2-inequality}
  \end{align}
  From Eq~\ref{eq:aplusb}, we observe that
    \begin{align}
      \sum_{i = F+1}^{N}|P_{u_i}^\text{maj}(\boldsymbol{z})| = \sum_{i = 1}^{N - F}|P_{u_{F + i}}^\text{maj}(\boldsymbol{z})| \leq tK_{\text{ext}}(N - F),
    \end{align}
    \noindent which can be used, along with $\sum_{i=1}^{N}|P_{u_i}^\text{maj}(\boldsymbol{z})| = |P_G^\text{maj}(\boldsymbol{z})|$, to show 
    \begin{align}
      \sum_{i = 1}^{F}|P_{u_i}^\text{maj}(\boldsymbol{z})| - \sum_{i = 1}^{F}|P_{u_i}^\text{min}(\boldsymbol{z})| &\geq \sum_{i = 1}^{F}|P_{u_i}^\text{maj}(\boldsymbol{z})| - |P_G^\text{min}(\boldsymbol{z})| \nonumber \\
      &= \sum_{i=1}^{N}|P_{u_i}^\text{maj}(\boldsymbol{z})| - \sum_{i = F+1}^{N}|P_{u_i}^\text{maj}(\boldsymbol{z})| - |P_G^\text{min}(\boldsymbol{z})| \nonumber \\
      &\geq |P_G^\text{maj}(\boldsymbol{z})| - tK_{\text{ext}}(N-F) - |P_G^\text{min}(\boldsymbol{z})|.
    \end{align}\noindent
  Thus by Equations~\ref{eq:case2-inequality} and~\ref{eq:proof_DeltaE},
  \begin{equation}
    0 < K_{\text{ext}}F(N-F) + \sum_{i = 1}^{F}|P_{u_i}^\text{maj}(\boldsymbol{z})| - \sum_{i = 1}^{F}|P_{u_i}^\text{min}(\boldsymbol{z})| < \Delta_I^{(F)}(\boldsymbol{z}).
  \end{equation}
  Therefore by Eq.~\ref{eq:spin-flip-energy} we have shown $H_I^{(F)}(\boldsymbol{z}) > H_I^{(0)}(\boldsymbol{z})$ for $N/2 \leq F \leq N$ and $P_G(\boldsymbol{z}) \neq 0$.
  \newline{}\newline{}
  \textbf{Case 3}: Assume that $N/2 < F < N$ and $P_G(\boldsymbol{z}) = 0$. Since $0 < t \leq N/2$ and $0 < K_{\text{ext}}$, we can justify the following
  \begin{align}
    0 &< K_\text{ext}(F-t)(N-F) \nonumber\\
    &= K_\text{ext}F(N-F) - K_\text{ext}t(N-F). \label{eq:case3-statement}
  \end{align}
  Since $|P_G^\text{maj}(\boldsymbol{z})| = |P_G^\text{min}(\boldsymbol{z})|$ and $|P_{u_i}^\text{maj}(\boldsymbol{z})| - |P_{u_i}^\text{min}(\boldsymbol{z})| \leq t K_\text{ext}$ we have,
  \begin{align}
    \sum_{i=F+1}^N |P_{u_i}^\text{maj}(\boldsymbol{z})| - \sum_{i=F+1}^N |P_{u_i}^\text{min}(\boldsymbol{z})| &= \sum_{i=1}^F |P_{u_i}^\text{min}(\boldsymbol{z})| - \sum_{i=1}^F |P_{u_i}^\text{maj}(\boldsymbol{z})| \\
    &\leq K_\text{ext} t (N-F),
  \end{align}
  hence
  \begin{equation}
    - K_\text{ext}t(N-F) \leq \sum_{i=1}^F |P_{u_i}^\text{maj}(\boldsymbol{z})| - \sum_{i=1}^F |P_{u_i}^\text{min}(\boldsymbol{z})|.
  \end{equation}
  Thus by Equations~\ref{eq:case3-statement} and~\ref{eq:proof_DeltaE}, 
  \begin{equation}
    0 < K_\text{ext}F(N-F) + \sum_{i=1}^F |P_{u_i}^\text{maj}(\boldsymbol{z})| - \sum_{i=1}^F |P_{u_i}^\text{min}(\boldsymbol{z})| < \Delta_I^{(F)}(\boldsymbol{z}).
  \end{equation}
  Therefore by Eq.~\ref{eq:spin-flip-energy} we have shown $H_I^{(F)}(\boldsymbol{z}) > H_I^{(0)}(\boldsymbol{z})$ for $N/2 < F < N$ and $P_G(\boldsymbol{z}) = 0$.
  \newline{}\newline{}
  \textbf{Case 4}: Assume that $F = N/2$ and $P_G(\boldsymbol{z}) = 0$. Since $|P_{u_i}^\text{min}(\boldsymbol{z})| - |P_{u_i}^\text{maj}(\boldsymbol{z})| \leq tK_\text{ext}$ and $t \leq N/2$, we have
  \begin{align}
    \sum_{i=1}^{F}|P_{u_i}^\text{min}(\boldsymbol{z})| - \sum_{i=1}^{F}|P_{u_i}^\text{maj}(\boldsymbol{z})| &\leq t K_\text{ext} F\\
    &\leq (N/2)K_\text{ext}F \\
    & = K_\text{ext}F (N-F).
  \end{align}
  Thus by Equation~\ref{eq:proof_DeltaE},
  \begin{align}
    0 &\leq K_\text{ext} F (N-F) + \sum_{i=1}^{F}|P_{u_i}^\text{maj}(\boldsymbol{z})| - \sum_{i=1}^{F}|P_{u_i}^\text{min}(\boldsymbol{z})| \leq \Delta_I^{(F)}(\boldsymbol{z}).
  \end{align}
  Therefore by Eq.~\ref{eq:spin-flip-energy} we have shown $H_I^{(F)}(\boldsymbol{z}) \geq H_I^{(0)}(\boldsymbol{z})$ for $F = N/2$ and $P_G(\boldsymbol{z}) = 0$.
  \newline{}\newline{}
  \textbf{Case 5}: Assume that $F = N$ and $P_G(\boldsymbol{z}) = 0$. Since $|P_G^\text{maj}(\boldsymbol{z})| = |P_G^\text{min}(\boldsymbol{z})|$,
  \begin{align}
    \Delta_I^{(F)}(\boldsymbol{z}) &= K_\text{ext} F (N-F) + \sum_{i=1}^{F}|P_{u_i}^\text{maj}(\boldsymbol{z})| - \sum_{i=1}^{F}|P_{u_i}^\text{min}(\boldsymbol{z})| \\
    &= |P_G^\text{maj}(\boldsymbol{z})| - |P_G^\text{min}(\boldsymbol{z})| \\
    &= 0.
  \end{align}
  Therefore by Eq.~\ref{eq:spin-flip-energy} we have shown $H_I^{(F)}(\boldsymbol{z}) = H_I^{(0)}(\boldsymbol{z})$ for $F = N$ and $P_G(\boldsymbol{z}) = 0$.
  
  We have successfully shown Eq.~\ref{eq:p-not-equal-to-zero} for $P_G(\boldsymbol{z}) \neq 0$ and Eqs.~\ref{eq:p-equal-zero-1}, \ref{eq:p-equal-zero-2} and \ref{eq:p-equal-zero-3} for $P_G(\boldsymbol{z}) = 0$. Hence $F=0$ corresponds to a ground state.
  Next we use the above findings to show that each of the vertices of $G$ are logically equivalent to the vertex of $G^\prime$ according to Definition~\ref{def:logically-equivalent}, that is, all of the spin states in ground states are the same. To do this we show the following two conditions:
  \begin{align*}
    \textbf{Condition 1: } &\text{For all environmental spin states } \boldsymbol{z}=\boldsymbol{z^\prime} \text{, the value of any} \\
    & \text{spin state } Z_{u_i} \text{ in } G \text{ that }\text{minimises } H_I \text{ with respect to } \boldsymbol{Z}_G  \\
    &\text{is the same as the value of the single spin state } Z_{u^\prime}\text{ in }G^\prime \\
    &\text{that minimises }H_{I^\prime} \text{ with respect to } Z_{u^\prime}.\\
    \textbf{Condition 2: } &\text{The environmental spin states }\boldsymbol{z}^\text{gs} \text{ and } \boldsymbol{z}^{\prime \text{gs}} \text{ are equal after} \\
    &\text{minimising }H_I \text{ and } H_{I^\prime}.
  \end{align*}
  \newline{}\newline{}
  \noindent \textbf{Condition 1}: Assume $\boldsymbol{z} = \boldsymbol{z}^\prime$. Note that this implies $P_G(\boldsymbol{z}) = P_{G^\prime}(\boldsymbol{z}^\prime)$ since~$(G, E_G)$ 
  and~$(G^\prime, E_{G^\prime})$ are comparable logical devices. For~$P_G(\boldsymbol{z}) \neq 0$, 
  minimising~$H_I$ and~$H_{I^\prime}$ requires all spin 
  states of $G$ and $G^\prime$ to be equal to the majority sign, by Equation~\ref{eq:p-not-equal-to-zero}. In the case 
  that~$P_G(\boldsymbol{z})=0$, both~$H_I$ and~$H_{I^\prime}$ are minimised for both spin state assignments of vertices
  in~$G$ and~$G^\prime$ respectively by Eqs.~\ref{eq:p-equal-zero-1}, \ref{eq:p-equal-zero-2} and \ref{eq:p-equal-zero-3}, hence for all~$Z_{u_i} \in \boldsymbol{Z}_G$, 
  $Z_{u_i}=Z_{u^\prime} = \{+1, -1\}$. Thus if~$Z_{u_i}^\text{gs}$ is the assignment of~$Z_{u_i}$ minimising~$H_I$, and~$Z_{u^\prime}^\text{gs}$ is the assignment of~$Z_{u^\prime}$ minimising~$H_{I^\prime}$, then~$Z_{u_i}^\text{gs} = Z_{u^\prime}^\text{gs}$ for any value of~$P_G(\boldsymbol{z})$, satisfying condition~1.
  \newline{}\newline{}
  \noindent \textbf{Condition 2}: First, we observe that
  \begin{align}
    \min_{\boldsymbol{Z}_G}[H_I(\boldsymbol{Z}_G, \boldsymbol{z})] &= H_I^{(0)}(\boldsymbol{z}) \\
                                                                     &= |P_G^\text{min}(\boldsymbol{z})| - |P_G^\text{maj}(\boldsymbol{z})| - \sum_{u_iu_j \in E(G)}J_{u_iu_j} + H_{G_\text{env}}(\boldsymbol{z}) \\
                                                                     &= \min_{Z_{u^\prime}}[H_{I^\prime}(Z_{u^\prime}, \boldsymbol{z})] - \sum_{u_iu_j \in E(G)}J_{u_iu_j}\\
                                                                     &=  \min_{Z_{u^\prime}}[H_{I^\prime}(Z_{u^\prime}, \boldsymbol{z})] + c,
  \end{align}
  where~$c$ is constant. So that
  \begin{align}
    \boldsymbol{z}^\text{gs} &= \underset{\boldsymbol{z}}{\mathrm{argmin}} \left[ \min_{\boldsymbol{Z}_G}[H_I(\boldsymbol{Z}_G, \boldsymbol{z})]\right]\\
                              &= \underset{\boldsymbol{z}}{\mathrm{argmin}} \left[ \min_{Z_{u^\prime}}[H_{I^\prime}(Z_{u^\prime}, \boldsymbol{z})] + c\right]\\
                              &= \boldsymbol{z}^{\prime \text{gs}}.
  \end{align}
  Thus $\boldsymbol{z}^\text{gs} = \boldsymbol{z}^{\prime \text{gs}}$ after minimising~$H_I$ and~$H_{I^\prime}$, satisfying condition~2.
  
  These two conditions can be summarised as follows. Minimising the two Hamiltonians $H_I$ and $H_{I^\prime}$ with respect to all spin states of $I$ and $I^\prime$ respectively implies that~$\boldsymbol{z}^\text{gs} = \boldsymbol{z}^{\prime \text{gs}}$ by condition~2. This implies that for all $Z_{u_i}^\text{gs} \in \boldsymbol{Z}^\text{gs}_G$, ~$Z_{u_i}^\text{gs} = Z_{u^\prime}^\text{gs}$ by condition~1.
  Hence $(Z_{u_i}^\text{gs}, \boldsymbol{z}^{\text{gs}}) = (Z_{u^\prime}^{\text{gs}}, \boldsymbol{z}^{\prime \text{gs}})$ for 
  all $Z_{u_i}^\text{gs} \in \boldsymbol{Z}^\text{gs}_G$ and therefore all vertices of~$(G, E_G)$ are logically equivalent to the 
  single vertex of~$(G^\prime, E_{G^\prime})$ according to Definition~\ref{def:logically-equivalent}.
\end{proof}

For an Ising graph, any vertex of high degree can be replaced with a graph of vertices of lower degrees such that a ground state of the original graph is determined by a ground state of the transformed graph. This can be recursively applied to reduce a vertex of any finite degree to a graph consisting of vertices with at most degree five. To further reduce the maximum degree of a reduced graph to three, two different degree reduction transformations can be used, shown in Figure~\ref{fig:degree-reduction}.

\newpage
\section{Complexity of complete degree reduction} \label{sec:degree-reduction-complexity}
\setcounter{figure}{0} 
\setcounter{definition}{0} 
\setcounter{lemma}{0}
\setcounter{theorem}{0}
\setcounter{corollary}{0}
\begin{figure}[h!]
  \centering
  \includegraphics[scale=1.1]{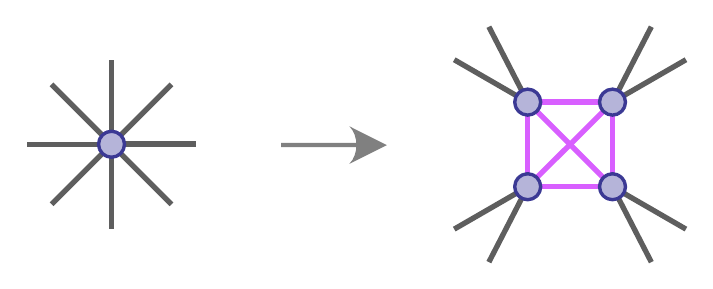}
  \caption{Complete degree reduction example (degree $8$ $\rightarrow$ degree $5$). The initial graph has $N_0=1$ vertex incident to $D_0=8$ edges and the transformed graph has $N_1=4$ vertices which are each incident to $D_1=5$ edges. All internal edges of the transformed graph have coupling strength greater than or equal to the largest coupling strength magnitude of external edges.}
  \label{fig:degree-reduction-example}
\end{figure}\noindent
\par\medskip
We will now derive the scaling relation for the reduction of a single vertex of arbitrary degree to a graph of vertices with at most degree three. This relation expresses how the number of vertices after degree reduction scales with the initial degree of the vertex. We assume that the following degree reduction procedure is applied. Given a vertex with degree~$D_0$, apply the complete degree reduction to all vertices recursively until the maximum degree in the graph is at most five. Then, use the degree~$5 \rightarrow 3$ reduction to reduce all vertices to a maximum degree of three. 

The objective in the following derivations is to find relatively tight upper bounds for how the degree reduction transformations scale. The chosen upper bounds in these calculations are not the tightest possible, because determining the tightest upper bounds further complicates the calculations and is not necessary for achieving a working approximation for the order of scaling.

Let~$n$ be the number of degree-reduction steps applied to the graph,~$D_n$ be the maximum possible degree of each vertex after~$n$ reduction steps, and~$N_n$ be the maximum possible number of vertices in each complete graph of the~$n^\text{th}$ reduction (see Fig.~\ref{fig:degree-reduction-example} for an example). For each complete degree reduction step, assume that the degree of all vertices after reducing is at least five, hence~$D_n \geq 5$ and~$N_n \geq 4$ for~$n\geq 1$.

When applying a complete degree reduction step, the degree of the vertex before reducing it is equal to the number of external edges of the subgraph replacing it, giving the relation
\begin{equation}\label{eq:recursive-relation-1}
  \left \lfloor{\frac{N_{n+1}}{2}}\right \rfloor N_{n+1} = D_n,
\end{equation}
noting that the first factor is the maximum number of external edges incident to a vertex according to complete degree reduction in Theorem~\ref{theorem:complete-degree-reduction}. From examining the structure of the complete graph after the $n^\text{th}$ reduction, $D_n$ can also be expressed as
\begin{equation}\label{eq:recursive-relation-2}
  D_n = \left \lfloor{\frac{N_n}{2}}\right \rfloor  + N_n - 1,
\end{equation}
where the first term corresponds to external edges and the second and third correspond to the other vertices of the complete graph.
Equating Eq.~\ref{eq:recursive-relation-1} and Eq.~\ref{eq:recursive-relation-2} while letting~$\left \lfloor{\frac{N_n}{2}}\right \rfloor = \frac{N_n}{2} - \alpha$ and~$\left \lfloor{\frac{N_{n+1}}{2}}\right \rfloor = \frac{N_{n+1}}{2} - \beta$ for some~$\alpha, \beta \in [ 0,1 )$ gives
\begin{equation}\label{eq:recursive-relation-2-2}
  N_{n+1} = \beta + \sqrt{\beta^2 + 3 N_n - 2\alpha - 2} \leq \sqrt{5 N_n},
\end{equation}
for $N_n \geq 4$. Let~$T^k_d$ be the total number of vertices resulting from~$k$ degree reduction steps such that after reduction the vertices each have a maximum possible degree of $d$. By using Eq.~\ref{eq:recursive-relation-2-2} recursively, noting the geometric series in the exponents, we get through simplification
\begin{equation}\label{eq:Tk}
  T^k_d = N_1 N_2 \ldots N_k \leq 5^{k} \left( N_1/5\right)^{2(1-2^{-k})}.
\end{equation}

Although $T_d^k$ is an expression for the total number of vertices in terms of $k$, the goal is to find it in terms of $D_0$, the degree of the initial vertex. This can be achieved by expressing $k$ in terms of $D_n$ by deriving a recursive relation for~$D_n$. Solving Eq.~\ref{eq:recursive-relation-1} for~$N_{n+1}$ gives,

\begin{equation}\label{eq:recursive-relation-3}
  N_{n+1} = \beta + \sqrt{ \beta^2 + 2 D_n} \leq \sqrt{ 5 D_n},
\end{equation}
for $D_n \geq 5$. By writing Eq.~\ref{eq:recursive-relation-2} with respect to $D_{n+1}$, we get

\begin{equation}\nonumber
  D_{n+1} = \frac{3N_{n+1}}{2} - \beta - 1,
\end{equation}
and substituting Eq.~\ref{eq:recursive-relation-3} in gives
\begin{equation}\nonumber
  D_{n+1} = \frac{3(\beta + \sqrt{ \beta^2 + 2 D_n})}{2} - \beta - 1 \leq \frac{3}{2}\sqrt{3D_n},
\end{equation}
for $D_n \geq 5$. So again, using recursion and the geometric series in the exponent, we get
  \begin{equation}\nonumber
    D_k \leq \left( \frac{27}{4}\right)^{1-2^{-k}} D_0^{2^{-k}}.
  \end{equation}\noindent
We use the complete degree reduction until every vertex has degree at most five and then apply the more optimized $5\rightarrow 3$ degree reduction which introduces a factor of nine to the total number of vertices.
So, letting~$D_k = 5$ in the above equation and solving for~$k$, we get
\begin{equation}\nonumber
  k \leq \log_2 \left(\frac{\ln(4D_0/27)}{\ln(20/27)} \right).
\end{equation}
After substituting this into Eq.~\ref{eq:Tk}, we can apply bounding arguments and Eq.~\ref{eq:recursive-relation-3} to reduce it to $T_5[D_0] \leq A D_0 \ln(D_0)^{\log_2 5}$ for some constant $A$ and where $k$ has been suppressed due to its dependence on $D_0$. Finally, using the~$5 \rightarrow 3$ degree reduction, the total number of vertices needed to degree reduce the vertices to degree three is
\begin{equation}\nonumber
  T_3[D_0] = 9 T_5[D_0] \leq A^\prime D_0 \ln(D_0)^{\log_2 5},
\end{equation}
for some constant $A^\prime$. Thus the number of vertices produced by the degree reduction procedure with respect to the degree of the initial vertex is of order
\begin{equation} \label{eq:reduction-complexity}
  \mathcal{O}\left( D_0\log(D_0)^{\log_2 5}\right).
\end{equation}

\end{document}